\documentclass{article}
\usepackage{latexsym}
\usepackage{amsmath, amssymb, amsthm, bm}
\usepackage{hyperref}
\usepackage{graphicx}
\usepackage{multirow}
\usepackage[authoryear]{natbib}
\usepackage{subfigure}
\usepackage{titling}

\usepackage{tikz}
\usetikzlibrary{cd}
\usepackage{footnote}

\newtheorem{assump}{Assumption}
\newtheorem{thm}{Theorem}

\newtheorem{lem}{Lemma}
\newtheorem{prop}{Proposition}

\DeclareMathOperator*{\argmin}{arg\,min}
\DeclareMathOperator{\tr}{tr}

\hypersetup{
	colorlinks=false,
    citecolor=black,
    hidelinks
}

\begin{document}
\setcounter{page}{0}
\title{\bf Individualized Group Learning}
\author{Chencheng Cai, Rong Chen and Min-ge Xie\thanks{
Chencheng Cai is a Ph.D student at Department of Statistics, Rutgers University, Piscataway, NJ 08854. E-mail: chencheng.cai@rutgers.edu.
Rong Chen is Professor at Department of
    Statistics, Rutgers University, Piscataway, NJ 08854. E-mail:
    rongchen@stat.rutgers.edu.
Minge Xie
    is Professor at Department of Statistics, Rutgers
    University, Piscataway, NJ 08854. E-mail:
    mxie@stat.rutgers.edu.
Rong Chen is the
    corresponding author. Chen's research is supported
in part by National Science Foundation
grants DMS-1503409, DMS-1737857 and IIS-1741390. Xie's research is supported in part by National Science Foundation grants DMS-1513483, DMS-1737857 and DMS-1812048.
}
\\ \vspace{0.2cm} {Rutgers University}}
\date{}
\maketitle

\newcommand{\addabstractandkeywords}{
\begin{abstract}
Many massive data are assembled through collections of information of a large number of individuals in a population. The analysis of such data, especially in the aspect of individualized inferences and solutions, has the potential to create significant value for practical applications. 
Traditionally, inference for an individual in the data set is either solely relying on the information of the individual or from summarizing the information about the whole population. However, with the availability of big data, we have the opportunity, as well as a unique challenge, to make a more effective individualized inference that takes into consideration of both the population information and the individual discrepancy. To deal with the possible heterogeneity within the population while providing effective and credible inferences for individuals in a data set, this article develops a new approach called the individualized group learning (iGroup). The iGroup approach uses local nonparametric techniques to generate an individualized group by pooling other entities in the population which share similar characteristics with the target individual, even when individual estimates are biased due to limited number of observations.  Three general cases of iGroup are discussed, and their asymptotic performances are investigated. Both theoretical results and empirical simulations reveal that, by applying iGroup, the performance of statistical inference on the individual level are ensured and can be substantially improved from inference based on either solely individual information or entire population information. 
The method has a broad range of applications. Two examples in financial statistics and maritime anomaly detection are presented. 
\end{abstract}

\noindent {\bf Keywords:} Similarity Measures, Clustering, Fusion Learning, Individualized Inference, Kernel Smoothing, Nonparametric, Bayesian Inference}
\addabstractandkeywords

%%%%%%%% Begin Anonymous Page
%\newpage
%\setcounter{page}{0}
%\author{}
%\makeatletter
%\def\@thanks{}
%\makeatother
%\maketitle
%\addabstractandkeywords
%%%%%%%% END Anonymous Page

\newpage
\section{Introduction}
With the massive data readily available in the digital and information era, advanced statistical learning methodologies for analysis of big data are in high demand. Traditional statistical methods are often used to discover the general rule of the population. 
However, in many applications we are also interested in an individual entity for personalized solutions or products. For instance, in precision medicine, each patient has his/her own traits. Therefore, it is crucial and beneficial to make individualized treatments and prescribe personalized medicine \citep{liu2016there,qian2011performance,zhao2012estimating,yang2012classification,collins2015new,wang2007statistics}. In business, the so-called 'Market of One' strategy that makes a customer feel that he or she is exclusive or preferred by the firm, becomes popular for companies to design personalized products. 
Indeed, individualized learning and inference matters in many applications.

\par
Since no two patients or two customers are exactly the same, heterogeneity often exists in a population. It poses a challenge to combine the data from different individuals, especially for making improved inferences in individualized learning. 
%One of the key challenge for combining data from different individuals is the existence of potential heterogeneity in population as no two patients are exactly the same in the personalized medical case. 
A class of conventional methods is to cluster/group individual entities into subgroups and, assuming homogeneity within each subgroup, then use the data in the same subgroup for statistical analysis 
\citep{jain1999data,xu2005survey,agrawal1998automatic, binder1978bayesian,ng1994cient, gan2007data, liao2005clustering, jain2010data}.
%\citep[][among others]{su2009subgroup, loh2016identification,lipkovich2017tutorial,wang2007statistics,jain1999data,xu2005survey,aggarwal2002redefining,agrawal1998automatic,alpert1994multi,ichino1994generalized,milligan1988study,binder1978bayesian,ng1994cient,gan2007data,liao2005clustering,koperski1995discovery,hartigan1985statistical,jain2010data,ernst2005clustering}. 
The clustering and grouping in the conventional methods are typically performed in \textit{a priori}. Such approaches have several disadvantages. Firstly, the constitution of subgroups often depends on a predetermined total number of subgroups, which is a parameter that is either difficult or not reliable to choose in practice. Secondly, since analytic outcomes and inference (e.g. estimated parameters and testing) are the same for all individuals in the same subgroup, such a procedure potentially diminishes hidden local structures. More importantly, in many cases, there may not be clearly-cut and well-divided subgroups in the population. In these situations, the conventional subgroup analysis may impose an artificial grouping structure to the population, which can potentially lead to large biases and invalid inference for many individuals. Another class of conventional methods is to assume mixture models, including classical hierarchical models and Bayesian nonparametric models \citep{duda1973pattern, lindsay1995mixture,figueiredo2000unsupervised, ferguson1973bayesian, antoniak1974mixtures, lo1984class,teh2005sharing}. Similar to the clustering method, the mixture models assume that the population contains several homogeneous subpopulations, but unlike clustering, there is no clear boundary between the subpopulations. 
However, inference on each individual is not the focus of such a procedure. It is often done as an afterthought, by estimating the mixture likelihood.
%In such a model, the parameters of the subpopulations are assumed to be latent and follow some distributions indexed by unknown hyper-parameters. Estimation of these hyper-parameters often involves the whole dataset and more computational effort. 
Furthermore, a mixture model may not be able to explain the population heterogeneity when the assumed latent structure is invalid. In addition, when given an observation, it is usually difficult to tell which subpopulation it belongs to.

\par
In this article, we propose a new method called individualized group learning, abbreviated as {\it iGroup}. Instead of grouping at the population level, the iGroup approach focuses on each individual and forms an individualized group for the target individual, by locating individuals that share similar characteristics of the target. It sidesteps aforementioned difficulties by forming an iGroup specifically for the target individual while ignoring other entities that have little in common with the target. Figure \ref{fig:group} demonstrates the difference between group identifications in a two-dimensional feature space. The left panel shows the result from a k-means clustering method with three groups. Each point is assigned with one cluster label. Data points having the same label are assumed to follow an identical statistical model, even though a large amount of heterogeneity may still exist among the individuals in the same group. The right panel demonstrates the individualized groups for two selected points (bold). Instead of assuming disjoint cluster regions, the individualized group, whose boundary is shown as a solid line, is specific and unique for each individual. Therefore, the laws for two individuals are generally different as their identified individualized groups are different. iGroup corresponds essentially to a local nonparametric approach.
\begin{figure}[!hptb]
\centering
\includegraphics[width=\textwidth]{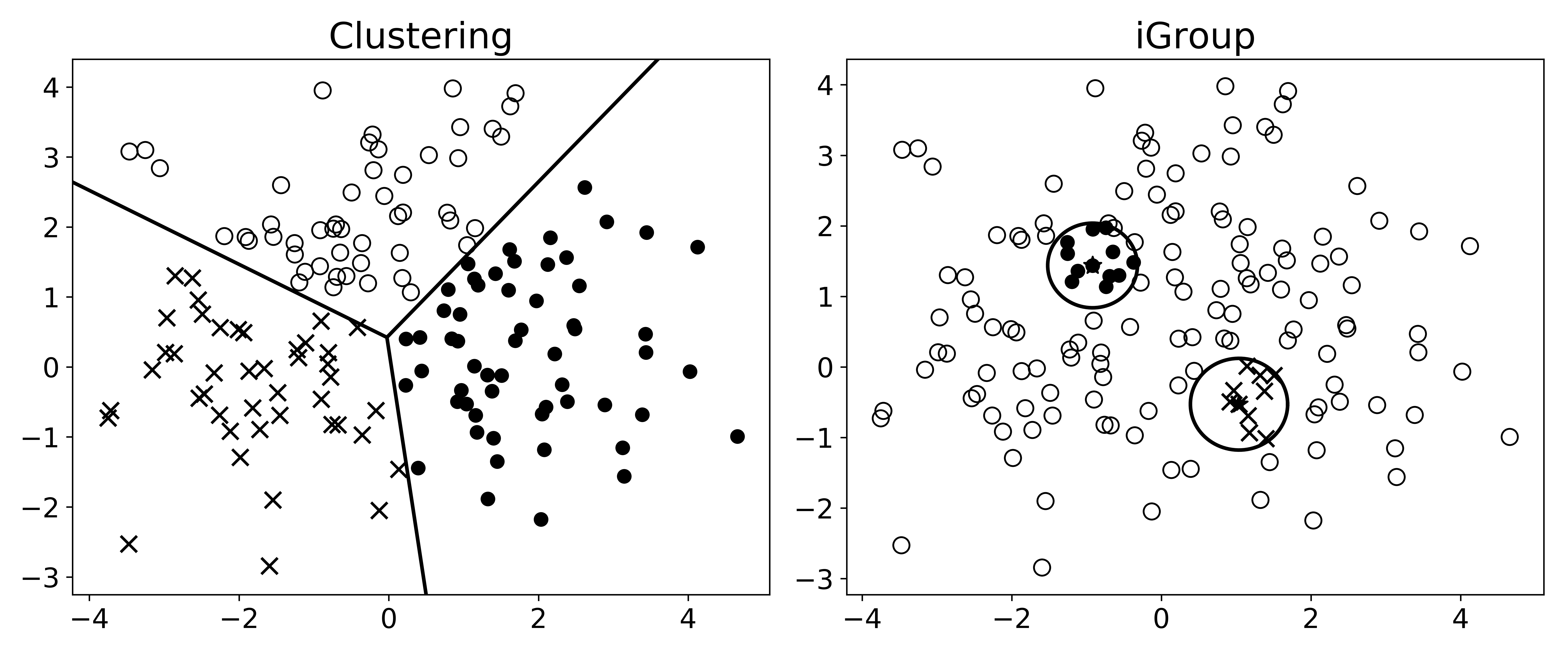}
\caption{(Left) Convention clustering method divides the population into several predetermined number of groups. (Right) iGroup method find the individualized group for any given target individual.}\label{fig:group}
\end{figure}

\par
In this paper, two sets of information are utilized in our proposed framework to define similarity and to form groups. 
One is individual level estimator $\hat\theta_k$, which is a direct estimation of $\theta_k$, the parameter of interest, for each individual % $k=1,2,\dots, n$
$k \in \{0, 1, \ldots, K\}$ in a parametric model with observation $\bm x_k$, without any grouping. The other is exogenous information $\bm z_k$, which is observed outside of the parametric model but can reveal similarity between the parameters. Both $\hat\theta_k$ and $\bm z_k$ can provide useful information in identifying groups so that closeness in the space of $(\hat\theta_k, \bm z_k)$ implies closeness in the space of $\theta_k$. Depending on the feasibility and availability of the two information sets, iGroup can be constructed based on three different information sets: $\{\hat\theta_k\}$, $\{\bm z_k\}$, $\{\hat\theta_k, \bm z_k\}$. They will be discussed in detail in later sections. 

% \par
% Nearest neighbor methods \citep{Altman1992nonparametric, Hall2008kth} also consider conducting inference with respect to a set of similar individuals. The iGroup method differs from the nearest neighbor method in the following ways. Firstly, while a size-k neighborhood in k-nearest neighbor methods is usually chosen based on the similarity of covariates, iGroup method measures similarity based on both the exogenous variable $\bm z$ and the individual level point estimate $\hat\theta$. Secondly, instead of treating every individual in the neighborhood equally as in nearest neighbor methods, iGroup method allows for assigning different weights for different individuals based on their similarity to the target one, and we strongly suggest using such a weight function as described in Section \ref{sec: agg-weight}. 

\par
To ease our notation, from now on, let us say our goal is to provide an estimation on $\theta_0$ for the individual 0. The estimator is constructed with a specified loss function $L$, the observations $(\bm x_0, \bm z_0)$ on individual 0 and all other available observations $\mathcal D_x=\{\bm x_k\}_{k=1}^K$ and $\mathcal D_z=\{\bm z_k\}_{k=1}^K$.
% Specifically, $\Theta_0$ minimizes the expected loss given the observation $(\bm x_0, \bm z_0)$ and other individuals' available data $\mathcal D_x$ and $\mathcal D_z$.
% $$
% \Theta_0(\mathcal D_x, \mathcal D_z;L) = \argmin_{\theta}\  \mathbb E[L(\theta, \theta_0)\mid \bm x_0, \bm z_0, \mathcal D_x, \mathcal D_z].
% $$
% $\Theta_0(\mathcal D_x, \mathcal D_z;L)$ corresponds to the Bayes estimator when $\theta_0$ is assumed to follow a 'prior' distribution induced from the population-wise observations $\mathcal D_x$ and $\mathcal D_z$. In particular, under squared loss $L(\theta, \theta_0) = (\theta- \theta_0)^2$, $\Theta_0(\mathcal D_x, \mathcal D_z;L)$ corresponds to the 'posterior' mean. 
% To be specific, our goal is to provide a Bayesian point estimation for $\theta_0$ for the individual 0, denoted by $\Theta(\theta_0| {\cal D}_A)$,  given available data information ${\cal D}_A$ that may include the observations $\bm x_0$ and $\bm z_0$ plus perhaps information from other relevant individuals.
% The individualized group with relevant individuals is used to form an aggregated estimator that can approximate the corresponding Bayes estimator asymptotically. 
By focusing on individualized local structures, the proposed iGroup learning is robust and effective for handling heterogeneity arising from diverse sources in big data, and it is ideally suited for specific objective-oriented applications in individualized inference. Additionally, in terms of computation, by ignoring a large number of irrelevant entities and zooming directly to the relevant individuals, the iGroup learning is parallel in nature and can scale up better for big data. In this paper, we investigate the validity and theoretical property of iGroup learning and provide simulation studies and applications to demonstrate the grouping effectiveness of the proposed methodology. 

\par
There are also other methods that borrow strength from others to strengthen inference results for the target individual. 
A related classical approach is the k-nearest neighbor methods  (k-NN) \citep{Altman1992nonparametric, Hall2008kth}. The main different between the k-NN and the iGroup methods is the covariates used for identifying similarity and near neighbors. The k-NN method identifies neighborhoods usually based on covariates often without measurement errors, for example, the regressors in a regression problem. In iGroup, the covariates used for grouping, such as $\bm z$ and $\hat\theta$, are both assumed to have measurement errors. Especially, the individual level point estimator $\hat\theta$ has never been used to measure similarity in nearest neighbor algorithms. Additionally, while the k-NN method treats every instance in the neighborhoods as equally important, the iGroup method allows different weight assignments for different individuals, which brings more flexibility. We recommend to use new weight functions in Section \ref{sec: agg-weight} to incorporate the similarity between neighbor individuals and the target one. Theoretically, when number of individuals $K$ approaches infinity, the radius of neighborhood identified by the k-NN method shrinks to zero as a result of bias-variance tradeoff. However, in iGroup approach, the radius of the target neighborhood does not necessarily shrink to zero, because the measurement error in $\hat\theta$ always exists due to finite sample size $n_k=O(1)$. 

\par
The most recent related development is perhaps the individualized fusion learning ($i$-Fusion) approach proposed in \cite{Shen2018ifusion}. The $i$Fusion approach is developed under the asymptotic settings that $n_k\rightarrow \infty$, $n_k / \sum_{i=1}^K n_k =O(1)$ and $K$ is large but finite, where $n_k = |\bm x_k|$ is the effective sample size for individual $k$. The requirement that $n_k\rightarrow \infty$ ensures the individual studies are not biased. Furthermore, the target neighbor, referred to as {\it clique} in the  
$i$Fusion approach, is defined through the parameter space using only $\{\hat\theta_k\}$'s. \cite{Shen2018ifusion} has demonstrated that the $i$Fusion approach is effective with good theoretical properties (including consistency, oracle efficiency and asymptotic normality) under their assumed setting. 
The iGroup approach in this article, however, focuses on a different setting where each individual has only a limited number of observations with $n_k=O(1)$ and infinite numbers of individuals are available as $K\rightarrow \infty$, under which $i$-Fusion is not applicable. A key development of the proposed iGroup method is that we need to make the efforts to overcome the biases from individual estimates. Furthermore, in addition to borrow information through  $\{\hat\theta_k\}$, we also investigate how we can effectively borrow strength from other individuals when the information sets $\{\bm z_k\}$ and  $\{\hat\theta_k, \bm z_k\}$ are available.
\par

The rest of the article is arranged as below.  In Section \ref{sec:general-framework}, we introduce the general framework of iGroup learning. Section \ref{sec:theoretical-results} focuses on three different information sets with asymptotic analysis and theoretical results. Section \ref{sec:simulation} provides three simulated studies and Section \ref{sec:examples} provides two real data applications. Section \ref{sec:conclusion} concludes.

\section{General Framework}\label{sec:general-framework}
\subsection{Problem setup}\label{sec:framework}
Assume for each individual $k\in \{0, 1, 2,\dots, K\}$, we observe $(\bm x_k, \bm z_k)$, where observations ${\bm x}_k$ and ${\bm z}_k$ differ in their utilities. Specifically,  ${\bm x}_k$ is the observed data that is directly related to the parameter of interest $\theta_k$ at the individual level, with a known distribution $\bm x_k\sim p(\cdot | \theta_k)$. The exogenous variable${\bm z}_k$ serves as a proxy that reveals the similarity among $\theta$'s in the population level. Specifically, we assume that $\bm z_k$ is related to an unknown parameter $\bm\eta_k$ through an unknown distribution $q(\cdot; \bm \eta_k)$, and the parameter $\theta$ is an unknown continuous function of $\bm\eta$, i.e. $\theta = g(\bm\eta)$, where the function $g(\cdot)$ is not necessarily an one-to-one mapping. The continuity of $g(\cdot)$ guarantees that closeness in $\bm\eta$ implies closeness in $\theta$. The hierarchical structure and the relationship among the variables are demonstrated in Figure \ref{fig: diagram},
% \begin{figure}[!htpb]
% \centering
% $$\begin{aligned}[c]
%     \theta_k &\sim \pi(\cdot),\\
%     \bm x_k |\theta_k &\sim p(\cdot;\theta_k),
% \end{aligned}
% \hspace{60pt}
% \begin{aligned}[c]
%     \theta_k &= g(\bm\eta_k),\\
%     \bm z_k|\bm\eta_k &\sim q(\cdot;\bm\eta_k).
% \end{aligned}
% \hspace{60pt}
% \begin{tikzpicture}[baseline= (a).base]
% \node[scale=1.2] (a) at (0,0){
% \begin{tikzcd}
% \theta_k \arrow{d} & \bm\eta_k \arrow{d} \arrow{l}\\
% \bm x_k & \bm z_k 
% \end{tikzcd}};
% \end{tikzpicture}
% $$
% \caption{Hierarchical structure and parameter diagram.}\label{fig: diagram}
% \end{figure}
\begin{figure}[!htpb]
\centering
\begin{tabular}{c@{\hspace{60pt}}c@{\hspace{60pt}}c}
$\begin{aligned}
    \theta_k &\sim \pi(\cdot),\\
    \bm x_k |\theta_k &\sim p(\cdot;\theta_k),
\end{aligned}$ & $\begin{aligned}[c]
    \theta_k &= g(\bm\eta_k),\\
    \bm z_k|\bm\eta_k &\sim q(\cdot;\bm\eta_k).
\end{aligned}$ &
$$\begin{tikzpicture}[baseline= (a).base]
\node[scale=1.2] (a) at (0,0){
\begin{tikzcd}
\theta_k \arrow{d} & \bm\eta_k \arrow{d} \arrow{l}\\
\bm x_k & \bm z_k 
\end{tikzcd}};
\end{tikzpicture}$$ \\
$\bm x$ model & $\bm z$ model & diagram
\end{tabular}
\caption{Hierarchical structure and parameter diagram.}\label{fig: diagram}
\end{figure}
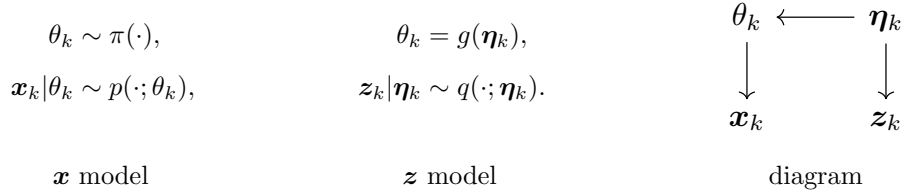
where $\pi(\cdot)$ is an unknown (prior) population distribution of $\theta$, which may be heterogeneous in nature. 
Although $\pi(\cdot)$ is unknown and unspecified, it appears in theoretical calculations throughout  
 the theoretical analysis in this paper. 
Without further clarification, all unconditioned expectations $\mathbb E[\cdot]$ are assumed to take over all random variables including $\theta_k$, which follows the unknown prior $\pi(\cdot)$. Posterior expectations on $\theta$ conditioned on certain observed information are explicitly noted with $\pi$ in the subscript such as $\mathbb E_\pi [\theta_0\mid \hat\theta_0]$.
The distribution $p(\cdot;\theta_k)$ is known except the parameter $\theta_k$, but both the function $g(\cdot)$ and the distribution $q(\cdot;\cdot)$ are unknown.
The role of the exogenous variable ${\bm z}_k$ will be discussed further in later sections. In some cases ${\bm z}_k $ may not be available. 
\par
One example of the above setup is that $\bm x_k$ is the daily stock price returns of company $k$, which follows a $Normal(0, \theta_k^2)$ distribution and $\bm z_k$ is the company's characteristics (e.g. sectors, capital sizes, financial exposure, etc), which is related to stock volatility $\theta_k$. Another example is that $\bm x_k$ is a binary indicator whether individual $k$ has a certain disease and $\bm z_k$ is the individual's health indices such as weight, height, blood pressure, etc., where the underlying $\theta_k=P(\bm x_k=1)$ is the probability of infection.
% \begin{figure}[!hptb]
% \centering
% \begin{tikzcd}
% \bm x_k & \bm z_k \\
% \theta_k \arrow{u} & \eta_k \arrow{u} \arrow{l}
% \end{tikzcd}
% \caption{Parameter diagram.}\label{fig: diagram}
% \end{figure}
\par
Denote by $C_0 (\epsilon) = \{k|\tilde d(\theta_k, \theta_0)<\epsilon, k=0,\dots, K\}$ an $\epsilon$-neighborhood (or a {\it clique}) of individual 0, where $\tilde d(\cdot, \cdot)$ is a distance/similarity measure and $\epsilon$ is the threshold value. Thus, the clique $C_0 (\epsilon)$ is a set of indexes of individuals that are similar to individual 0. 
In our model development, we impose two regularity assumptions as below. 
% These two  assumptions can often be easily satisfied. \\
\begin{assump}[Dense Assumption]\label{assump:dense}
% Denote the neighborhood or a clique of the individual $i$, which is a set of the index of individuals that are similar to the $i$-th individual, as $C_i (\epsilon) = \{k|d(\bm z_k, \bm z_i)<\epsilon, k=0,\dots, K\}$, where $d(\cdot, \cdot)$ is a distance/similarity measure and $\epsilon$ is the threshold value. 
There exists a constant $d\geqslant 1$ such that for all $i=1,\dots, K$, $|\mathcal C_0(\epsilon)|\asymp K\epsilon^d$ in probability when $K\rightarrow \infty, \epsilon\rightarrow 0$.
\end{assump}
\begin{assump}[Smooth Parameter Assumption]\label{assump:smooth}
There exists a positive constant $\kappa$, such that for all $\theta,\theta'\in\Omega_\theta$
$$\sup_{\bm x} |p(\bm x;\theta)-p(\bm x;\theta')|\leqslant\kappa \|\theta-\theta'\|,$$
where $\|\cdot\|$ is a metric on $\Omega_\theta$.
\end{assump}

The dense assumption suggests that individual 0 of interest is not isolated from other individuals, i.e. for arbitrarily small $\epsilon$, there are a sufficiently large number of other individuals in its neighborhood as $K \to \infty$. The smooth parameter assumption guarantees that whenever $\theta$ and $\theta'$ are close, the distributions of $\bm x$ and $\bm x'$ induced from $\theta$ and $\theta'$,  respectively, are close to each other.
%, so does reduction of $\bm x$, e.g. the sufficient statistics for $\theta$. 
Under these two assumptions, it is beneficial to aggregate information from the neighborhood to estimate $\theta$ since one can always find sufficient number of similar individuals in the neighborhood of individual $\theta$. A key consideration in this aggregation is the familiar bias-variance trade-off --- aggregation over a larger group increases the sample size thus reduces estimation variance, but it also brings bias.

\subsection{Aggregated estimation in iGroup}\label{sec: agg-weight}
There are two common methods to aggregate information by creating `pooled' estimators for $\theta_0$. 
The first approach constructs a weighted estimator $\hat\theta_0^{(c)}(\bm x_0, \bm z_0, \mathcal D_x, \mathcal D_z)$ for the target individual 0, directly using the point estimators $\hat \theta_k$ of other individuals based on $\bm x_k$.
%, which do not necessarily lie in the neighborhood of the individual $i$ as 0 weight can be assigned. 
The second approach aggregates objective functions $M_k(\theta) = M_k(\theta, \bm x_k)$ of other individuals, where the point estimator $\tilde\theta_0^{(c)}$ is obtained by optimizing an aggregated objective function. Specifically, these two methods can be formulated as 
\begin{align}
&(\text{Aggregating estimators}) & \hat\theta_0^{(c)} &=\dfrac{\sum_{k=0}^K\hat\theta_kw(k;0)}{\sum_{k=0}^Kw(k;0)},\label{eq: aggtheta}\\
&(\text{Aggregating objective functions}) & \tilde\theta_0^{(c)} &=\argmin_\theta \sum_{k=0}^KM_k(\theta)w(k;0),\label{eq: aggm}
\end{align}
where $w(k;0)$ is the weight assigned to individual $k$ when constructing iGroup estimator for individual 0. 

The weight $w(k;0)$ is crucial for the aggregated estimators as it controls how much information is borrowed from other individuals. We propose to incorporate both individual level estimator $\hat\theta_k$ and exogenous observation $\bm z_k$ into the weight function as both can provide useful information of $\theta_0$. Specifically, let
\begin{equation}\label{eq: weight0}
    w(k; 0) = w(\hat\theta_k, \bm z_k;\hat\theta_0, \bm z_0) = w_1(\bm z_k, \bm z_0)w_2(\hat\theta_k, \hat\theta_0|\bm z_0, \bm z_k).
\end{equation}
The weight is decomposed into two parts. The first part $w_1(\bm z_k, \bm z_0)$ measures the similarity between $\bm z_k$ and $\bm z_0$, and can be a kernel function
\begin{equation}\label{eq: weight1}
    w_1(\bm z_k, \bm z_0) =\mathcal K_1\left(\dfrac{\|\bm z_k-\bm z\|}{b_1}\right),
\end{equation}
When $\mathcal K_1$ has a finite support, the weight function has a hard grouping structure --- individuals lying far enough from individual 0 are not considered at all. Otherwise, it has a soft grouping structure such that dissimilar individuals are assigned with non-zero but tiny weights.
\par 
The second part $w_2(\hat\theta_k, \hat\theta_0|\bm z_0)$ measures the similarity between $\hat\theta$'s. But unlike $w_1$, using a distance measure such as $\mathcal K_2(\| \hat\theta_k - \hat\theta_0\|/b_2)$ is not a good practice, since it ignores the error in $\hat\theta_0$ and  $\hat\theta_k$ and $\hat\theta_0$ may be biased. Note that when $K\rightarrow \infty$ and $b_2\rightarrow 0$, the kernel concentrates on a smaller and smaller area adjacent to $\hat\theta_0$. In this area, aggregating individual $\hat\theta_k$ will not improve the estimation of $\theta_0$. An example of one-dimension case is shown in Figure \ref{fig: wrong_kernel}. Vertical bars mark the locations of $\hat\theta_k$. When $\hat\theta_0$ is away from its target value $\theta_0$, a small bandwidth $b_2$ tends to give large weights to individuals in a local region around $\hat\theta_0$. Aggregating these individual $\hat\theta_k$ in such a local region will not correct the bias $\hat\theta_0-\theta_0$.

\begin{figure}
    \centering
    \includegraphics[width=0.95\textwidth]{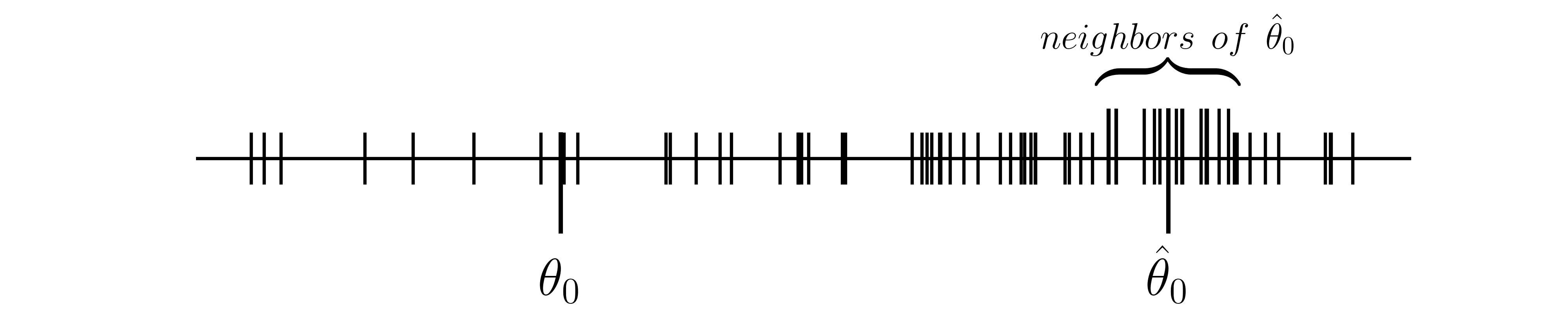}
    \caption{A one-dimension % population
    %with
    example 
    in which
    $\hat\theta_0$ is
    % far 
    away from $\theta_0$. %iGroup 
    If one naively  select individuals according  to  $\hat\theta_0$ and $\hat\theta_k$ directly, individuals  adjacent to $\hat\theta_0$, but not those close to $\theta_0$, are often selected.}
    \label{fig: wrong_kernel}
\end{figure}

\par
We propose the following weight function that considers the distribution $p(\hat\theta|\theta)$ instead of the point estimator $\hat\theta$. Specifically, let 
\begin{equation}
w_2(\hat\theta_k,\hat\theta_0| \bm z_0, \bm z_k) = \dfrac{\int p(\hat\theta_k|\theta)p(\hat\theta_0|\theta)p(\theta|\bm z_0)d\theta}{p(\hat\theta_k|\bm z_k)p(\hat\theta_0|\bm z_0)}. \label{eq: weight}
\end{equation}
Notice that, the posterior distribution of $\theta_0$, given $(\hat\theta_0, \bm z_0)$, is 
$$p(\theta_0|\hat\theta_0, \bm z_0) = p(\theta_0, \hat\theta_0|\bm z_0)/p(\hat\theta_0|\bm z_0) \allowbreak = \allowbreak p(\hat\theta_0|\theta_0)p(\theta_0|\bm z_0)/p(\hat\theta_0|\bm z_0).$$ 
If $\theta_k \equiv \theta_0$ (hence $\hat\theta_k$ provides useful information about $\theta_0$), then the predictive distribution of $\hat\theta_k$, given 
$(\hat\theta_0, \bm z_0)$,  is
$$p(\hat\theta_k | \hat\theta_0, \bm z_0) = \int p(\hat\theta_k|\theta)p(\theta|\hat\theta_0, \bm z_0)d\theta = \dfrac{\int p(\hat\theta_k|\theta)p(\hat\theta_0|\theta)p(\theta|\bm z_0)d\theta}{p(\hat\theta_0|\bm z_0)}.$$
% It measures the likelihood that $\hat\theta_k$ shares the same parameter $\theta_k$ with individual 0.
Thus, the weight function $w_2(\hat\theta_k, \hat\theta_0|\bm z_0, \bm z_k)$ in (\ref{eq: weight}) is the Radon-Nikodym derivative between the predictive distribution $p(\hat\theta_k|\hat\theta_0, \bm z_0)$ and the sampling distribution $p(\hat\theta_k | \bm z_k)$.
% $$w_2(\hat\theta_k,\hat\theta_0| \bm z_0)=\dfrac{p(\hat\theta_k|\hat\theta_0, \bm z_0)}{p(\hat\theta_k | \bm z_k)},$$
As a result, for any measurable function $h(\cdot)$, we have
$$\mathbb E_{p(\hat\theta_k|\bm z_k)} [h(\hat\theta_k)w_2(\hat\theta_k,\hat\theta_0|\bm z_0, \bm z_k)] = \mathbb E_{p(\hat\theta_k|\hat\theta_0, \bm z_0)}[h(\hat\theta_k)].$$
That is, the weighted expectation of $h(\hat\theta_k)$ under the sampling distribution $p(\hat\theta_k |\bm z_k)$ equals to its expectation under the predictive distribution $p(\hat\theta_k|\hat\theta_0, \bm z_0)$ if $\theta_k =\theta_0$. 
This property brings invariance under different sampling distributions. More importantly, it shows that the weighted averages, such as (\ref{eq: aggtheta}) and (\ref{eq: aggm}), estimates the expectations under the predictive distribution. This gives the iGroup estimators promising asymptotic properties as we will discuss later in Section \ref{sec:theoretical-results}.
\par
% There is no bandwidth parameter in the weight formula in (\ref{eq: weight}). 
The shape (thin or flat) of the weight $w_2(\cdot)$ as a function of $\hat\theta_k$ does not change with the number of individuals $K$. However, the shape is influenced by the variation (accuracy) of $\hat\theta$. The larger the variance of $\hat\theta$ is, the flatter the weight function tends to be. If $\hat\theta_k$ is estimated without any measurement error, the weight $w_2(\hat\theta_k, \hat\theta_0|\bm z_0, \bm z_k)$ is proportional to the indicator function $I_{\{\hat\theta_k= \hat\theta_0\}}$. It reduces to the case in which the individual estimator $\hat\theta_0$ or the individual objective function $M_0(\theta)$ is used without grouping.

\subsection{Evaluating the weight functions}\label{sec:evaluate-weight}
% Given the form of weight function (\ref{eq: weight}), it remains to calculate  the weight function directly. The first thing is to consider 

The weight function
 $w_1(\bm z_k, \bm z_0)$ in (\ref{eq: weight1}) can be directly evaluated. Similar to a bandwidth selection problem for kernel smoothing, one can choose the bandwidth $b_1$ for $w_1(\bm z_k, \bm z_0)$ in (\ref{eq: weight1}) by either using the plug-in method \citep{chiu1991bandwidth} or through cross-validation procedure. The plug-in bandwidth is proportional to $K^{-\frac{1}{d+4}}$ (see Section \ref{sec:theoretical-results}). Also, the leave-one-out cross validation process gives an empirical optimal bandwidth, as discussed in Section \ref{sec:bandwidth-selection}. 
\par
The evaluation of the weight function $w_2(\hat\theta_k,\hat\theta_0\mid \bm z_0, \bm z_k)$ in (\ref{eq: weight}) is more complicated, since the conditional probability $p(\hat\theta|\bm z)$ and the integral $\int p(\hat\theta_0|\theta)p(\hat\theta_k|\theta)p(\theta|\bm z_0)d\theta$ are unknown as the relationship between $\theta$ and $\bm z$ is not explicit. We propose an approximation method to evaluate $w_2(\hat\theta_k,\hat\theta_0\mid \bm z_0, \bm z_k)$ below. 
\par
Denote the estimator of $\theta_k$ and the observed exogenous variable $\bm z_k$ as the tuple
$(\hat\theta_k, \bm z_k), k=0,\dots, K.$ 
To calculate the weight in (\ref{eq: weight}), we treat them as $K+1$ samples from the joint distribution of $(\hat\theta, \bm z)$. We use the kernel method to estimate the conditional probability $p(\hat\theta\mid \bm z)$ nonparametrically by
$$\hat p(\hat\theta|\bm z) = \dfrac{\displaystyle\sum_{j=0}^K \mathcal K_1\left(\dfrac{\|\bm z-\bm z_j\|}{b_1}\right)\mathcal K_2\left(\dfrac{\|\hat\theta-\hat\theta_j\|}{b_2}\right)}{\displaystyle\sum_{j=0}^K \mathcal K_1\left(\dfrac{\|\bm z-\bm z_j\|}{b_1}\right)},$$
where $\mathcal K_1, \mathcal K_2$ are two kernel functions with $b_1$, $b_2$ as the corresponding bandwidths. To estimate the integral in (\ref{eq: weight}), we use the interpretation discussed above that it is the conditional distribution $p(\hat\theta_k\mid \hat\theta_0, \bm z_0)$ given $\theta_k=\theta_0$. Hence we need samples from the joint distribution of $(\hat\theta,\hat\theta',\bm z)$ observed from the same individual with parameter $\theta$. However, this is infeasible because in our problem setting, no two individual share the same true parameter $\theta$ and for each individual only one $\hat\theta$ is observed. To generate samples from such a distribution, we consider a bootstrap method. 
Denote $\hat\theta^{(1)}_k$ and $\hat\theta^{(2)}_k$ as the two bootstrap estimators for $\theta_k$, obtained by re-sampling $\bm x_k$ with replacement (not applicable when $\bm x_k$ has few observations). Then $(\hat\theta^{(1)}_k, \hat\theta^{(2)}_k, \bm z_k), k=0,\dots,K$ is an approximate sample of $(\hat\theta, \hat\theta', \bm z)$, guaranteeing $\hat\theta^{(1)}_k, \hat\theta^{(2)}_k, \bm z_k$ are generated from the same individual $k$. Therefore the integral can be estimated by 
$$\int p(\hat\theta_0|\theta)p(\hat\theta_k|\theta)p(\theta|\bm z_0)d\theta\approx \dfrac{\displaystyle\sum_{j=0}^K\mathcal K_1\left(\dfrac{\|\bm z_0 - \bm z_j\|}{b_1}\right)\mathcal K_2\left(\dfrac{\|\hat\theta_0 - \hat\theta^{(1)}_j\|}{b_2}\right)\mathcal K_3\left(\dfrac{\|\hat\theta_k - \hat\theta^{(2)}_j\|}{b_3}\right)}{\displaystyle\sum_{j=0}^K\mathcal K_1\left(\dfrac{\|\bm z_0 - \bm z_j\|}{b_1}\right)},$$
where $\mathcal K_1, \mathcal K_2, \mathcal K_3$ are three kernel functions with $b_1,b_2,b_3$ as the corresponding bandwidths. The bandwidths can be selected by either minimizing asymptotic mean integrated squared error (AMISE) or a rule-of-thumb bandwidth estimator. This estimation of the integral is an approximation that requires $K$ to be sufficiently large. %Moreover, $\hat\theta^{(1)}$ and $\hat\theta^{(2)}$ are generally not independent conditioning on $\theta$, but they are approximately conditionally independent when the number of observations in $\bm X$ is sufficiently large. 
\par
%After estimating both the numerator and the denominator, the only thing left 

\section{Theoretical Results}\label{sec:theoretical-results}
%The general framework of iGroup considers both $\hat\theta$ and $\bm z$ for the similarity measure. However, for some practical applications, $\hat\theta$ or $\bm z$ may not be feasible. The situations that $\hat \theta$ is feasible or $\hat\theta$ is too noisy to be considered are discussed in Case 1 and Case 2. One example for infeasible $\hat\theta$ happens when each $\bm x_k$ has only one observation, resulting in the failure to construct the estimate of variance $\hat\sigma^2$ and the estimate of the 5 percentage quantile $\hat Q_{0.05}$. Another case that $\hat \theta$ is ignored happens when $\hat \theta$ is very noisy, contributing little about the accuracy but requiring intensive computation, because our weights on $\hat \theta$ is only an approximation based on bootstrap. If no exogenous data resources could be found, $\bm z$ cannot be constructed (Case 3). For the above cases, iGroup need to be constructed solely on $\hat\theta$ or $\bm z$, whichever is available. Also, the complete general framework with both $\hat\theta$ and $\bm z$ feasible is discussed in Case 4.  

In this section, we consider several model settings for which we apply the proposed iGroup method and discuss their corresponding theoretical properties, especially in terms of their asymptotic performance. In particular, we first define a target estimator $\Theta_0$ that minimizes the Bayes risk, 
%based on risk decomposition, 
and then investigate the asymptotic performance of iGroup estimators in (\ref{eq: aggtheta}) and (\ref{eq: aggm}) in approximating the target estimator $\Theta_0$. We also quantify the bias and variance of iGroup estimators as well as the target estimator $\Theta_0$ in term of estimating $\theta_0$. Throughout this paper, we consider the asymptotic framework that the number of individuals $K$ goes to infinity, while the number of observations for each individual $n$ is fixed and finite.
\subsection{Risk decomposition and the target estimator}
We are interested in making inference about individual 0, with given data information $\mathcal D_x, \mathcal D_z$ that may include the observations $\bm x_0$ and $\bm z_0$ plus information from other relevant individuals. Let $\delta_0(\mathcal D_x, \mathcal D_z)$ be a point estimator for $\theta_0$, which is constructed with information sets $\mathcal D_x$ and $\mathcal D_z$. The iGroup estimator $\hat\theta_0^{(c)}$ in (\ref{eq: aggtheta}) is such an estimator. Similarly, $\delta_0(\mathcal D_x)$ and $\delta_0(\mathcal D_z)$ are point estimators constructed solely based on either $\mathcal D_x$ or $\mathcal D_z$. Under squared loss, the overall risk of $\delta_0$ in estimating $\theta_0$  can be decomposed into two nonnegative parts: the expected squared error of $\delta_0$ in estimating the corresponding posterior mean and the overall risk of the posterior mean itself, as shown in Proposition \ref{thm: risk-decomposition-simple-l2}.
\begin{prop}\label{thm: risk-decomposition-simple-l2}
Suppose $\theta_0$ has a prior distribution $\pi(\cdot)$. Under squared loss, we have the following overall risk decomposition.
\begin{align*}
    \mathbb E[(\delta_0(\mathcal D_x, \mathcal D_z) -\theta_0)^2] &= \mathbb E[(\delta_0(\mathcal D_x, \mathcal D_z) -\mathbb E_\pi[\theta_0\mid \bm x_0, \bm z_0])^2] + \mathbb E[(\mathbb E_\pi[\theta_0\mid \bm x_0, \bm z_0] - \theta_0)^2],\\
    \mathbb E[(\delta_0(\mathcal D_x) -\theta_0)^2] &= \mathbb E[(\delta_0(\mathcal D_x) -\mathbb E_\pi[\theta_0\mid \bm x_0])^2] + \mathbb E[(\mathbb E_\pi[\theta_0\mid \bm x_0] - \theta_0)^2],\\
    \mathbb E[(\delta_0( \mathcal D_z) -\theta_0)^2] &= \mathbb E[(\delta_0(\mathcal D_z) -\mathbb E_\pi[\theta_0\mid\bm z_0])^2] + \mathbb E[(\mathbb E_\pi[\theta_0\mid\bm z_0] - \theta_0)^2],
\end{align*}
where $\mathbb E_\pi[\theta_0\mid \bm x_0, \bm z_0]$, $\mathbb E_\pi[\theta_0\mid \bm x_0]$ and $\mathbb E_\pi[\theta_0\mid\bm z_0]$ are the posterior means under prior $\pi(\cdot)$ and observations $(\bm x_0, \bm z_0)$, $\bm x_0$ and $\bm z_0$ correspondingly.
\end{prop}
The proof is given in Appendix. 
\par 
Proposition \ref{thm: risk-decomposition-simple-l2} reveals that the overall risk is minimized by setting $\delta_0$ to the corresponding posterior mean under the prior $\pi(\cdot)$, which is the population-level (unknown) distribution for $\theta_0$. Throughout this paper, we call the estimator 
that minimizes the overall risk 
the {\it target estimator}.
%and denote it with $\Theta_0$.
More specifically, under squared loss and different information sets, we denote the target estimators with 
\begin{equation} \label{eq:Thetal2}\Theta_0(\bm x_0;\ell_2) = \mathbb E_\pi[\theta_0\mid\bm x_0],\quad \Theta_0(\bm z_0;\ell_2) = \mathbb E_\pi[\theta_0\mid\bm z_0] \text{  and  }\ \Theta_0(\bm x_0, \bm z_0;\ell_2) = \mathbb E_\pi[\theta_0\mid\bm x_0, \bm z_0].\end{equation}
Here, $\ell_2$ refers to the squared loss. 
For the ease of presentation, we also use a simple notation $\Theta_0$ to represent one of the Bayes estimators in (\ref{eq:Thetal2}) when its meaning is apparent.
\par
Similarly, for a general loss function $L(\hat\theta, \theta)$, we define the target estimator as 
the Bayes estimator
that minimizes the expected loss, given the available observation on individual 0 and the prior $\pi(\cdot)$ such that
\begin{align}\label{eq:ThetaL}
    \Theta_0(\bm x_0;L)&=\argmin_{\delta} \mathbb E_\pi[L(\delta, \theta_0)\mid \bm x_0],\nonumber \\
    \Theta_0(\bm z_0;L)&=\argmin_{\delta} \mathbb E_\pi[L(\delta, \theta_0)\mid \bm z_0],\\
    \Theta_0(\bm x_0, \bm z_0;L)&=\argmin_{\delta} \mathbb E_\pi[L(\delta, \theta_0)\mid \bm x_0, \bm z_0]. \nonumber
\end{align}
A similar risk decomposition is demonstrated in Proposition \ref{thm: risk-decomposition-simple-L} below. Again, for the ease of notation, we simply use $\Theta_0$ to represent one of the Bayes estimators in (\ref{eq:ThetaL}) when its meaning is apparent.
\begin{prop}\label{thm: risk-decomposition-simple-L}
Suppose $\theta_0$ has a prior distribution $\pi(\cdot)$ and $L(\hat\theta, \theta)$ is a loss function, which is second-order partially differentiable with respect to $\hat\theta$ such that $L'(\hat\theta, \theta) = \partial L/\partial \hat\theta$ and $L''(\hat\theta, \theta) = \partial^2 L/\partial \hat\theta^2$. Then for estimator $\delta_0$ constructed based on information set $\mathcal D_x$, $\mathcal D_z$ or $(\mathcal D_x, \mathcal D_z)$, we have
\begin{align*}
    \mathbb E[L(\delta_0,\theta_0)] &=\dfrac{1}{2}\mathbb E[L''(\Theta_0, \theta_0)(\delta_0 - \Theta_0)^2] + \mathbb E[L(\delta_0, \theta_0)] + o(\mathbb E[(\delta_0 - \Theta_0)^2]),
\end{align*}
where $\Theta_0$ is the corresponding Bayes estimator based on the same information set as $\delta_0$.
\end{prop}
The proof is given in Appendix. 

The target estimator $\Theta_0$ as a function of $\bm x_0$ and $\bm z_0$ is not directly available, %feasible, 
because neither the population distribution $\pi(\theta_0)$ nor the likelihood function $p(\bm z_0\mid \theta_0)$ is explicitly known or assumed. The iGroup estimator $\hat\theta_0^{(c)}$ in (\ref{eq: aggtheta}) constructed based on observed finite sample $\mathcal D_x, \mathcal D_z$ is desired to approach the target estimator $\Theta_0$ when more and more similar individuals contribute to the estimator $\hat\theta_0^{(c)}$. See \cite{diaconis1986consistency} for discussions of target point estimators and target parameters in Bayesian literature.

\subsection[]{Case 1: With exogenous variable $\bm z$ only}\label{sec:exogenous-variable}
% In cases when the individual level estimator $\hat\theta_k$ is infeasible or too noisy to be considered, iGroup need to
In the cases when the individual level estimator $\hat\theta_k$ is not reliable to construct the individual groups, iGroup may  
be constructed with the exogenous variable $\bm z$ only. 
In this case, the corresponding target estimator is defined as: 
\begin{equation}
\Theta_0(\bm z_0;\ell_2) = \mathbb E_\pi[\theta_0\mid \bm z_0],
\label{eq: target-z}
\end{equation}
where $p(\theta_0\mid\bm z_0) \propto p(\bm z_0\mid \theta_0)\pi(\theta_0)$. Although $\bm x_0$ is not used for grouping and thus does not appear in (\ref{eq: target-z}), the data $\mathcal D_x$ is used in iGroup estimators in (\ref{eq: aggtheta}) and (\ref{eq: aggm}). 

% $\Theta_0(\mathcal D_z) = \mathbb E[\theta_0\mid \bm z_0]$ under squared loss and is $\Theta_0(\mathcal D_z; L)$ under a general loss function $L$.

%[Question: Should we use $E(\theta_0 | \bm z_0)$ to replace $\theta_0$ in the theorems of this subsection? (I think some places are indeed $\theta_0$, so I guess we should make the distinction)]

% \subsubsection{Deterministic Exogenous Variable}\label{sec:deterministic-exogenous-variable}
% When $\hat \theta$ is not considered and 
Recall that the relationship between $\theta_k$ and $\bm \eta_k$ is given by a deterministic relationship
\begin{equation}\label{eq:relation}\theta_k = g(\bm\eta_k),\quad \text{for }k=0,1,\dots, K, 
\end{equation}
where $g(\cdot)$ is an unknown continuous function. Furthermore, $\bm z_k$ is a noisy observation of $\bm \eta_k$. Since $\bm\eta$ is a conceptual parameter, we may simply assume that
$$\bm z_k = \bm \eta_k +\epsilon_k,\quad \text{for }k=0,\dots, K,$$
where the error satisfies $\mathbb E(\epsilon_k)=0$, $\text{Var}(\epsilon_k)=\sigma_z^2\bm\Sigma_z$ with $\|\bm\Sigma_z\|=1$. 
% Consider the case that $\bm\eta_k$ is by a function of $\bm z_k$ without error. Specifically, 
% In this case, $\Theta_0 = \mathbb E_\pi[\theta_0\mid \bm z_0]= g^*(\bm z_0)\equiv \theta_0$.
\par
Suppose $\hat \theta_k$ is an unbiased estimator of $\theta_k$. Then, the combined estimator 
\begin{equation}\label{eq:estimator}
\hat \theta_0^{(c)} = \dfrac{\sum_{k=0}^K\mathcal K\left(\dfrac{\|\bm z_k-\bm z_0\|}{b}\right)\hat \theta_k}{\sum_{k=0}^K\mathcal K\left(\dfrac{\|\bm z_k-\bm z_0\|}{b}\right)}
\end{equation}
has all the properties of a conventional kernel smoothing estimator if $\mathcal K$ is a standard kernel function. The boundary and asymptotic conditions/assumptions on the weight function $\mathcal K$ and the bandwidth $b$ are summarized in Assumption \ref{thm:condition}.
\begin{assump}[Boundary and asymptotic conditions]\label{thm:condition}
The kernel function $\mathcal K(\cdot)$ satisfies
$$
\mathcal K \geqslant 0,\quad
\int |\mathcal K(u)|du<\infty,\quad
\lim_{|u|\rightarrow\infty} u\mathcal K(u)\rightarrow 0.
%K\rightarrow \infty, b\rightarrow 0, b^dK\rightarrow \infty.
$$
And, in addition, when $K\rightarrow \infty$, $b$ satisfies $b\rightarrow 0,\quad b^dK\rightarrow \infty.$

\end{assump}
\begin{thm}\label{thm:hdkernel}
Under the conditions in Assumption \ref{assump:dense} - \ref{thm:condition},
we have 
$$\hat \theta_0^{(c)} \longrightarrow \Theta_0(\bm z_0;\ell_2)\quad \text{in probability.}$$
The optimal choice of the bandwidth is $\hat b \asymp K^{-1/(d+4)}$ such that the optimal MSE is $\mathbb E[(\hat \theta_0^{(c)}-\Theta_0)^2]\asymp K^{-4/(d+4)}$.
\end{thm}
% Under deterministic relationship between $\bm\eta$ and $\bm z$, $\mathbb E(\theta_0\mid \bm z_0)=\theta_0$. 

Theorem \ref{thm:hdkernel} follows immediately from consistency theorem on a standard multivariate kernel smoothing estimator \citep{wassermann2006all}.
When the number of individuals $K$ goes to infinity, the bias of $\hat\theta_0^{(c)}$ with bandwidth $b$ is of order $b^2$ and the variance is of order $(b^dK)^{-1}$, where $d$ is the dimension of $\bm z$ as defined in Assumption \ref{assump:dense}. In such case, the asymptotic optimal choice of bandwidth that minimizes the mean squared error, $b^4+(b^dK)^{-1}$, is of order $K^{-1/(d+4)}$, same as a $d$-dimensional kernel smoothing problem.
\par
Another way of combining individuals is aggregating the objective functions as shown in (\ref{eq: aggm}). A combined estimator with respect to kernel $\mathcal K(\cdot)$ is defined by
$$\tilde \theta^{(c)}_0=\argmin_\theta\sum_{k=0}^K\mathcal K\left(\dfrac{\|\bm z_k-\bm z_0\|}{b}\right)M_k(\theta).$$
% To show this estimator is consistent, we assume a sufficient condition on the kernel function $\mathcal K$ and the objective function $M_k(\theta)$ as in Assumption \ref{assump: aggtheta}. 
The estimator is consistent and has a similar asymptotic performance to a $d$-dimensional kernel smoothing estimator as stated in Theorem \ref{thm:mest}.
This approach is useful especially when $\hat \theta_k$ is not available, such as in the cases that the number of observations for each individual is less than the number of parameters. 
% \begin{assump}Kernel function $\mathcal K$ and objective function $M_k(\theta)$ satisfy:\label{assump: aggtheta}
% \begin{itemize}
%     \item $\mathcal K$ is non-negative, 
%     \item $M_k(\theta)$ is a convex function of $\theta$ for all $\bm x_k$.
% \end{itemize}
% \end{assump}
\begin{thm}\label{thm:mest}
Suppose the conditions in Assumption \ref{thm:condition} hold and in addition,
\begin{enumerate}
\item $M_k(\theta)$ is convex and second order partial differentiable with respect to $\theta$,
\item for any given $\theta$, $\mathbb E_{\bm x|\bm z}[\dfrac{\partial M_{\bm x}(\theta)}{\partial\theta}]$ as a function of $\bm z$ is continuous,
\item $\mathbb E_{\bm x|\bm z_0}[M_{\bm x}(\theta)]$ has a unique minimum at $\theta=\Theta_0(\bm z_0;\ell_2).$
\end{enumerate}
Then 
$$\tilde \theta^{(c)}_0\longrightarrow \Theta_0\quad in\  probability.$$
The optimal choice of bandwidth $b$ is $\hat b\asymp K^{-1/(d+4)}$ and the optimized mean squared error is $\mathbb E[(\hat \theta_0^{(c)}-\Theta_0)^2]\asymp K^{-4/(d+4)}.$
\end{thm}
The proof is given in Appendix. 

%Move to conclusion part.
% We assume a quite strong sufficient condition on the objective function $M_k(\theta)$ such that the argmin of the aggregated objective function will converge to the true value. Instead of assuming second-order differentiable and convexity, other sufficient conditions can also guarantee the convergence of argmin \citep{van2000asymptotic}. But most of them depends on the explicit formula of kernel $\mathcal K$ and the objective function $M_k(\theta)$.
\par
% In conclusion, under some regulation conditions, 
The above theorems suggest that the individualized combined estimator by aggregating either individual estimators $\hat\theta_k$ or objective functions $M_k(\theta)$ would result in an improvement in mean squared error and it shares a similar asymptotic performance as a $d$-dimensional kernel smoothing estimator.

% \subsubsection[]{Noisy Exogenous Variable}\label{sec:noisy-exogenous-variable}
% Suppose the relation in Equation (\ref{eq:relation}) holds for $\theta$ and $\bm \eta$ such that $\theta = g(\bm \eta)$,
% but the relationship between $\bm z$ and $\bm \eta$ is not deterministic. Since $\bm \eta$ is a conceptual parameter, we may simply assume that
% $$\bm z_k = \bm \eta_k +\epsilon_k,\quad \text{for }k=0,\dots, K,$$
% where the error satisfies $\mathbb E(\epsilon_k)=0$, $\text{Var}(\epsilon_k)=\sigma^2\bm\Sigma$ with $\|\bm\Sigma\|=1$. \\

When $\sigma_z=0$, $\Theta_0(\bm z_0;\ell_2) = \mathbb E_\pi[\theta_0\mid\bm z_0] \equiv \theta_0$. Hence, estimating $\Theta_0$ becomes estimating the unknown function $g(\cdot)$ evaluated at $\bm z_0$. When $\sigma_z>0$, $\Theta_0$ and $\theta_0$ are in general different. Let $B_0$ and $V_0$ be the bias and variance of the target estimator $\Theta_0(\bm z_0;\ell_2)$ in estimating $\theta_0$ such that
\begin{equation}
    B_0(\theta_0) := \mathbb E_{\theta_0}[\Theta_0(\bm z_0;\ell_2)] - \theta_0,\quad V_0(\theta_0) = Var_{\theta_0}[\Theta_0(\bm z_0;\ell_2)].
    \label{eq: intrinsic-terms}
\end{equation}

The above bias and variance are defined with respect to a fixed $\theta_0$ with random $\bm z_0$. 
\begin{thm}\label{thm:bias-variance}
The asymptotic bias and variance of $\hat\theta^{(c)}_0$ in estimating a fixed $\theta_0$ are given by
\begin{align*}
\mathbb E_{\theta_0}[\hat\theta_0^{(c)}]-\theta_0 &= B_0(\theta_0)+O_p(b^2),\\
\text{Var}_{\theta_0}[\hat\theta_0^{(c)}]&=V_0(\theta_0)+O_p\left(\dfrac{1}{Kb^d}\right),
\end{align*}
where the intrinsic bias $B_0$ and the intrinsic variance $V_0$ are defined in (\ref{eq: intrinsic-terms}). 
% % $$B_0=\mathbb E[\mathbb E(\theta|\bm z_0)|\bm \eta_0]-\theta_0\quad \text{and}\quad V_0=\text{Var}[\mathbb E(\theta|\bm z_0)|\bm \eta_0]$$
% $$B_0=\mathbb E[\Theta_0|\bm \eta_0]-\theta_0\quad \text{and}\quad V_0=\text{Var}[\Theta_0|\bm \eta_0]$$
% respectively. Note that $\theta = g(\bm\eta)$ and $\theta_0 = g(\bm\eta_0)$.
\end{thm}
The proof is given in Appendix. In the conditional probabilities, $\Theta_0=\mathbb E_\pi[\theta_0\mid \bm z_0]$, as a function of $\bm z_0$, is considered random under a given $\theta_0$.
\par
% The nonparametric kernel density estimator here, in fact, estimates the target parameter, which is $\Theta_0(\bm z_0)=\mathbb E_\pi[\theta_0|\bm z_0]$ under the unknown empirical prior $\pi(\theta)$. 
The bias and variance of $\hat\theta_0^{(c)}$ in terms of estimating a fixed $\theta_0$ can therefore be decomposed into two parts. The first part (the intrinsic part) comes from the bias and variance of estimating $\Theta_0[\bm z_0]$ itself to $\theta_0$ and the second part comes from estimating $\Theta_0$ nonparametrically.
Since $\bm z$ is observed with error, this is similar to error in variable problem where certain intrinsic bias cannot be avoided \citep{fuller2009measurement, carroll1995nonlinear, wansbeek2000measurement, bound2001measurement}. 
Such intrinsic bias and variance are asymptotically linear of $\sigma_z^2$, which is the noise level of $\bm z_k$, as shown in Theorem \ref{thm: intrinsic_asymp}. Especially, when $\sigma_z^2$ is exactly zero, all intrinsic terms vanish, and it reduces to the exact case when $\Theta_0=\theta_0$. 
\begin{thm}\label{thm: intrinsic_asymp}
Suppose $g(\cdot)$ is second-order differentiable and the distribution of $\epsilon_k$ has finite higher moments. Then, for a fixed $\theta_0$,
when $\sigma_z^2\rightarrow 0$, 
\begin{align*}
B_0\asymp \sigma_z^2,\quad V_0\asymp \sigma_z^2.
\end{align*}
\end{thm}
The proof is given in Appendix. 
\par
Research in nonparametric regression with error in variable shows a slower convergence rate to recover the function $\theta_0=g(\bm\eta)$ at any given $\bm\eta$ \citep{stefanski1990deconvolving, fan1993nonparametric}. Our problem is different. We focus on providing a point estimator of $\theta_0=g(\bm\eta_0)$ without knowning $\bm\eta_0$, but its noisy version $\bm z_0$. Even if we known the function $g(\cdot)$ precisely, $\theta_0$ is not known as we do not observe $\eta_0$. 
% Since $\bm\eta_0$ is not directly observed, the Bayes estimator for $\theta_0$ under squared loss is the expectation $\mathbb E[\theta|\bm z_0]$. 
When considering an individual with fixed but unobserved $(\theta_0, \bm\eta_0)$, it is difficult to choose an optimal bandwidth by bias-variance optimization with the non-zero intrinsic terms in Theorem \ref{thm:bias-variance}, because in this case the asymptotic mean squared error $(B_0+O_p(b^2))^2+V_0+ O_p((Kb^d)^{-1})$ may not have a local minimum. 
However, if we assume the target individual 0 is randomly chosen from the population, the target estimator $\Theta_0$ is the estimator that minimizes the overall risk under squared loss, i.e. a Bayes estimator, because it minimizes the squared loss pointwise for any $\bm z_0$.
Furthermore, immediately from Theorem \ref{thm:hdkernel}, $\hat\theta_0^{(c)}$ is a consistent estimator for $\Theta_0$. The overall performance of $\hat\theta_0^{(c)}$ for all individuals of the population could be optimized by choosing a proper bandwidth $b$ as stated in the following Theorem \ref{thm: overall-risk}. It provides a way to optimize the bandwidth globally.
% , though the individual asymptotic optimal bandwidth is still undetermined without knowing the intrinsic bias.

\begin{thm}\label{thm: overall-risk}
Assume Assumption \ref{assump:dense} - \ref{thm:condition} hold, then the estimator $\hat\theta_0^{(c)}$ has the following Bayes risk under squared loss
$$\mathbb E[(\hat\theta_0^{(c)}-\theta_0)^2]=R_0+O_p(b^4)+O_p\left(\dfrac{1}{Kb^d}\right),$$
where 
$$R_0=Var[\Theta_0-\theta_0]$$
is the risk of the Bayes estimator $\Theta_0=E_\pi[\theta|\bm z_0]$, and all above expectations is taken over all random variables assuming an empirical population distribution $\pi(\cdot)$ for $\theta_0$.
The optimal choice of the bandwidth $b$ is $b\asymp K^{1/(d+4)}$ with the corresponding overall risk $R_0+ O_p(K^{4/(d+4)})$.
\end{thm}
The proof is given in Appendix. 
\par
The magnitude of the measurement error of $\bm z_k$, measured by $\sigma_z^2$, compared to that of the individual estimation error is crucial for the performance of the iGroup method. The bias and variance of iGroup estimator increase when $\sigma_z^2$ increases (see Theorem \ref{thm: intrinsic_asymp}). And the asymptotic Bayes risk $R_0$ also depends on $\sigma_z^2$. When iGroup is based on unreliable $\bm z$, it could result in a worse estimator compared to the one without any grouping. This phenomenon will be demonstrated in Section \ref{sec:simulation}.\\
\noindent
\textbf{Remark:}
Results in Theorems \ref{thm:bias-variance}, \ref{thm: intrinsic_asymp} and \ref{thm: overall-risk} can be generalized to the iGroup estimator $\tilde\theta^{(c)}_0$, which combines the objective functions, except that the target estimator changes from $\mathbb E_\pi[\theta|z_0]$ is replaced by $\argmin_\theta \mathbb E_\pi[M(\theta)|z_0]$. As shown in (\ref{eq: aggtheta_kernel}) in the Appendix, $\tilde\theta_0^{(c)}$ is asymptotically a kernel smoothing estimator with the same bias and variance rates. 

% Since the population distribution for $\bm z$ is unknown, the Bayes estimator is not directly feasible. The iGroup estimator turns out to be the Bayes estimator without knowing the prior when $K$ goes to infinity. 
% This is similar to empirical Bayes approach \citep{robbins1985empirical}, but the estimator is structured differently and our objective is different.

\subsection[]{Case 2: Without exogenous variables}\label{sec:theta-only}
In this case, we assume the exogenous variable $\bm z$ is not available. Our target estimator is $\Theta_0(\bm x;\ell_2) = \mathbb E_\pi[\theta_0 |\bm x_0]$ under squared loss and is $\Theta_0(\bm x_0; L) = \argmin_\theta \mathbb E_\pi [L(\theta, \theta_0) \mid \bm x_0]$ under a general loss function $L$. The iGroup estimation depends solely on $\hat\theta$. The weight function (\ref{eq: weight}) used in (\ref{eq: aggtheta}) and (\ref{eq: aggm}) now reduces to
\begin{equation}\label{eq: weightreduce}
w_2(\hat\theta_k,\hat\theta_0)=\dfrac{\int p(\hat\theta_k|\theta)p(\hat\theta_0|\theta)\pi(\theta)d\theta}{\int p(\hat\theta_k|\theta)\pi(\theta)d\theta\int p(\hat\theta_0|\theta)\pi(\theta)d\theta},
\end{equation}
where $\pi(\theta)$ corresponds to the unknown distribution of $\theta$ in the whole population. As discussed in Section \ref{sec:evaluate-weight}, an estimation of this weight function can be achieved by kernel density estimation on the bootstrapped samples $(\hat\theta^{(1)}_k, \hat\theta^{(2)}_k)$.
\par
The weight function (\ref{eq: weightreduce}) is used to aggregated individual unbiased estimators to the posterior mean, and to aggregate objective functions $M:\Omega_\theta\times\Omega_\theta\rightarrow \mathbb R$ to the corresponding Bayes estimator under certain loss function, as shown in Theorems \ref{thm: theta-theta-only} and \ref{thm:m-theta-only}.
\begin{thm}\label{thm: theta-theta-only}
Suppose $w_2(\hat\theta_k,\hat\theta_0)$ is defined as in Equation (\ref{eq: weightreduce}) and $\hat\theta_k$ is a sufficient and unbiased estimator of $\theta_k$ for all $k$, then as $K\rightarrow\infty$:
$$\hat \theta^{(c)}_0\rightarrow \Theta_0(\bm x_0;\ell_2)\quad\text{in probability}.$$
Furthermore, if $\mathbb E_{\hat\theta_0}[w_2^2(\hat\theta_k, \hat\theta_0)]<\infty$ for any fixed $\hat\theta_0$ and $\mathbb E_\pi[\hat\theta^2]<\infty$, then
$$\sqrt{K}(\hat\theta^{(c)}_0 - \Theta_0)=O_p(1).$$
\end{thm}
The proof is given in Appendix.\\
% \begin{cor}
% Under squared loss and prior $\pi(\theta)$, the above combined estimator converges in probability to the Bayes estimator.
% \end{cor}
% Suppose $L:\Theta\times\Theta\rightarrow\mathbb R$ is a loss function, the corresponding Bayes estimator is
% $$\hat\theta_0^{Bayes}=\argmin_\theta\int L(\theta,\theta_0)p(\hat\theta_0|\theta_0)\pi(\theta_0)d\theta_0.$$
For the aggregated estimator (\ref{eq: aggm}), suppose the objective function $M: \Omega_\theta\times\Omega_\theta\rightarrow \mathbb R$ used satisfies
\begin{equation}\label{eq: objective_function}
    \int M(\theta,\hat\theta)p(\hat\theta|\theta')d\hat\theta=L(\theta,\theta')+C(\theta'),
\end{equation}
where $L$ is non-negative and $L(\theta, \theta)=0$ for all $\theta$, and $C$ is constant with respect to $\theta$.
Then $L$ is the loss function corresponding to $M$, under which the target estimator is
$$\Theta_0(\bm x_0; L)=\argmin_\theta\int L(\theta,\theta_0)p(\hat\theta_0|\theta_0)\pi(\theta_0)d\theta_0.$$
For example, if the objective function $M$ is the negative log-likelihood function
$M(\theta,\hat\theta)=-\log p(\hat\theta|\theta),$
then the corresponding loss function $L(\theta,\theta')$ is the Kullback-Leibler divergence of the given parameters. 
\begin{thm}\label{thm:m-theta-only}
If for any given $\hat\theta$, $M(\theta, \hat\theta)$ as a function of $\theta$ is convex and second-order differentiable, then the combined estimator $\tilde \theta_0^{(c)}$ using the objective function $M$ converges in probability to the target estimator under the loss function $L$ as $K\rightarrow\infty$:
$$\hat\theta_0^{(c)}=\argmin_\theta \sum_{k=0}^Kw_2(\hat\theta_k,\hat\theta_0)M(\theta,\hat\theta_k)\xrightarrow{\enskip P\enskip}\Theta_0(\bm x_0; L).$$
Furthermore, if $\mathbb E_{\hat\theta_0}[w_2(\hat\theta_k,\hat\theta_0)M'_\theta(\theta_0, \hat\theta)]^2<\infty$ for any fixed $\hat\theta_0$,
$$\sqrt{K}(\tilde\theta_0^{(c)} - \Theta_0)=O_p(1).$$
\end{thm}
The proof is given in Appendix.
\par
The finite second moment conditions in Theorems \ref{thm: theta-theta-only} and \ref{thm:m-theta-only} are satisfied in most cases. Both Theorems \ref{thm: theta-theta-only} and \ref{thm:m-theta-only} assume an accurate estimation of the weight $w_2(\hat\theta_k, \hat\theta_0)$ (with an error rate smaller than $O_p(K^{-1/2})$. 
With the accurate weights $w_2(\hat\theta_k, \hat\theta_0)$, both iGroup estimators have faster convergence rates to the target estimator $\Theta_0$ than the nonparametric one in Theorems \ref{thm:hdkernel}. 
\par
When no accurate estimations for $w_2(\hat\theta_k, \hat\theta_0)$ are feasible, we proposed an approximate estimator for $w_2(\hat\theta_k,\hat\theta_0)$ in Section \ref{sec:evaluate-weight}, using a set of bootstrap samples $(\hat\theta_k^{(1)},\hat\theta_k^{(2)})$ for $k=0,\dots, K$. 
When $\bm z$ is not available, 
the integral $\int p(\hat\theta_k|\theta)p(\hat\theta_0|\theta)\pi(\theta)d\theta$ can be estimated by a kernel density estimator in a lower dimensional space:
% $$\int p(\hat\theta|\theta)p(\hat\theta'|\theta)\pi(\theta)d\theta\approx \dfrac{1}{K}\sum_{k=1}^K\mathcal K_1\left(\dfrac{|\hat\theta_k^{(1)}-\hat\theta|}{b_1}\right)\mathcal K_2\left(\dfrac{|\hat\theta_k^{(2)}-\hat\theta'|}{b_2}\right),$$
$$ \dfrac{1}{K+1}\sum_{j=0}^K\mathcal K_1\left(\dfrac{|\hat\theta_j^{(1)}-\hat\theta_k|}{b_1}\right)\mathcal K_2\left(\dfrac{|\hat\theta_j^{(2)}-\hat\theta_0|}{b_2}\right),$$
where $\mathcal K_1$ and $\mathcal K_2$ are two kernel functions with $b_1$, $b_2$ the corresponding bandwidths. The bootstrap estimation of the weight $w_2(\hat\theta_k, \hat\theta_0)$ has a nonparametric error rate $O_p(K^{-1/(d'+2})$, where $d'$ is the dimension of $\theta_0$. This inaccuracy gives rise to the final error rate in Theorem \ref{thm: theta-theta-only} and \ref{thm:m-theta-only} such that for $\hat\theta_0^{(c)}$ (or $\tilde \theta_0^{(c)}$) constructed based on $\hat w_2(\hat\theta_k, \hat\theta_0)$ with error rate $O_p(K^{-1/(d'+2)})$, 
$\hat\theta_0^{(c)} - \Theta_0(\bm x_0;\ell_2) = O_p(K^{-1/(d'+2)})$ and $\tilde\theta_0^{(c)} - \Theta_0(\bm x_0;L) = O_p(K^{-1/(d'+2)})$. Both are slower than $O_p(K^{-1/2})$.

\par
The performance of the target estimator $\Theta_0(\bm x_0;\ell_2)$ in estimating $\theta_0$ strongly depends on the accuracy of individual level $\hat\theta_k$. Define the bias and variance of the target estimator $\Theta_0(\bm x_0;\ell_2)=\mathbb E_\pi[\theta_0\mid \hat\theta_0]$ by 
\begin{equation}
B_0(\theta_0) = \mathbb E_{\theta_0}[\Theta_0(\bm x_0;\ell_2)] - \theta_0,\quad
V_0(\theta_0) = \text{Var}_{\theta_0}[\Theta_0(\bm x_0;\ell_2)].
\label{eq: bias-variance-define-theta-only}
\end{equation}
Suppose $\hat\theta_0 = \theta_0 + \zeta_0$ with $\mathbb E[\zeta_0] = 0$ and $\mathbb E[\zeta_0^2] = \sigma_\theta^2$. Similar to Theorem \ref{thm: intrinsic_asymp}, $B_0$ and $V_0$ are of order $\sigma_\theta^2$ when $\sigma_\theta^2\rightarrow 0$.
\begin{thm}
\label{thm: intrinsic_asymp_theta_only}
Suppose $\zeta_0$ has finite higher moments. Then, when $\sigma_\theta^2\rightarrow 0$, the bias and variance of the target estimator $\Theta_0(\bm x_0;\ell_2)$ with respect to a fixed $\theta_0$ are
$$B_0\asymp \sigma_\theta^2,\quad V_0\asymp \sigma_\theta^2,$$
where $B_0$ and $V_0$ are defined in (\ref{eq: bias-variance-define-theta-only}).
\end{thm}
The proof is provided in Appendix.
\par
When $\hat\theta_0$ is exact such that $\sigma_\theta=0$, the target estimator equals to the true parameter $\theta_0$ as the weight function $w_2(\hat\theta_k, \hat\theta_0)$ assigns zero weight for all other individuals except individual 0. 
Similar results hold for the target estimator $\Theta_0(\bm x_0;L)$.

\subsection{Case 3: The complete case}\label{sec:complete-case}
When both $\hat \theta$ and $\bm z$ are available and reasonably accurate, we should use both information to improve the inference via grouping. Assuming $\hat\theta$ is sufficient for $\theta_0$, the target estimator %of $\theta_0$ 
is $\Theta_0(\bm x_0, \bm z_0;\ell_2) = \mathbb E_\pi[\theta_0\mid\hat\theta_0, \bm z_0]$ under squared loss and $\Theta_0(\bm x_0, \bm z_0; L) = \argmin_\theta \mathbb E_\pi[L(\theta, \theta_0)\mid\hat\theta_0, \bm z_0]$ under other loss function $L$. 
The following results are based on a combination of both information.
\begin{thm}\label{thm:theta-complete}
Suppose $\hat\theta_k$ is a sufficient and unbiased estimator for $\theta_k$, and $\hat\theta_0^{(c)}$ is a combined estimator as in (\ref{eq: aggtheta}) with the weight functions (\ref{eq: weight0}), (\ref{eq: weight1}) and (\ref{eq: weight}), where $\mathcal K(\cdot)$ is a kernel function satisfying Assumption \ref{thm:condition}. Then under Assumptions (\ref{assump:dense}) and (\ref{assump:smooth})
$$\hat\theta_0^{(c)}\rightarrow \Theta_0(\bm x_0, \bm z_0;\ell_2)\quad \text{in probability}.$$
With the optimal bandwidth $\hat b$ chosen to be $\hat b\asymp K^{1/(d+4)}$, the optimal mean squared error is $\mathbb E[\hat\theta_0^{(c)}-\Theta_0]^2\asymp K^{-4/(d+4)}$.
\end{thm}
The proof is given in Appendix.
\par
% With observing both $\bm z_0$ and $\hat\theta_0$, the Bayes estimator under the loss function $L:\Omega_\theta\times\Omega_\theta\rightarrow\mathbb R$ is
% \begin{equation}\label{eq: complete-bayes}
%     \hat\theta^{Bayes}_0=\argmin_\theta\int L(\theta,\theta_0)p(\hat\theta_0|\theta_0)\pi(\theta_0| \bm z_0)d\theta_0.
% \end{equation}

Let $M(\theta, \hat\theta)$ be the corresponding objective function as defined in (\ref{eq: objective_function}). We have that the aggregated estimator (\ref{eq: aggm}) based on the objective function $M(\theta, \hat\theta)$ converges to the target estimator $\Theta_0(\bm x_0, \bm z_0; L)$ as shown in the following Theorem \ref{thm:m-complete}.

\begin{thm}\label{thm:m-complete}
If for any given $\hat\theta$, $M(\theta, \hat\theta)$ as a function of $\theta$ is convex and second-order differentiable, then under Assumptions (\ref{assump:dense}) and (\ref{assump:smooth}), the combined estimator $\tilde\theta^{(c)}$ using the objective function $M$ satisfying (\ref{eq: objective_function}) converges to the target estimator:
$$\tilde\theta^{(c)}_0=\argmin_\theta \sum_{k=1}^Kw(\hat\theta_k,\bm z_k; \hat\theta_0, \bm z_0)M(\theta,\hat\theta_k)\xrightarrow{\enskip P\enskip}\Theta_0(\bm x_0, \bm z_0; L).$$
With the optimal bandwidth $\hat b$ chosen to be $\hat b\asymp K^{1/(d+4)}$, the optimal mean squared error is $\mathbb E[\tilde\theta_0^{(c)}-\Theta_0]^2\asymp K^{-4/(d+4)}$.
\end{thm}
The proof is given in Appendix. 
\par
Define the bias and variance of the target estimator $\Theta_0(\bm x_0, \bm z_0;\ell_2)$ as 
\begin{equation}
    B_0(\theta_0) = \mathbb E_{\theta_0}[\Theta_0(\bm x_0, \bm z_0;\ell_2)] -\theta_0,\quad
    V_0(\theta_0) = \text{Var}_{\theta_0}[\Theta_0(\bm x_0, \bm z_0;\ell_2)].
    \label{eq: intrinsic_complete}
\end{equation}
The asymptotic rate of $B_0$ and $V_0$ as $\sigma_\theta^2$ or $\sigma_z^2$ approaches zero is shown in Theorem \ref{thm: intrinsic_asymp_complete_case}.

\begin{thm}
\label{thm: intrinsic_asymp_complete_case}
Suppose $g(\cdot)$ is second order differentiable and $\epsilon_k$ and $\zeta_k$ have finite higher moments.
If $B_0$ and $V_0$ are as defined in (\ref{eq: intrinsic_complete}), then\\
(i) for a fixed $\sigma_z^2$, when $\sigma_\theta^2\rightarrow 0$,
$$B_0\asymp \sigma_\theta^2,\quad V_0\asymp \sigma_\theta^2.$$
(ii) for a fixed $\sigma_\theta^2$, when $\sigma_z^2\rightarrow 0$,
$$B_0\asymp \sigma_z^2,\quad V_0\asymp \sigma_z^2.$$
\end{thm}
The proof is provided in Appendix.
The bias and variance of the target estimator is of the order of the more accurate one between $\bm z_0$ and $\hat\theta_0$. Especially, when either is exact such that $\sigma_z^2=0$ or $\sigma_\theta^2=0$, the target estimator equals the true parameter $\theta_0$. 

\subsection{Further results on risk decomposition}
Let $\hat\theta^{(c)}_0$ be an iGroup estimator as defined in (\ref{eq: aggtheta}) based on information sets $\{\bm z\}$, $\{\hat\theta\}$ or $\{\hat\theta, \bm z\}$ as in Sections \ref{sec:exogenous-variable}, \ref{sec:theta-only} and \ref{sec:complete-case}, respectively. Let $\Theta_0$ be the target estimator in any of the three cases: $\Theta_0(\bm x_0;\ell_2)$, $\Theta_0(\bm z_0;\ell_2)$ or $\Theta_0(\bm x_0, \bm z_0;\ell_2)$, depending on the information set used in $\hat\theta_0^{(c)}$. We have $\hat\theta_0^{(c)}\rightarrow \Theta_0$ in probability.
When both $\hat\theta$ and $\bm z$ are available for all individuals, the overall risk of $\hat\theta^{(c)}_0$ under the prior $\pi(\theta)$  can be decomposed into three components as shown in Proposition \ref{thm:decomposition} as an extension to Proposition \ref{thm: risk-decomposition-simple-l2}.
\begin{prop}\label{thm:decomposition}
% Suppose $\hat\theta^{(c)}_0$ is an iGroup estimator as defined in (\ref{eq: aggtheta}), and $\hat\theta^{(c)}_0 \rightarrow \hat\Theta_0$ in probability when $K\rightarrow\infty$ for every $\hat\theta_0$ and $\bm z_0$. 
% Then 
% $$
% R(\hat\theta^{(c)}_0) = R_{np}(\hat\theta^{(c)}_0) + 
%                         R_{lim}(\hat\theta^{(c)}_0) +
%                         R_{0},
% $$
% where $R(\hat\theta^{(c)}_0)=\mathbb E[(\hat\theta^{(c)}_0-\theta_0)^2]$ is the overall risk of $\hat\theta^{(c)}_0$ under squared loss and prior $\pi(\theta)$, and 
% $$
% R_{np}(\hat\theta^{(c)}_0)= \mathbb E[(\hat\theta^{(c)}_0-\hat\Theta_0)^2],\quad
% R_{lim}(\hat\theta^{(c)}_0)=\mathbb E[(\hat\Theta_0 -\Theta_0(\bm x_0, \bm z_0;\ell_2))^2],\quad
% R_{0} = \mathbb E[(\Theta_0(\bm x_0, \bm z_0;\ell_2)-\theta_0)^2]
% $$ 
% are the risk components from the nonparametric estimation, deviation of the limit estimator from $\Theta_0(\bm x_0, \bm z_0;\ell_2)$ and the intrinsic risk of $\Theta_0(\bm x_0, \bm z_0;\ell_2)$, respectively.
% \end{prop}
Suppose $\hat\theta^{(c)}_0$ is an iGroup estimator as defined in (\ref{eq: aggtheta}) with the target estimator $\Theta_0$. 
Then 
$$
R(\hat\theta^{(c)}_0) = R_{np}(\hat\theta^{(c)}_0) + 
                        R_{target}(\Theta_0),
$$
where $R(\hat\theta^{(c)}_0)=\mathbb E[(\hat\theta^{(c)}_0-\theta_0)^2]$ is the overall risk of $\hat\theta^{(c)}_0$ under squared loss and prior $\pi(\theta_0)$, and 
$$
R_{np}(\hat\theta^{(c)}_0)= \mathbb E[(\hat\theta^{(c)}_0-\Theta_0)^2],\quad
R_{target}(\Theta_0)=\mathbb E[(\Theta_0 -\theta_0)^2]
$$ 
are the risk components from the nonparametric estimation and the target estimator itself, respectively.\\
Furthermore, assuming both $\bm x$ and $\bm z$ are available, for $\Theta_0=\Theta_0(\bm x_0;\ell_2)$ or $\Theta_0=\Theta_0(\bm z_0;\ell_2)$, which only uses partial information, we have
$$
R_{target}(\Theta_0) = R_{inf}(\Theta_0) + R_0,
$$
where $R_{inf}(\Theta_0)= \mathbb E[(\Theta_0 -\Theta_0(\bm x_0, \bm z_0;\ell_2))^2]$ is the risk premium resulting from using partial information, and $R_0 = \mathbb E[(\Theta_0(\bm x_0, \bm z_0;\ell_2)-\theta_0)^2]$ is the overall risk of $\Theta_0(\bm x_0, \bm z_0;\ell_2)$.
\end{prop}
The proof is provided in Appendix.
\par
The decomposition in Proposition \ref{thm:decomposition} reveals a guideline to optimize the iGroup estimator. The overall risk of iGroup estimator $\hat\theta_0^{(c)}$ can be decomposed into two parts: one from the nonparametric estimation of the target estimator and the other from the risk of the target estimator itself. The risk component $R_{np}$ involves the bandwidth $b$ in the nonparametric estimation. The corresponding optimal bandwidth is chosen as in a high-dimensional kernel smoothing problem (see Theorems \ref{thm:hdkernel}, \ref{thm: overall-risk} and \ref{thm:theta-complete}), since the bandwidth does not appear in the other risk terms. 
\par
The risk component $R_{target}$ evaluates the performance of the target estimator. Different choices in constructing iGroup weight correspond to different $\Theta_0$'s.
Such difference is revealed by decomposing $R_{target}$ into two parts: $R_{inf}$ is the risk term arising from using partial information and $R_0$ is the risk of the target estimator $\Theta_0(\bm x_0, \bm z_0;\ell_2)$, which incorporates the full information set.
Since $R_{inf}$ obtains its minimum at $\Theta_0 = \Theta_0(\bm x_0, \bm z_0;\ell_2)$, it is always (asymptotically) optimal to use the full information set $\{\hat\theta, \bm z\}$ in grouping, if both are available as in the complete case. On the other hand, if $\hat\theta$ (or $\bm z$) is extremely noisy such that $\Theta_0=\mathbb E_\pi[\theta_0\mid\bm z_0]\approx \mathbb E_\pi[\theta_0\mid \hat\theta_0, \bm z_0]$ (or $\Theta_0=\mathbb E_\pi[\theta_0\mid\hat\theta_0]\approx \mathbb E_\pi[\theta_0\mid \hat\theta_0, \bm z_0]$, respectively), it is more practical to use $\bm z$ only (or $\hat\theta$ only, respectively) for grouping, since it will have similar performance but less computational cost, and finite sample variation.
%Note that, grouping with respect to $\hat\theta$ only has a faster convergence rate than the complete case as stated in Theorems \ref{thm: theta-theta-only} and \ref{thm:theta-complete}. However, such a faster convergence rate is achieved at the cost of the increase in the overall risk. 
\par
The last risk component $R_0$ is the minimum overall risk one can achieve. In our approach, such a minimum risk can be asymptotically reached when both $\hat\theta$ and $\bm z$ are included in grouping and the number of individuals $K$ approaches infinity. When $\hat\theta$ or $\bm z$ is exact, $\Theta_0(\bm x_0, \bm z_0;\ell_2) = \mathbb E_\pi[\theta_0|\hat\theta_0, \bm z_0] = \theta_0$ and $R_0$ is 0. In this case, all iGroup estimators in (\ref{eq: aggtheta}) converges to $\theta_0$. The three risk components of different iGroup models are compared in Table \ref{table: comparison}. Note that the rate of $R_{np}$ for Case 2 assumes an accurate evaluation of the weight function $w_2(\hat\theta_k, \theta_0)$.
\begin{table}[!htpb]
    \centering
    \begin{tabular}{|c|c|l|c|c|}
         \hline
         & \multirow{2}{*}{iGroup Set}&\multicolumn{1}{c|}{\multirow{2}{*}{$R_{np}$}} & \multicolumn{2}{c|}{$R_{target}$}\\ 
         \cline{4-5}
         & & & $R_{inf}$ & $R_0$\\ 
         \hline
         Case 1&$\{\bm z\}$&$\asymp K^{-4/(d+4)}$ & $>0$ & \\
         \cline{1-4}
         Case 2&$\{\hat\theta\}$& $\asymp K^{-1}$\hfill & $>0$ & same value \\
         \cline{1-4}
         Case 3&$\{\hat\theta, \bm z\}$&$\asymp K^{-4/(d+4)}$ & $=0$ & \\
         \hline
    \end{tabular}
    \caption{Comparison of the three risk components in different iGroup cases.}
    \label{table: comparison}
\end{table}

\par 
Similar to Proposition \ref{thm:decomposition}, the risk decomposition for the iGroup estimator $\tilde \theta_0^{(c)}$ in (\ref{eq: aggm}) is provided in Proposition \ref{thm:m-decomposition} as an extension to Proposition \ref{thm: risk-decomposition-simple-L}.
% \begin{lem}\label{thm:m-decomposition}
% Suppose the loss function $L$ is as defined in (\ref{eq: objective_function}). The iGroup estimator $\tilde \theta_0^{(c)}$ is defined in (\ref{eq: aggm}), and $\tilde\theta^{(c)}_0 \rightarrow \tilde\Theta_0$ in probability when $K\rightarrow\infty$ for every $\hat\theta_0$ and $\bm z_0$. If $L(\hat\theta, \theta)$ is second-order partially differentiable with respect to $\hat\theta$ such that $L'(\hat\theta, \theta) = \partial L/\partial \hat\theta$ and $L''(\hat\theta, \theta) = \partial^2 L/\partial \hat\theta^2$, then
% $$\tilde R(\tilde\theta_0^{(c)}) = \tilde R_{np}(\tilde\theta_0^{(c)})
% +\tilde R_{lim}(\tilde\theta_0^{(c)}) + \tilde R_0 + o(\mathbb E[(\tilde\theta_0^{(c)} - \tilde\Theta_0)^2]),
% $$
% where $\tilde R(\tilde\theta_0^{(c)}) = \mathbb E[L(\tilde\theta_0^{(c)}, \theta_0)]$ is the overall risk of $\tilde\theta_0^{(c)}$ under loss $L$ and prior $\pi(\theta)$, and 
% $$
% \tilde R_{np}(\tilde\theta_0^{(c)})=\frac{1}{2}\mathbb E[L''(\tilde\Theta_0, \theta_0)(\tilde\theta_0^{(c)}-\tilde\Theta_0)^2],\ 
% \tilde R_{lim}(\tilde\theta_0^{(c)})=\mathbb E[L(\tilde\Theta_0, \theta_0)] - \tilde R_0,\ 
% \tilde R_0=\mathbb E[L(\Theta_0(\bm x_0, \bm z_0; L), \theta_0)]
% $$
% are the risk components from the nonparametric estimation, deviation of $\tilde\Theta_0$ from $\Theta_0(\bm x_0, \bm z_0;L)$ and the intrinsic risk of $\Theta_0(\bm x_0, \bm z_0;L)$, respectively. 
% \end{lem}
\begin{prop}\label{thm:m-decomposition}
Suppose the loss function $L$ is as defined in (\ref{eq: objective_function}). The iGroup estimator $\tilde \theta_0^{(c)}$ is defined in (\ref{eq: aggm}) with the target estimator $\Theta_0$. If $L(\hat\theta, \theta)$ is second-order partially differentiable with respect to $\hat\theta$ such that $L'(\hat\theta, \theta) = \partial L/\partial \hat\theta$ and $L''(\hat\theta, \theta) = \partial^2 L/\partial \hat\theta^2$, then
$$\tilde R(\tilde\theta_0^{(c)}) = \tilde R_{np}(\tilde\theta_0^{(c)})
+\tilde R_{target}(\Theta_0) + o(\mathbb E[(\tilde\theta_0^{(c)} - \Theta_0)^2]),
$$
where $\tilde R(\tilde\theta_0^{(c)}) = \mathbb E[L(\tilde\theta_0^{(c)}, \theta_0)]$ is the overall risk of $\tilde\theta_0^{(c)}$ under loss $L$ and prior $\pi(\theta)$, and 
$$
\tilde R_{np}(\tilde\theta_0^{(c)})=\frac{1}{2}\mathbb E[L''(\Theta_0, \theta_0)(\tilde\theta_0^{(c)}-\Theta_0)^2],\ 
\tilde R_{target}(\Theta_0)=\mathbb E[L(\Theta_0, \theta_0)],
$$
are the risk components from the nonparametric estimation of the target estimator and the target estimator itself, respectively. \\
Furthermore,  assuming both $\bm x$ and $\bm z$ are available, for any $\Theta_0=\Theta_0(\bm z_0;L)$ or $\Theta_0=\Theta_0(\bm x_0;L)$, which only uses partial information, we have
$$\tilde R_{target}(\Theta_0) = \tilde R_{inf}(\Theta_0) + \tilde R_0,$$
where $\tilde R_0 = \mathbb E[L(\Theta_0(\bm x_0, \bm z_0;L), \theta_0)]$ is the overall risk of $\Theta_0(\bm x_0, \bm z_0;L)$ and $\tilde R_{inf}(\Theta_0)= \mathbb E[L(\Theta_0, \theta_0)] - \tilde R_0$ is the risk premium resulting from using partial information.
\end{prop}
The proof is given in Appendix.

\subsection{Bandwidth selection and other practical guide }\label{sec:bandwidth-selection}
For real applications, the bandwidth $b$ in the weight function (\ref{eq: weight1}) remains to be tuned. Ideally one would perform bandwidth selection to the target individual $\theta_0$. However, cross validation cannot be implemented to determine $b$ with only one estimator $\hat\theta_0^{(c)}$ for a single individual. Instead, we consider a set $\Omega_0$ around target individual 0 such that the bandwidth $b$ is tuned to minimize the averaged risk over $\Omega_0$.

% The goal is to find an optimal bandwidth $b$ such that the overall risk is minimized. For this optimization problem, leave-one-out cross validation can be implemented to determine the optimal bandwidth. . However, we still wish to adopt bandwidth selection to $\theta_0$. To determine the optimal bandwidths, we consider a set $\Omega_0$ around target 0. 
\par 
When $\Omega_0$ is chosen as the full set $\{1, 2, \dots, K\}$, it is the global bandwidth selection scheme that usually used in kernel smoothing and machine learning. However, the bandwidth selected by such global optimization is not optimal for the particular target individual 0. A cross validation set $\Omega_0$ localized to individual 0 is more appreciated to tune this individualized local bandwidth. When tuning the bandwidth in $w_1$ over $\bm z_k$'s, such a set $\Omega_0$ can be constructed based on $\bm z_0$ such as $\Omega_0(\bm z_0, \epsilon) =\{k\in \{1, \dots, K\}: \|\bm z_0 - \bm z_k\|\leqslant \epsilon\}$.
\par
Suppose $\hat\theta_k$'s are available and the individual estimators are aggregated to form an iGroup estimator as described in (\ref{eq: aggtheta}). The goal is to choose a bandwidth $b$ that minimizes the local risk function over $\Omega_0$ (under squared loss) around $\theta_0$
$$R_{\Omega_0}(b) = \mathbb E\left[\dfrac{1}{|\Omega_0|}\sum_{k\in\Omega_0} (\hat\theta_k^{(c)}-\theta_k)^2\right].$$
The cross-validation error we use is computed as 
$$CV_{\Omega_0}(b) = \dfrac{1}{|\Omega_0|}\sum_{k\in\Omega_0} \left(\hat\theta_{(-k)}^{(c)}-\hat\theta_k\right)^2,$$
where $\hat\theta_{(-k)}^{(c)}$ is the leave-one-out estimator defined by
\begin{equation}\label{eq: loo}
\hat\theta_{(-k)}^{(c)} = \dfrac{\sum_{l\neq k}\hat\theta_l w(l;k)}{\sum_{l\neq k} w(l;k)}.
\end{equation}
It is worth to point out that although the cross validation set $\Omega_0$ is localized/individualized, the leave-one-out estimators (\ref{eq: loo}) still utilize all individuals instead of limited to $\Omega_0$.
\par
It is seen in Proposition \ref{thm:lemma} that the leave-one-out cross-validation can estimate the local risk over $\Omega_0$ up to a constant and hence be useful. 
% However, the cross validation process is capable to find such a bandwidth $b^*$ that minimizes the cross-validation error $CV(b)$ rather than minimizing the infeasible risk function $R_K(b)$. The reason is that under certain conditions, the cross-validation error function is an unbiased estimation of the risk function up to a constant as $K\rightarrow\infty$.
\begin{prop}\label{thm:lemma}
Suppose $\hat\theta_k$ is an unbiased estimator for $\theta_k$ for all $k=1,\dots, K$ and the weight function $w(l;k)$ satisfies
\begin{equation}
\dfrac{w(k;k)}{\sum_{l\neq k} w(l;k)}=O\left(\dfrac{1}{K}\right).\label{eq: wcond}
\end{equation}
Then 
$$\mathbb E[CV_{\Omega_0}(b)]=R_{\Omega_0}(b) + C_{\Omega_0} +O\left(\dfrac{1}{K}\right),$$
where $C_{\Omega_0}$ is related to $\Omega_0$ but is a constant with respect to $b$.
\end{prop}
The proof is given in Appendix.\\
% The proof is provided in the appendix. 
\noindent\textbf{Remark I:}
A sufficient condition for the weight function to satisfy (\ref{eq: wcond}) is that the function is bounded. With bounded weights, we have 
$$\dfrac{w(k;k)}{\sum_{l\neq k} w(l;k)}\rightarrow \dfrac{w(k;k)}{K\mathbb E w(\cdot;k)}=O\left(\dfrac{1}{K}\right).$$
Common kernels such as the boxed, Gaussian and Epanechnikov kernels satisfy this condition. Our choice of weight function $(\ref{eq: weight})$ with a bounded kernel $\mathcal K$ satisfies the condition as well.\\
\textbf{Remark II:} Similar results hold for aggregating objective functions (\ref{eq: aggm}) as long as the objective function is convex and second-order differentiable, and a Taylor series expansion is available.
\par
Beside the theoretical discussions on iGroup's asymptotic performance, there are many other factors that may affect the accuracy in real applications with finite number of individuals. First of all, the weight component $w_2(\cdot)$ is estimated from bootstrapped samples. It lowers the convergence rate since bootstrapped samples from finite population are usually correlated. Secondly, computing the full weight function requires a kernel density estimation in a high dimensional space. When $K$ is finite, aggregating individuals with weights evaluated directly from a high dimensional space suffers from the lack of sample size. It often requires some feature selection procedures to reduce the dimension. 
\par
Therefore, when the weight estimation is not accurate and when the sample size is limited, the complete case may not be the best choice. In real application, we suggest using (local) cross-validation to tune the bandwidth and to choose the most appropriate weight formulation.

% \newpage
\section{Simulations}\label{sec:simulation}
\subsection{iGroup with noisy exogenous variables (Case 1 in Section \ref{sec:exogenous-variable})}\label{sec:simulation-noisy-z}
In this example, the performance of using an exogenous variable $z$ in iGroup is studied. Suppose, for each individual, the true parameter $\theta$ is a quadratic function of $\eta$:
$$\theta_k =g(\eta_k) =  (\eta_k+1)^2.$$
The relationship is set to a quadratic form because a continuous function of $z$ can be approximated by a quadratic function within a small enough neighborhood of $z_0$.  A population of size $K=1000$ is generated with their $\eta_k$'s following a Gaussian distribution $N(0.2,1)$. For each individual $k$, let $\hat\theta_k$ be a sufficient unbiased estimator of $\theta_k$ using $\bm x_k$ such that $\hat\theta_k$ is directly generated with error $\epsilon\sim N(0,\tau^2=1)$ and there is no need to generate $\bm x_k$ explicitly. $z_k$ is a noisy observation of $\eta_k$ such that $z_k\sim N(\eta_k, \sigma^2)$. 
\par
More specifically, the dataset is generated by the following hierarchical structure.
\begin{align*}
    \eta_k \sim N(0.2, 1),\quad
    \theta_k  = (\eta_k + 1)^2,\quad
    \hat\theta_k \sim N(\theta_k, 1),\quad
    z_k \sim N(\eta_k, \sigma^2),
\end{align*}
for $k=1,\dots, K$. 
%Suppose the bootstrapped versions of $\hat\theta_k$ are infeasible and hence $w_2(\cdot)$ is not available, an iGroup estimator can only be constructed by using $z_k$ for grouping, which falls in the
% We consider this as a special case in Section \ref{sec:exogenous-variable}. 
The estimator in (\ref{eq:estimator}) is used by setting $\mathcal K(\cdot)$ to the Gaussian kernel.
\par
The parameter $\sigma^2$ controls the noise level in the observed $z_k$. Both individualized performance at $\theta_0=1$ and the overall performance over the population are studied at six choices of noise levels $\sigma=0,0.2,0.4,0.6,0.8,1.0$ with 1000 replications each. 
\begin{figure}[!hpb]
\centering
\includegraphics[width=0.8\textwidth]{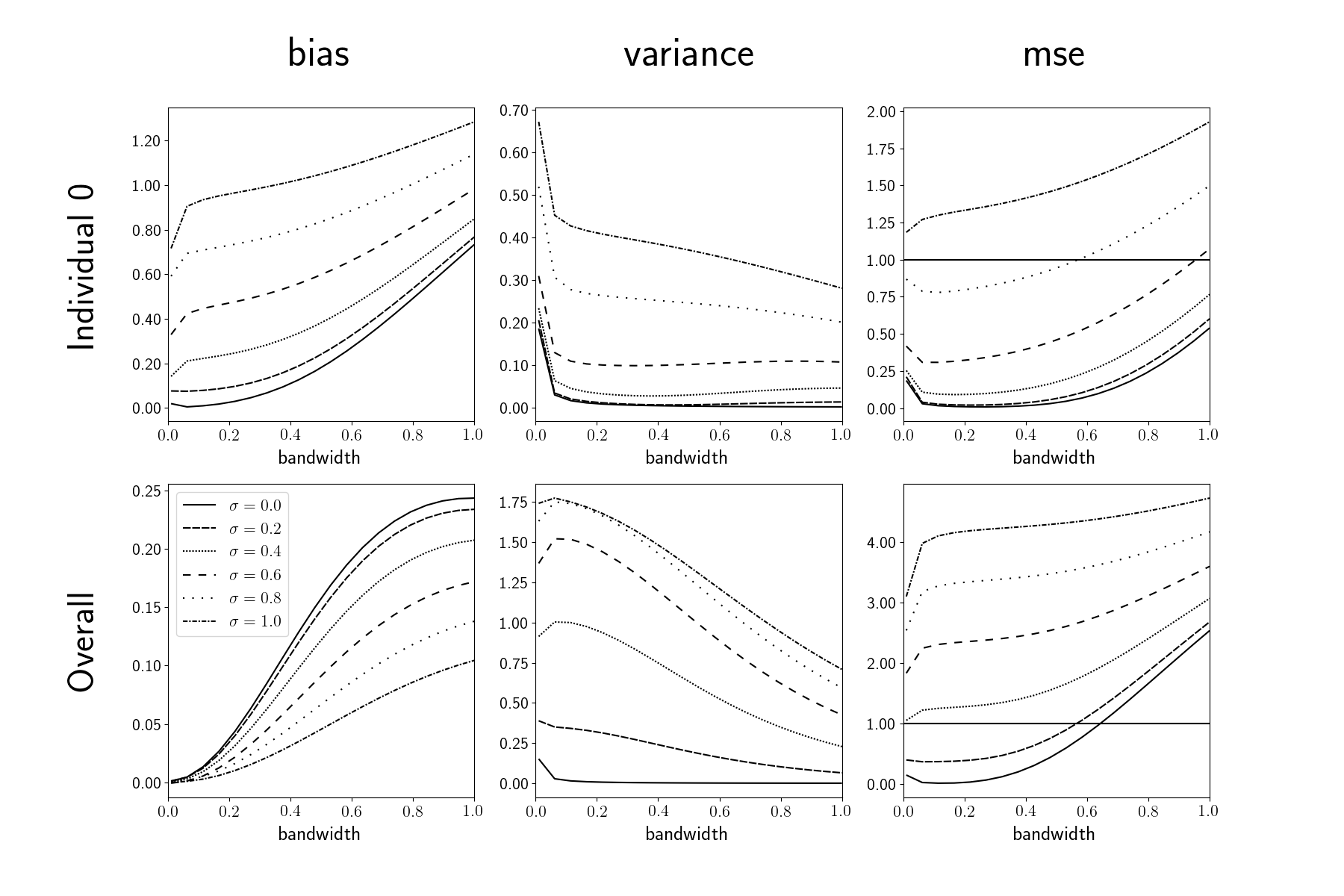}
\caption{Bias, variance and mean squared error as a function of bandwidth under different noise levels for individual 0 (top) and the population (bottom)}\label{fig:noise}
\end{figure}
\par
The in-sample performance of the iGroup estimators are demonstrated in Figure \ref{fig:noise}. The first row shows the bias, variance and mean squared error for the individual at $\theta_0=1$, while the second row plots the overall performance by averaging individual performance over the population. Every curve represents a performance measure (bias, variance or MSE) as a function of the bandwidth $b$ used in weight calculation in (\ref{eq: weight1}) and six different curves distinguish different noise levels $\sigma^2$. 
\par 
From Figure \ref{fig:noise}, it is seen that an increase in the noise level in $\bm z_k$ increases both the bias and variance of the iGroup estimator. When $\sigma>0$, an intrinsic bias is observed for individual 0 when the bandwidth shrinks to zero, while at the population level, the average bias vanishes when the bandwidth shrinks to zero as the iGroup estimator converges to the target estimator $\Theta_0(\bm z_0;\ell_2) = \mathbb E_\pi[\theta_0\mid \bm z_0]$, whose expectation is $\mathbb E_\pi[\theta_0]$. Recall that the individual estimate $\hat\theta_k$ without grouping has a risk $\tau^2=1.0$ by the simulation design. It is marked on the right panels by the horizontal line. When the noise level $\sigma$ exceeds 0.4, both the individual level and population level risk are worse than using $\hat\theta_k$ directly without grouping. Smaller noise in $z_k$ would significantly reduce the risk of the iGroup estimator.

\begin{figure}[!htpb]
\centering
\includegraphics[width=0.6\textwidth]{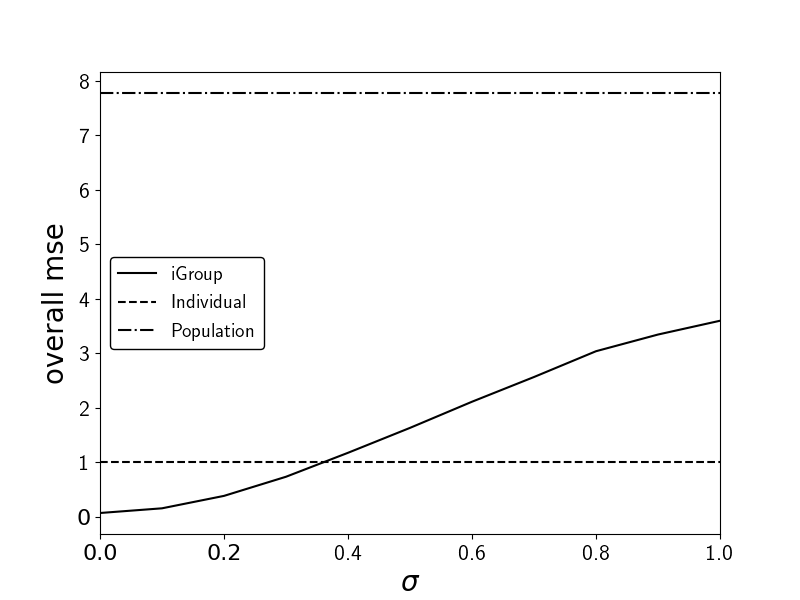}
\caption{Overall MSE of three estimators: individual level, iGroup with cross-validation and population level.}\label{fig:noisecompare}
\end{figure}

 In real applications, the performance plots such as Figure \ref{fig:noise} are not available without knowing the true parameter. As suggested in Section \ref{sec:bandwidth-selection}, an optimal bandwidth can be selected by leave-one-out cross validation. We simply use the global set $\Omega_0=\{1,\dots, K\}$ to tune the bandwidth. Figure \ref{fig:noisecompare} compares the mean square errors of three different estimators under different noise level settings for $\sigma^2$. The individual level estimator uses $\hat\theta_k$, which achieves a constant MSE at $\tau^2=1$. The population level estimator uses the averaged estimator $(\sum_{k=1}^K\hat\theta_k)/K$, assuming population homogeneity. The iGroup estimator uses the estimator ($\ref{eq:estimator}$) and selects the optimal bandwidth by leave-one-out cross validation over a grid of bandwidths. The population level estimator is always the worst because the homogeneity population assumption is invalid in this simulation. The overall MSE of the iGroup estimator is a monotone increasing function of the noise level $\sigma$, because the intrinsic bias and variance increase with $\sigma$. The iGroup estimator outperforms the individual estimator when $\sigma$ is below the threshold $\sigma=0.35$. It also suggests that the iGroup method works better when more accurate exogenous variable $z$ is used. 

\subsection{Short time series (Case 2 in Section \ref{sec:theta-only})}\label{sec:simulation-time-series}
In this simulation study, the individualized grouping learning method is applied to a set of short time series without any exogenous information. It is a simulation study for Case 2 in Section \ref{sec:theta-only}. Suppose we have $K=200$ time series following an AR(1) model. Their AR coefficients $\theta_1,\dots,\theta_{200}$ are drawn randomly from a beta-shaped distribution on $[-1,1]$ such that
\begin{equation}
    \dfrac{\theta_k+1}{2} \sim Beta(4,4),\quad k=1,\dots,200.
    \label{eq: ts-prior}
\end{equation}
The length of each time series is 10. They are generated from their stationary distributions:
\begin{align*}
x_{k,0}&\sim N\left(0,\dfrac{\sigma^2}{1-\theta_k^2}\right),\\
x_{k,t} &=\theta_kx_{k,t-1}+\epsilon_{k,t}, \quad k=1,\dots, 200,\ t=1,\dots, 10,
\end{align*}
where $\epsilon_{k,t}\sim N(0,\sigma^2)$ and $\sigma=3$.
\par
Four estimators are used and their mean squared errors averaged over the 200 individual time series are compared. The individual level estimator is based on each time series of 10 observations and does not borrow any information from the others. It is an unbiased estimator for each individual.
The iGroup1 estimator aggregates the log-likelihood functions according to (\ref{eq: aggm}), where the weight function used is (\ref{eq: weightreduce}), which is estimated by bootstrap samples. The bootstrap estimates are obtained based on multinomial samples of $(x_{t-1}, x_t)$ pairs for each individual. The bandwidth used in estimating $w_2(\hat\theta_k, \hat\theta_0)$ in (\ref{eq: weightreduce}) is chosen by cross-validation as in a kernel density estimation problem.
The iGroup2 estimator aggregates individual level estimators by the weight function in Equation (\ref{eq: weightreduce}), the same weight function as in the iGroup1 estimator. These three methods do not utilize the true prior distribution. The fourth estimator, the oracle one, uses the posterior mean as the estimator with the true population prior (\ref{eq: ts-prior}) as the prior. The oracle estimator, which is the best point estimator for $\theta_0$ given the prior information $\pi(\cdot)$, is the target estimator $\Theta_0(\bm x_0;\ell_2)$ for iGroup methods.

\begin{figure}[!htp]
\centering
\includegraphics[width=0.9\textwidth]{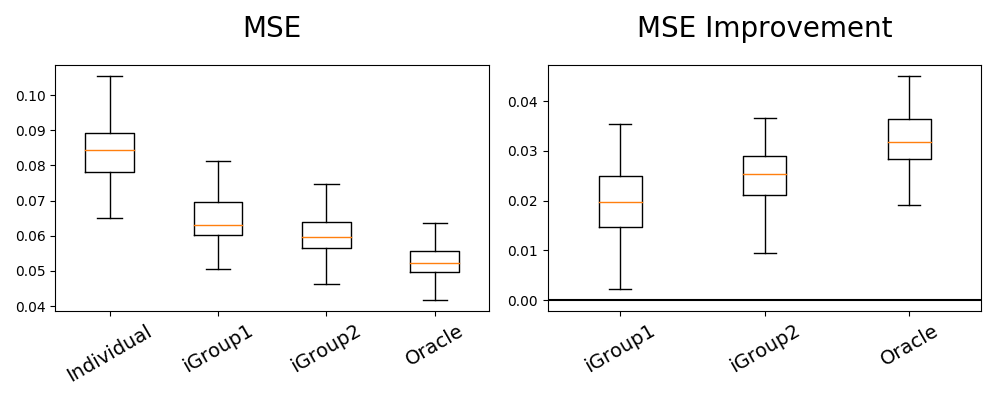}
\caption{Comparison of the averaged MSE over 200 individuals on 100 replications for four estimators}\label{fig:mse}
\end{figure}

\par
The simulation (including generating the data) is repeated 100 times. The box plots of the mean squared errors of the four estimators are reported in the left panel of Figure \ref{fig:mse}. On average, the iGroup1 and iGroup2 estimators achieve smaller mean squared errors and smaller variances compared with the individual one. The oracle estimator is the best among those four with the smallest average error and variation. The iGroup estimators are quite close to the oracle one. The slight worse performance is due to the approximation error when constructing the weight functions. Between the two iGroup estimators, iGroup2 is slightly better than iGroup1 because the loss function used in iGroup2 is the squared loss, whose overall risk is minimized by aggregating $\hat\theta_k$ (See Theorem \ref{thm: theta-theta-only}).
\par
The right panel in Figure \ref{fig:mse} plots the improvement (difference) of the mean square errors of the iGroup estimators and the oracle estimator over the individual estimator for the 100 replications. It shows that in all experiment replications, the mean square errors of the iGroup estimators are uniformly better than the individual one. Estimation does benefit from individualized grouping in this case. 

\subsection{A combined case (Case 3 in Section \ref{sec:complete-case})}\label{sec:simulation-complete-case}
In this simulation, we compare the performance of different iGroup estimators constructed on different information sets when both $\hat\theta$ and $z$ are available as in Case 3 discussed in Section \ref{sec:complete-case}. Consider a population with $n=1024$ individuals following:
\begin{align*}
\eta_k \sim N(0, 1),\quad
\theta_k=\sin (\pi\eta_k),\quad
z_k \sim N(\eta_k, \sigma^2),\quad
x_{k,1}, x_{k,2},\dots, x_{k, n}\sim N(\theta_k, \sigma_x^2),
\end{align*}
for $k=1,\dots, 1024$. $\theta$ is the parameter of interest. Individual estimator used is
$$\hat\theta_k = \dfrac{1}{n}\sum_{i=1}^n x_{k, i}\text{ for }k=1,\dots, 1024.$$
Four approaches are investigated here as special cases of the iGroup method. iGroup($\emptyset$) is the individual estimation without grouping, i.e. using $\hat\theta_k$ as the estimator.  iGroup($z$) uses the exogenous observation $z$ only for grouping and an iGroup estimator is obtained by aggregating $\hat\theta$'s using $w_1(\bm z_k, \bm z_0)$ in (\ref{eq: weight1}), where the bandwidth $b$ is selected by leave-one-out cross validation. iGroup($\hat\theta$) uses $\hat\theta_k$ only for grouping, using $w_2(\hat\theta, \hat\theta')$ in (\ref{eq: weight}) as the weight function. The weight is approximated by kernel density estimation on the bootstrapped samples with bandwidth selected by cross validation. And lastly, iGroup($z$, $\hat\theta$) uses both $z$ and $\hat\theta$ for calculating the weight function $w(\bm z_k, \hat\theta_k; \bm z_0, \hat\theta_0)$ in (\ref{eq: weight0}) as discussed in Section \ref{sec:complete-case}, with the bandwidth selected by leave-one-out cross validation.

\begin{table}[!htp]
\centering
\begin{tabular}{|c||c|c|c||c|c|c|c|}
\hline
$Configuration$ & $n$ & $\tau^2 = \sigma_x^2/n$ & $\sigma$ & iGroup($\emptyset$) & iGroup($\hat\theta$) & iGroup($z$) & iGroup($z$, $\hat\theta$)\\
\hline
1 & 5  & 0.20 & 0.10 & 0.200 & 0.163 & \textbf{0.044} & 0.154\\
2 & 5  & 0.20 & 0.15 & 0.200 & 0.163 & \textbf{0.090} & 0.163\\
3 & 5  & 0.20 & 0.20 & 0.200 & 0.163 & \textbf{0.137} & 0.170\\
4 & 5  & 0.20 & 0.30 & 0.200 & \textbf{0.163} & 0.200 & 0.179\\
5 & 10 & 0.10 & 0.10 & 0.100 & 0.089 & \textbf{0.048} & 0.059\\
6 & 10 & 0.10 & 0.15 & 0.100 & 0.089 & 0.089 & \textbf{0.070}\\
7 & 10 & 0.10 & 0.20 & 0.100 & 0.089 & 0.099 & \textbf{0.077}\\
8 & 10 & 0.10 & 0.30 & 0.100 & 0.089 & 0.100 & \textbf{0.084}\\
9 & 20 & 0.05 & 0.10 & 0.050 & 0.046 & 0.044 & \textbf{0.040}\\
10 & 20 & 0.05 & 0.15 & 0.050 & 0.046 & 0.050 & \textbf{0.044}\\
11 & 20 & 0.05 & 0.20 & 0.050 & 0.046 & 0.050 & \textbf{0.045}\\
12 & 20 & 0.05 & 0.30 & 0.050 & \textbf{0.046} & 0.050 & 0.047\\
\hline
\end{tabular}
\caption{Mean squared error for the experiment in Section \ref{sec:simulation-complete-case} in different configurations.}\label{table: scenario}
\end{table}
\par
Several different $(n,\sigma,\sigma_x)$ configurations are studied. The mean square errors are reported in Table \ref{table: scenario}. The smallest MSE across the different methods is shown in bold face for each configuration. From Table \ref{table: scenario}, it is seen that in Configurations 6 to 11, using both $\bm z$ and $\hat\theta$ outperforms the other three methods. However, it is worth to point out that it is not always the best. When $z$ is relatively accurate and $\hat\theta$ is not so as in Configurations 1, 2, 3 and 5, using $\bm z$ alone is better than involving $\hat\theta$ in the grouping. The reason is that the weight function used in the estimation is an approximation based on bootstrap sampling, which is not accurate when the sample size $n$ is too small (as discussed in Section \ref{sec:bandwidth-selection}). It is also intuitive since using inaccurate $\hat\theta_k$ for grouping may reduce the grouping quality. 
% Another reason is that adding a weight factor on $\hat\theta$ reduces the number of neighborhood points in averaging, which gives rise to the increase in variance. 
When $\bm z$ is quite noisy as in Scenario 4 and 12, using $\hat\theta$ only is better than using the complete information set. Note that when the bandwidth in $w_1(\bm z_k, \bm z_0)$ shrinks to zero, iGroup($z$) reduces to the individual estimator and the complete estimator iGroup($z$, $\hat\theta$) reduces to iGroup($\hat\theta$). However, due to the randomness from finite sample size and possible overfitting, iGroup($\hat\theta$) or iGroup($z$) sometimes performs better.
\par
In conclusion, we suggest the following brief guideline in choosing iGroup models. 
When $\hat\theta$ is relatively inaccurate and the bootstrap method has unignorable error, it is better not to use $\hat\theta$ in grouping. When $\bm z$ is relatively inaccurate, it is better to either use $\hat\theta$ only or use the full model. But when using the full model, the bandwidth needs to be tuned carefully around zero. When both $\hat\theta$ and $\bm z$ are considerably accurate, it is beneficial to consider both in grouping.  

\section{Examples}\label{sec:examples}
\subsection{Value at Risk (VaR) analysis based on Fama-French factors}\label{sec:example-var}
In this example we use iGroup to improve the estimation of Value at Risk in stock returns. Denote the return of stock $k$ in day $t$ as $r_{t,k}$. The one-day value at risk (VaR) of $r_{t,k}$, denoted as $\widehat{VaR}_{t,k}$, is defined as the smallest quantity $v$ such that the probability of the event $r_{t,k}\leqslant -v$ is no greater than a predetermined confidence level $\alpha$ (for example, 1\%). Statistically, $-v$ is the $\alpha$ quantile of $r_{t,k}$. VaR is widely used in quantitative finance and risk management to estimate the possible losses in worse cases (e.g. $1\%$  lower quantile) due to adverse market moves. In practice, it is usually difficult to estimate the value of risk because it requires a large size of data to estimate small quantiles accurately, but the market conditions change over time, which limits the available sample size.
In this application, we consider the daily return of 490 stocks in S\&P 500 for 2016. Three approaches to estimate VaR are compared.

\textbf{Individual VaR estimation using empirical quantiles:} A naive method to estimate VaR is to use the empirical quantile of $r_{t-1, k}, \dots, r_{t-S, k}$. When $\alpha$ is set to be $1\%$ and $S=100$, we have $\widehat{VaR}(t,k) = \min \{r_{t-1,k}, r_{t-2,k}, ..., r_{t-100,k}\}$. Such a quantile estimation is not very accurate. On one hand, when $S$ is small and there is not enough observations, the empirical quantile is not defined. On the other hand, $S$ cannot be very large as the market changes over time and so does the distribution of returns. 

\textbf{Market Level VaR:} The second approach assumes homogeneity among all stocks. The value-at-risk could then be estimated by pooling historical returns of all stocks. In this case, the estimator is
$$\widehat{VaR}(t,k) = Q_{\alpha}\left(\bigcup_{l=1}^{K}\bigcup_{s=1}^{S}\{r_{t-s,l}\}\right),$$
where $Q_{\alpha}(A)$ is the empirical $\alpha$ quantile estimator given a set of observations $A$. Pooling observations from other stocks bring a significant bias if the homogeneity assumption is not valid.

\textbf{iGroup Estimation:} The third approach is an application of the iGroup learning method. Assume on each day, each stock return follows the Fama-French three factor model \citep{fama1993common}: 
\begin{align*}
r_{t,k} &= \alpha_{t,k}+r_f+b_{0,t,k}(MKT_t-r_f)+b_{1,t,k}SMB_t+b_{2,t,k}HML_t+\epsilon_{t,k},\\
\epsilon_{t,k}&\sim \mathcal N(0,\sigma_{k}^2),
\end{align*}
where $MKT$, $SMB$ and $HML$ are the three Fama-French factors, and $b_{0,k,t}$, $b_{1,k,t}$ and $b_{2,k,t}$ are the corresponding coefficients for the stock labeled $k$ at time $t$. The three coefficients characterize stocks by their sensitivity to the corresponding factors. In this model, we assume the Fama-French coefficients $b_0, b_1, b_2$ vary over time slowly. Therefore, the Fama-French coefficients could be used as the exogenous variable $\bm z$ in our iGroup framework. To be more specific, the iGroup estimator is
$$\widehat{VaR}(t,k) = Q_{\alpha}^{(w)}\left(\bigcup_{l=1}^{K}\bigcup_{s=1}^{S}\{(r_{t-s,l}, w(\bm z_{t,l};\bm z_{t,k})) \}\right),$$
where $Q_{\alpha}^{(w)}(\cdot)$ is the empirical $\alpha$ quantile estimator from a weighted sample and $\bm z_{t,k} = (b_{0,t,k}, b_{1,t,k}, b_{2,t,k})$ are the Fama-French coefficients of stock $k$ fitted using the returns in the $S$ days before day $t$. The weight function here is chosen to be a Gaussian kernel
$$w(\bm z_{t,l};\bm z_{t,k})\propto \exp\left(-\dfrac{\|\bm z_{t,l}-\bm z_{t,k}\|_2^2}{2b^2}\right).$$ 
The bandwidth $b$ is the parameter to be tuned. Although the iGroup approach pools all other stocks just as the market level method, it assigns different weights to different stocks based on the similarity of characteristics of the stocks, e.g. the Fama-French coefficients in our case. The market level estimator can be viewed as an extreme case of iGroup estimation when the bandwidth $b$ approaches $\infty$. The individual estimator is another extreme when the bandwidth $b$ shrinks to $0$. 
Note that, the weighted empirical quantile function used in iGroup estimation is equivalent to aggregating the following objective function
\begin{equation*}
M_k(\theta;t)= \sum_{s=1}^S|r_{t-s, k}-\theta|\left(\alpha \bm 1_{\{r_{t-s, k}>\theta\}}+(1-\alpha)\bm 1_{\{r_{t-s, k}\leqslant \theta\}}\right)\label{eq: quantile}
\end{equation*}
by the weight $w_1(\bm z_k, \bm z_0)$ in (\ref{eq: weight1}).
\par 
In this study, we use $\alpha=0.01$, $S=100$, and $K=490$.
The prediction error is measured over 250 trading days in the year 2016 for 490 stocks using\\
$$RMSE = \left[\dfrac{1}{490}\sum_{k=1}^{490}\left(\dfrac{1}{250}\sum_{t=1}^{250}\bm 1_{\left\{r_{t,k}\leqslant \widehat{VaR}(t,k)\right\}}-0.01\right)^2\right]^{1/2},$$
where $\widehat{VaR}(t,k)$ is based on returns $\{r_{t-1, k},\dots, r_{t-100, k}, k=1,\dots, 490\}$.
\begin{figure}[!hpbt]
\centering
\includegraphics[width=.7\textwidth]{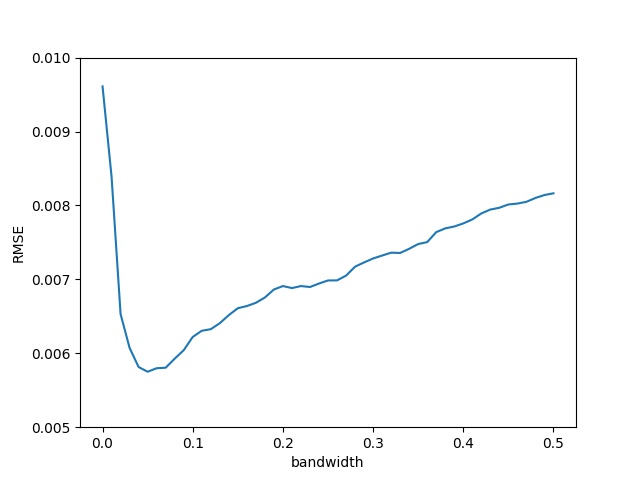}
\caption{Prediction error (RMSE) as a function of bandwidth.}\label{fig: rmse}
\end{figure}
\par
Figure \ref{fig: rmse} shows the RMSE curve as a function of the bandwidth $b$. The bandwidth controls the bias-variance tradeoff. It is seen from the figure that the V-shaped RMSE curve decreases at the beginning and achieves a minimal value at approximately $b=0.05$ with minimum RMSE being $5.75\times 10^{-3}$. The RMSEs of each model are shown in Table \ref{table: rmse}. The iGroup estimator improves the accuracy significantly. 

\begin{table}[!hpt]
\centering
\begin{tabular}{|c|ccc|}
\hline
Method & Individual Estimation & Market Estimation & iGroup Estimation\\
\hline
RMSE & $9.61\times 10^{-3}$ &$1.34\times 10^{-2}$ &$5.75\times 10^{-3}$\\
\hline
\end{tabular}
\caption{Prediction error for three candidate models.}\label{table: rmse}
\end{table}

\subsection{Maritime anomaly detection}\label{sec:example-maritime}
The maritime transportation system is critical to the U.S. and world economy. For security and environmental concerns, it is important to have an efficient detection and risk assessment system for maritime traffic over space and time. Automatic Identification System (AIS) is an automatic tracking system and are mandatory installed on ships such that the maritime information, including GPS location, speed, heading, etc., is reported periodically. The global AIS system receives data from approximately a million ships with updates for each ship as frequently as every two seconds while in motion and every three minutes while at anchor. The data are available at https://marinecadastre.gov/ais/.
\par
In this example, we focused on 534 voyages of tankers and cargo vessels arriving at the Port of Newark between July and November 2014. We investigated their approaching behaviors starting from crossing the 12 nautical mile US territorial sea (TS) boundary to arriving at the port. Two features are considered in this study: the trajectory and the sailing time (duration). The trajectory, treated as an exogenous variable $\bm z$, is a polygonal line consisting of a sequence of reported GPS locations during the approach. The 534 approaching trajectories are plotted in Figure \ref{fig: traj_all} along with the coastlines around the Port of Newark. The sailing time, treated as the observation $x_k$, is the time spent in the approaching procedure starting at the time of entering the 12 nautical miles territorial sea of U.S. and ending at one of the docks in the Port of Newark. Our goal is to identify outliers in sailing time given the trajectory. In this case the parameter of interest is the mean and standard deviation of sailing time, $\theta_k = (\mu_k, \sigma_k)$, such that an outlier can be identified by two standard deviation rule, i.e. individual $k$ is an outlier in time if $|x_k - \hat\mu_k|\geqslant 2\sigma_k$.

\begin{figure}[!htp]
\centering
\includegraphics[width=0.9\textwidth]{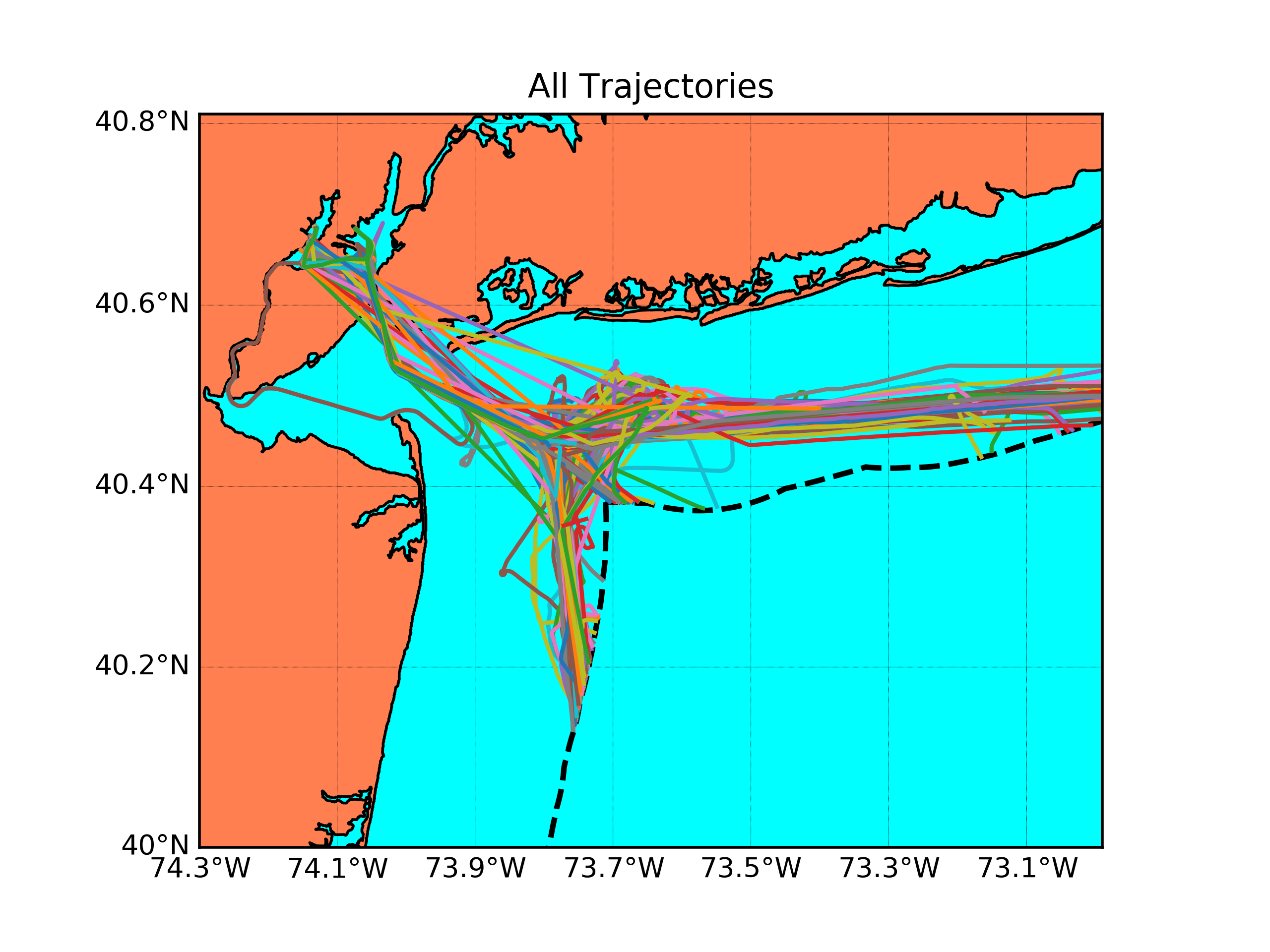}
\caption{All 534 trajectories approaching the Port of Newark}\label{fig: traj_all}
\end{figure}
\par
The trajectory is a functional feature that requires special treatment. Every trajectory consists of a sequence location reports ordered in time. Since the reporting intervals are irregular, it cannot be considered as a 2-dimensional regular time series of equal time intervals. However, since we utilize the trajectory as an exogenous variable $\bm z$ in the iGroup framework, we only need a proper distance/similarity measure defined for any trajectory pairs. Here, we use the dynamic time warping (DTW) distance as the similarity measure. Dynamic time warping is widely used as a similarity measure between two time series for studies in speech recognition and other applications \citep{sakoe1978dynamic, juang1984hidden, nakagawa1988speaker, koenig2008speech}. It finds the optimal monotone one-to-one mapping between two sequences such that the average pairwise distance is minimized. 
\par 
For simplicity, for each individual voyage, we use its nearest 40 neighbors in terms of DTW to form iGoups with equal weight.
Figure \ref{fig:cliques} shows four typical trajectories (top) and their individualized groups identified by its DTW neighbors (bottom). 
\begin{figure}[!hp]
\centering
\includegraphics[width=0.23\textwidth]{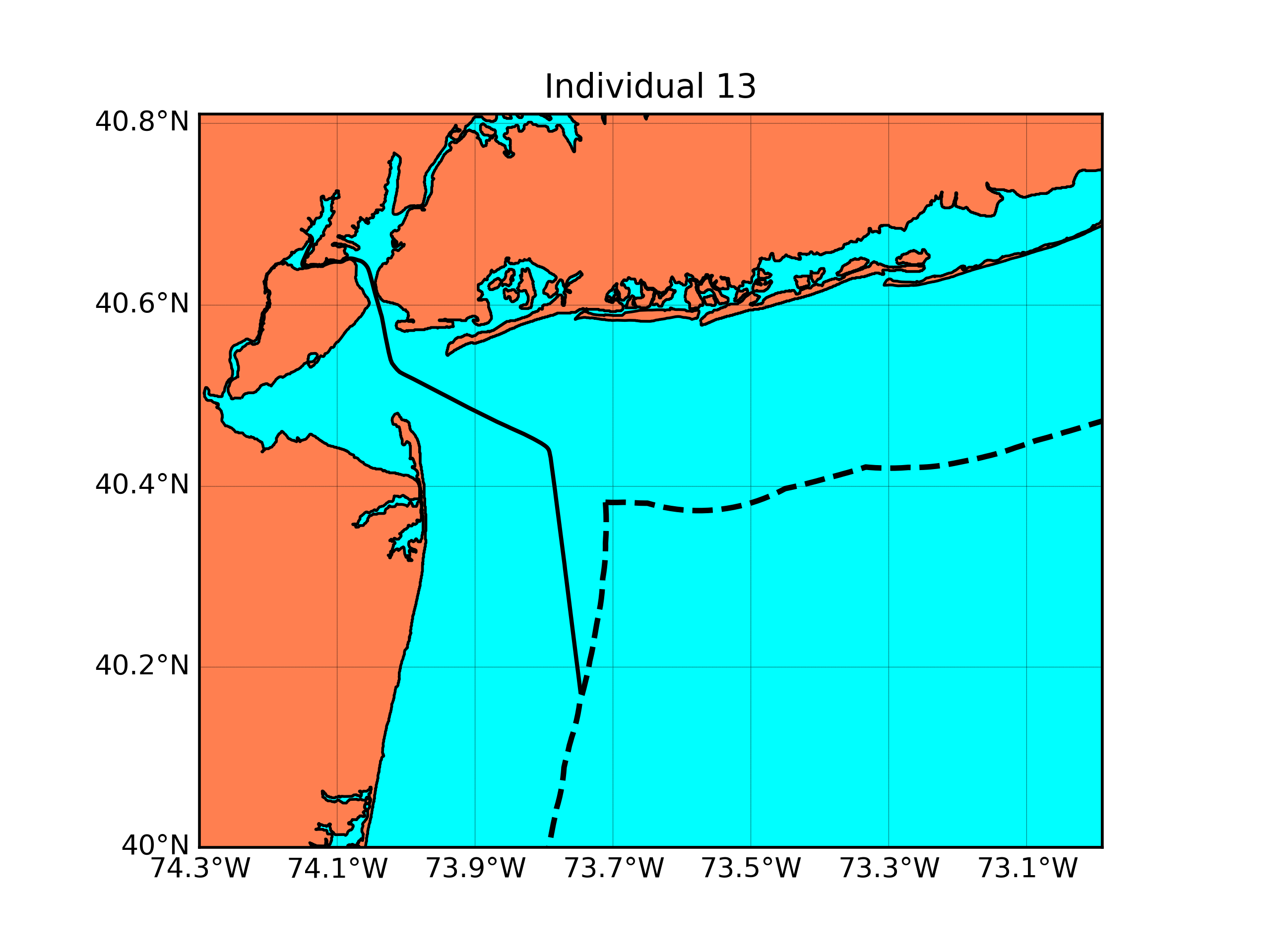}
\includegraphics[width=0.23\textwidth]{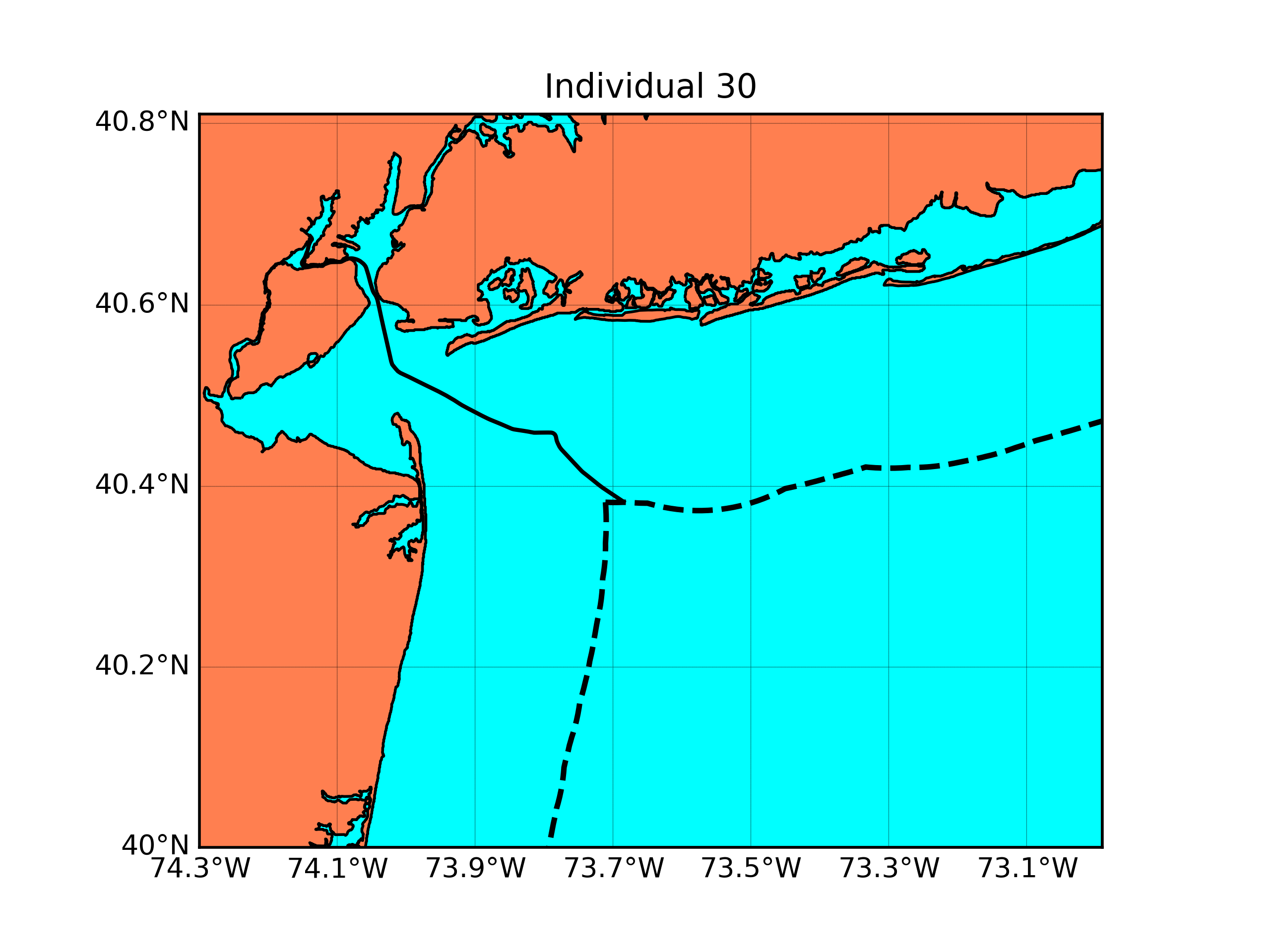}
\includegraphics[width=0.23\textwidth]{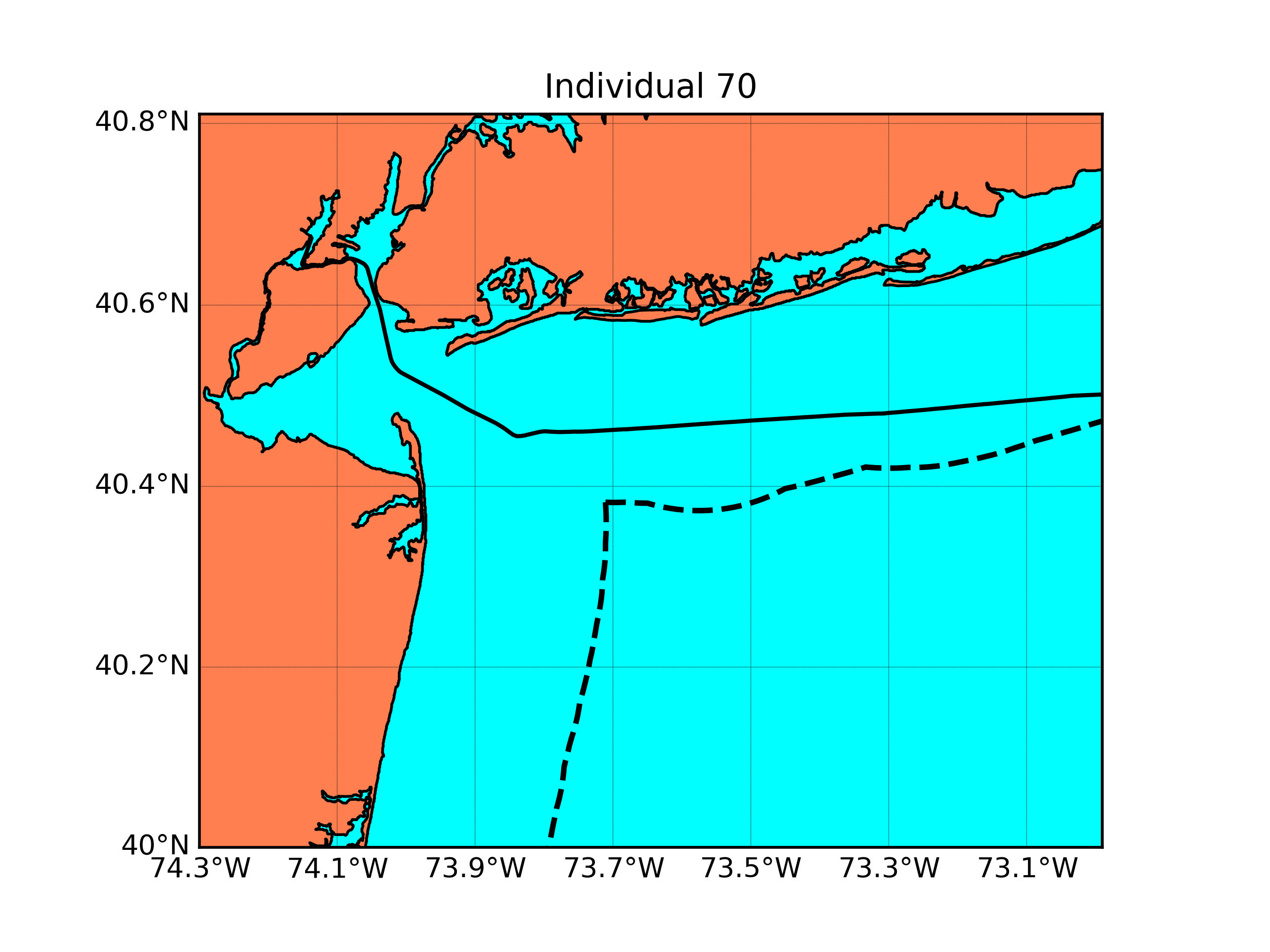}
\includegraphics[width=0.23\textwidth]{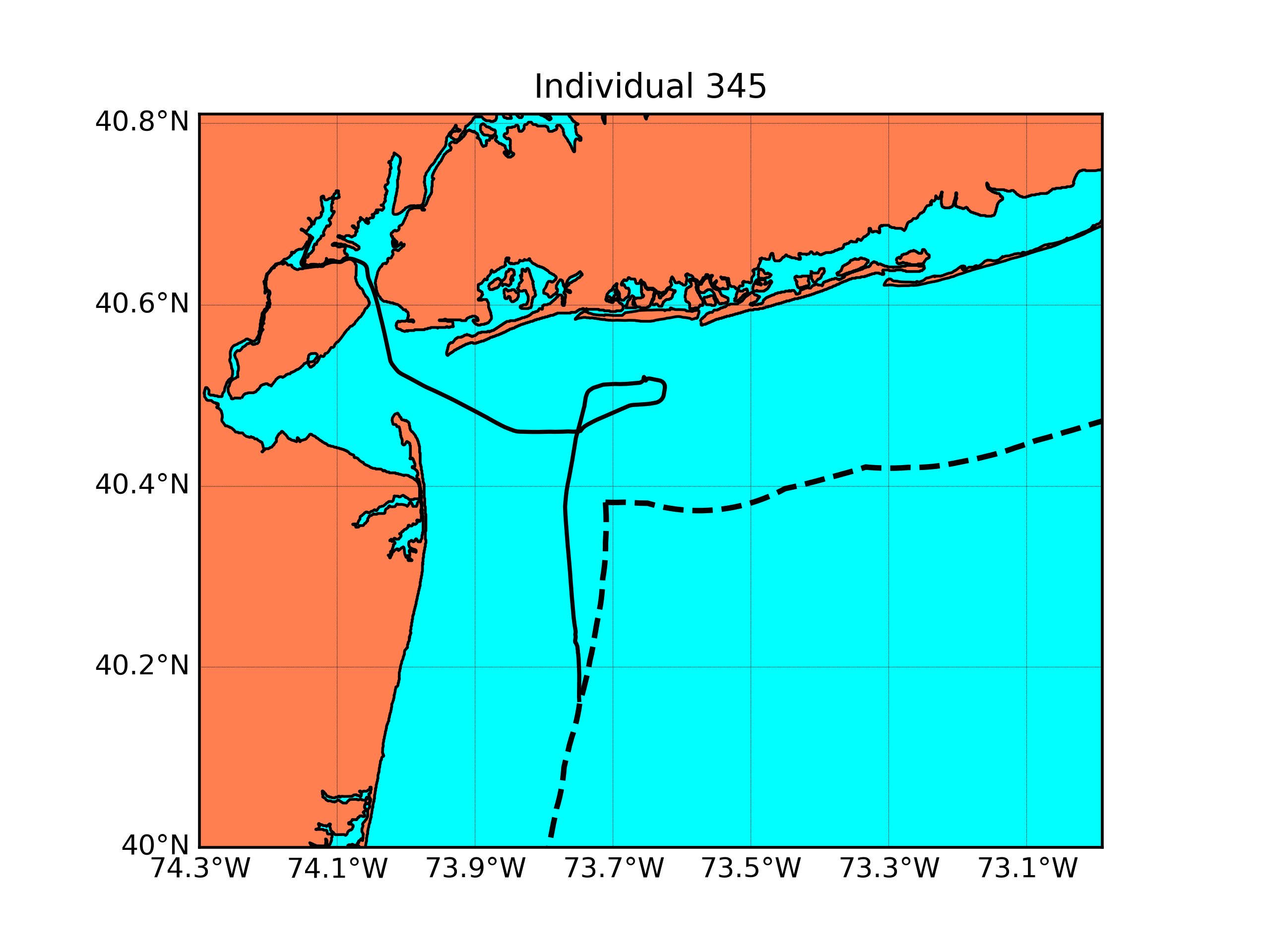}\\
\includegraphics[width=0.23\textwidth]{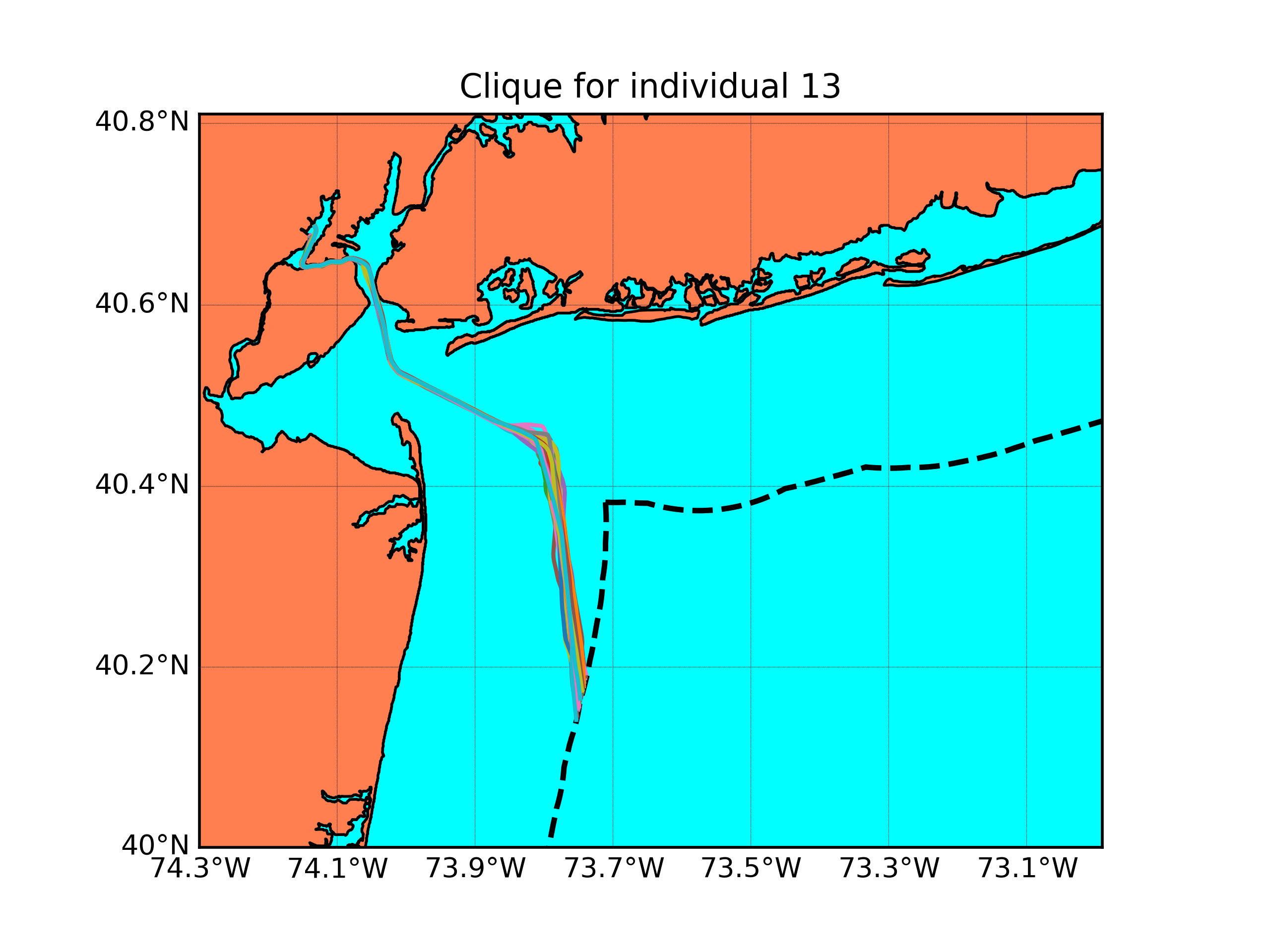}
\includegraphics[width=0.23\textwidth]{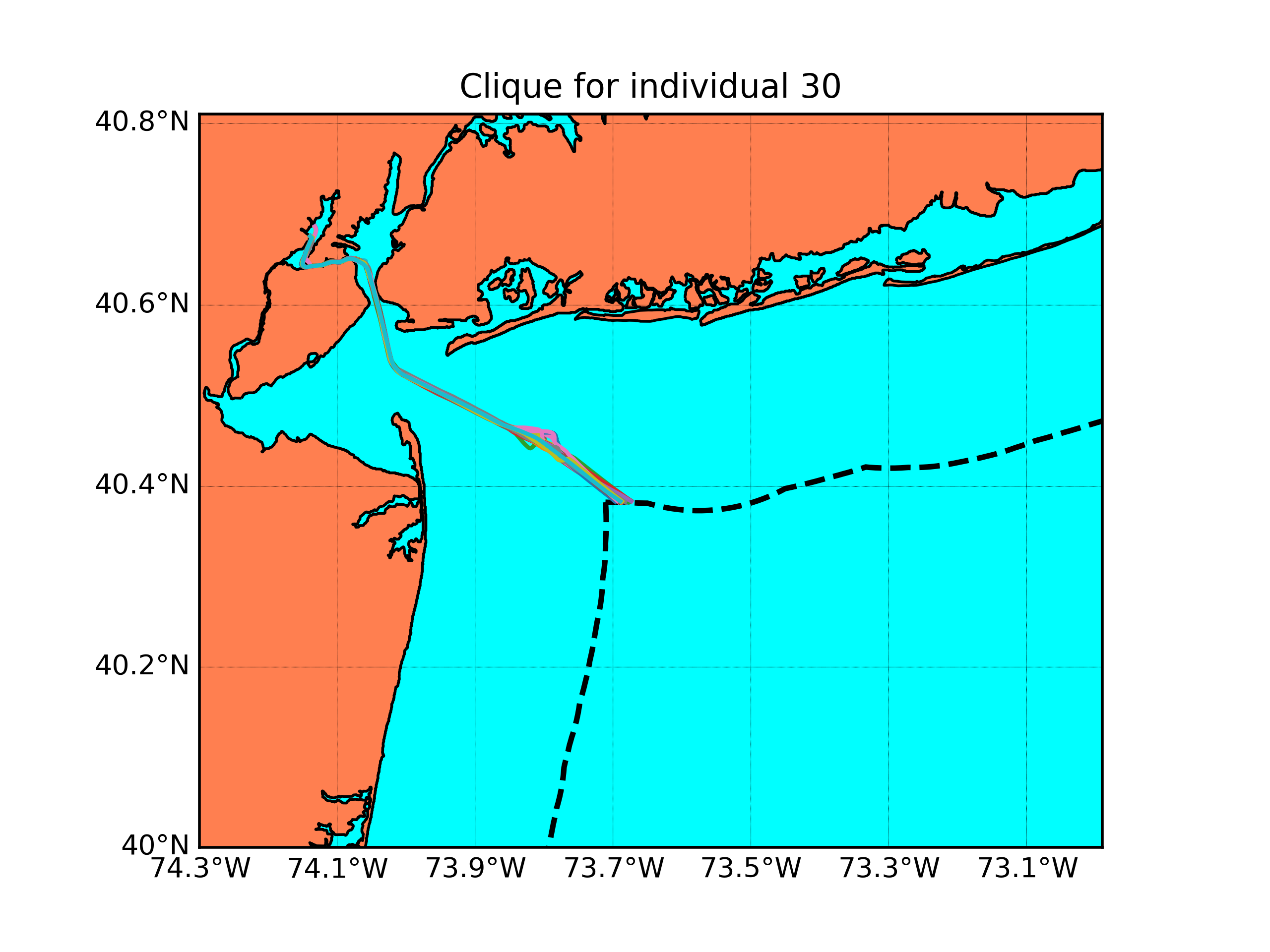}
\includegraphics[width=0.23\textwidth]{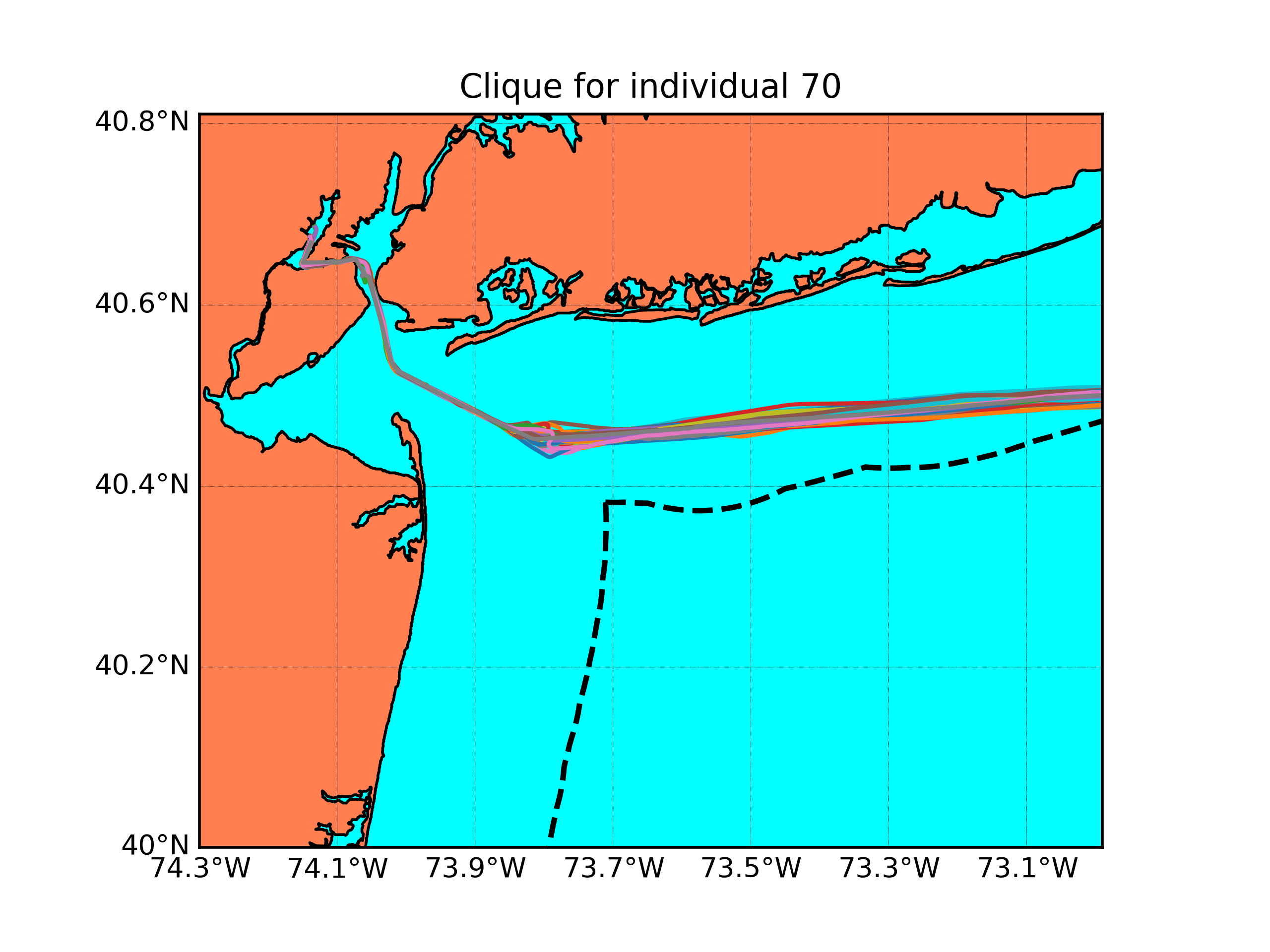}
\includegraphics[width=0.23\textwidth]{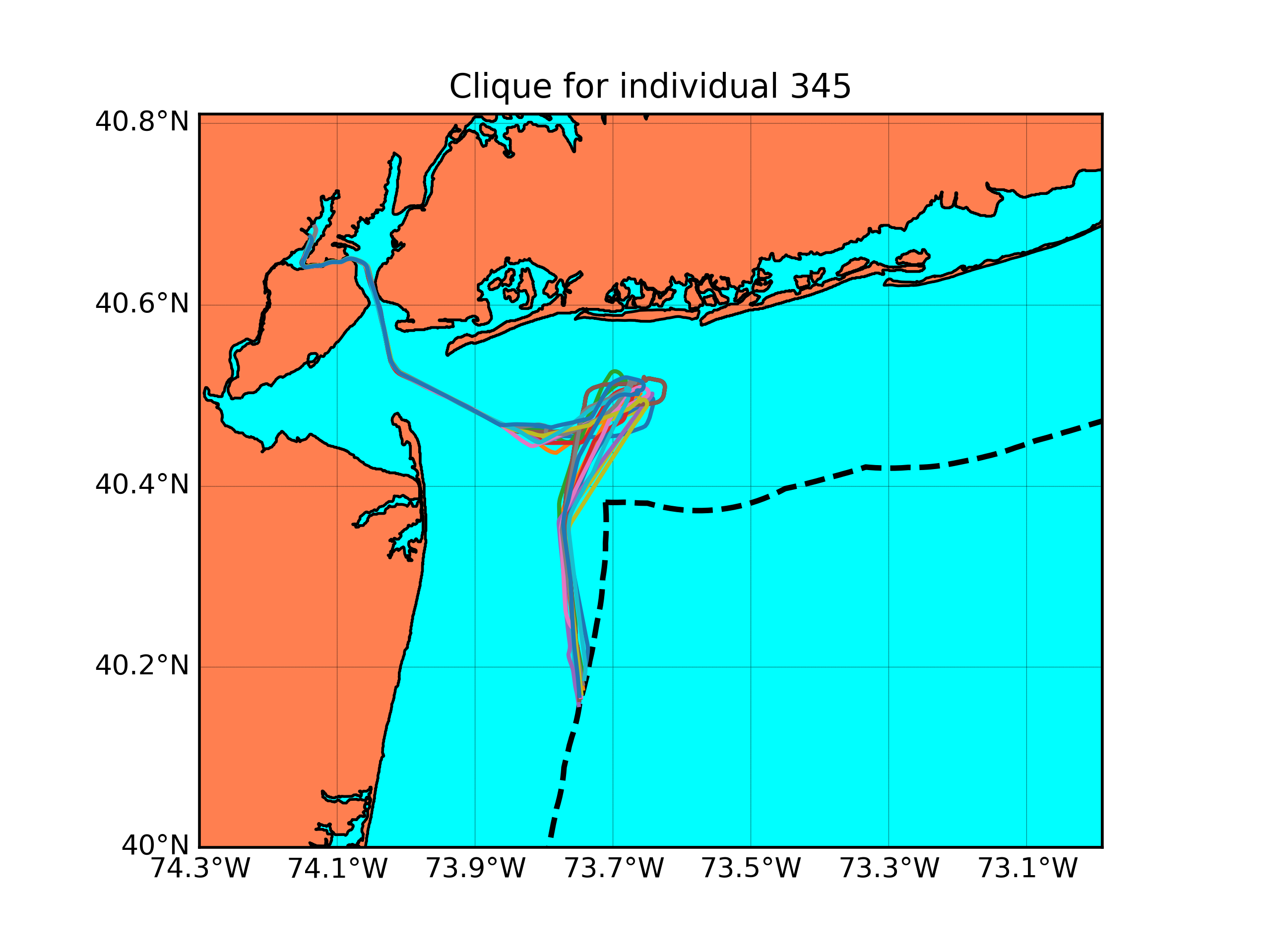}
\caption{Four typical trajectories and their identified individualized groups.}\label{fig:cliques}
\end{figure}
\par
Since the individual level estimator for $\theta_k$ is not available as we only have one observation $x_k$ per individual. The iGroup estimator is constructed by aggregating the log-likelihood functions. In this case, it is equivalent to estimate $\theta_k$ by the sample mean and the sample standard deviation from the formed igroup. Since our main focus is to identify outliers, we exclude the target from the estimation. Denote $\mathcal C_k$ as the individualized group (clique) identified by the DTW distance for voyage $k$. Note that we control $|\mathcal C_k|=40$. The iGroup estimator can be constructed as 
\begin{align*}
    \mu_k^{(c)} = \dfrac{\sum_{i\in \mathcal C_k}x_i}{|\mathcal C_k|},\quad
    \sigma_k^{(c)}=\dfrac{\sum_{i\in \mathcal C_k}(x_i-\mu_k^{(c)})^2}{|\mathcal C_k| - 1}.
\end{align*}
Then the risk score (the likelihood of being an outlier) of individual $k$ can be obtained as 
$$1-2P\left(Z>\left|\dfrac{x_k - \mu_k^{(c)}}{\sigma_k^{(c)}}\right|\right),$$
where $Z\sim N(0, 1).$\\
\begin{figure}[!htp]
\centering
\includegraphics[width=0.3\textwidth]{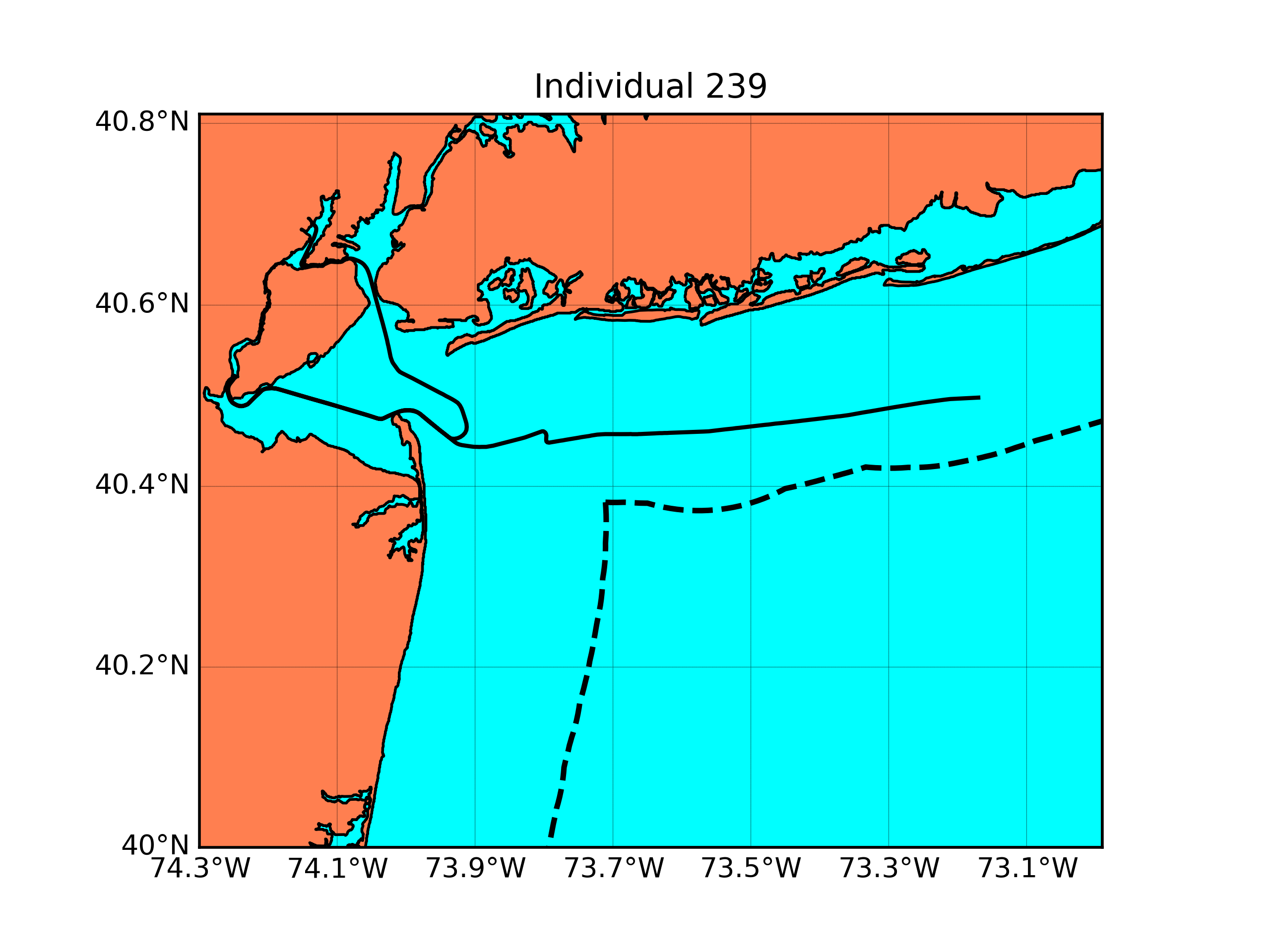}
\includegraphics[width=0.3\textwidth]{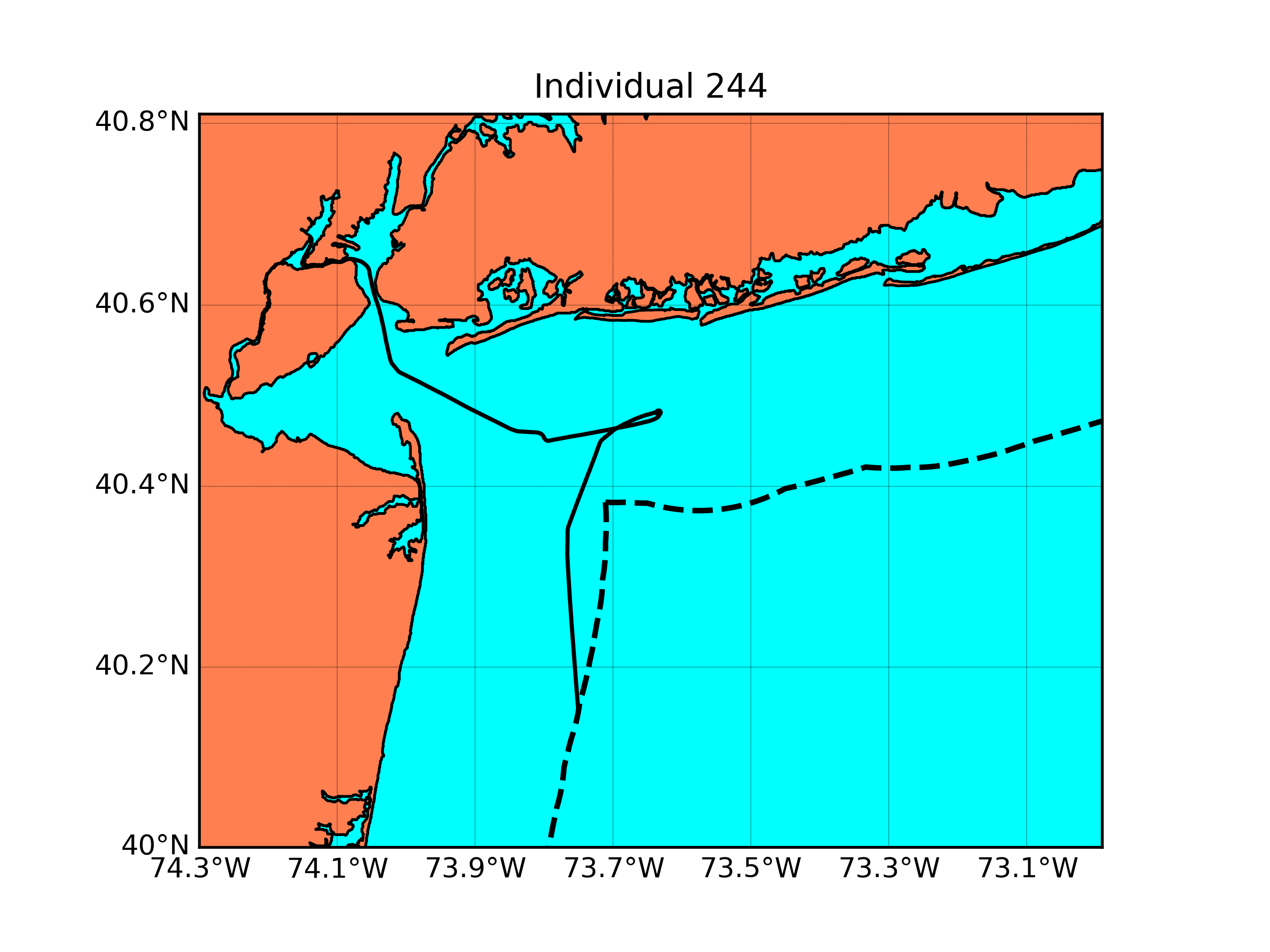}
\includegraphics[width=0.3\textwidth]{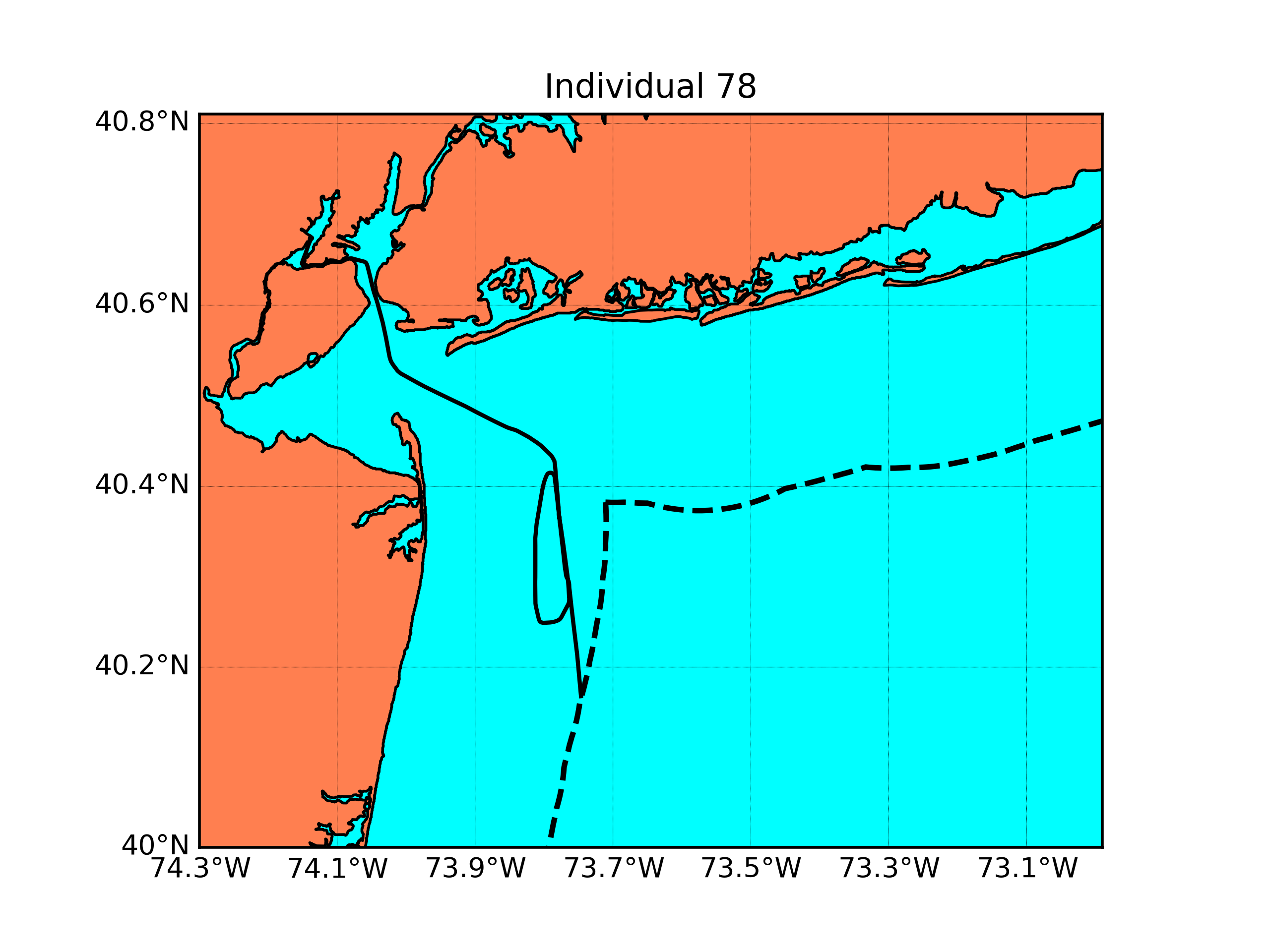}
\caption{Outliers among vessels/voyages in trajectories of vessels heading to Port of Newark}\label{fig:outlier}
\end{figure}
\par
In these 534 vessels, 95 outliers with risk scores larger than 95\% were determined as abnormal. A manual inspection reveals that they belong to three categories (with some overlaps between (a) and (b)): (a) 40 vessels had a prior dock before the Port of Newark (left panel of Figure \ref{fig:outlier}); (b) 18 vessels were anchored somewhere outside the port for an extremely long time (middle panel); (c) the other 43 vessels were traveling too fast/slow compared with their iGroup (right panel). Figure \ref{fig:outlier} shows typical trajectories of the three categories. Due to the limited population, vessels with few similar trajectories are also classified as abnormal such as the one shown in the right panel in Figure \ref{fig:outlier}.

\section{Conclusion and Discussion}\label{sec:conclusion}
In conclusion, the proposed iGroup method provides an effective tool for efficient inference in a heterogeneous population. The approach is essentially nonparametric. It has several special features:
(1) The grouping idea can facilitate and answer some inference questions that are otherwise difficult or impossible to address such as estimating variance/quantile when each individual has only one observation. 
(2) It reduces the standard error of the estimator by pooling together individuals with similar characteristics. 
(3) The grouping can take a non-standard exogenous variable $\bm z$ into consideration, as long as a similarity/distance measure is defined. 
(4) Noisy exogenous variable $\bm z$ can contribute to grouping as well.
(5) A useful weight function measuring similarity between $\hat\theta$'s is designed with statistical interpretation. 
(6) The method can be extended to a wide range of estimating methods, which optimizes an objective function, such as regularized least squares estimation, generalized moment estimation, etc.
(7) The bandwidth can be tuned by leave-one-out cross validation either globally or locally.
\par
In addition, we showed the asymptotic performance and theoretical properties of the method, which assess the accuracy and efficiency of iGroup and provide practical guidance in implementation. More specifically, when the loss function is given and weight function is properly constructed by our approach, the iGroup estimator converges to the Bayes estimator that minimizes the overall risk without knowing the prior. 
Computationally, as the group construction and inference procedure are identical for all individuals, the iGroup method can be easily parallelized for large datasets.

\par
In Theorems \ref{thm:mest}, \ref{thm:m-theta-only} and \ref{thm:m-complete}, we assumed a quite strong sufficient condition on the objective functions $M_k(\theta)$ or $M(\theta, \hat\theta)$ such that the minimum point of the aggregated objective function will converge to the true value. Instead of assuming second-order differentiability and convexity, other sufficient conditions can also guarantee the convergence of the minimum point \citep{van2000asymptotic}. But most of them depends on the explicit formula of kernel $\mathcal K$ and the objective function $M_k(\theta)$.

\par 
The iGroup approach has its connection to the empirical Bayes approach \citep{robbins1985empirical}, where the prior is unknown, but a Bayes estimator is constructed. Although an unknown population distribution for $\theta$ is assumed to be $\pi(\theta)$ viewed as the prior, it does not appear explicitly in either $\hat\theta_0^{(c)}$ or $\tilde\theta_0^{(c)}$ in our approach. And we showed in Section \ref{sec:theoretical-results} that under mild conditions, the iGroup estimators converge to certain Bayes estimators under the unknown prior. In empirical Bayes, the prior is usually estimated by either discretization or deconvolution. But the iGroup approach is different. The unknown $\pi(\theta)$ is not directly estimated and it is not needed. The prior information is taken into consideration by taking a (weighted) average of sample estimators or sample objective functions. And the weight function $w_2(\cdot)$, which is related to $\pi(\theta)$ in close form, is approximated using the bootstrap method in Section \ref{sec:evaluate-weight}.

\newpage
\renewcommand\refname{Bibliography}
\bibliographystyle{apalike}
\bibliography{reference}

\begin{thebibliography}{}

\bibitem[Agrawal et~al., 1998]{agrawal1998automatic}
Agrawal, R., Gehrke, J., Gunopulos, D., and Raghavan, P. (1998).
\newblock Automatic subspace clustering of high dimensional data for data
  mining applications.
\newblock {\em SIGMOD Rec.}, 27(2):94--105.

\bibitem[Altman, 1992]{Altman1992nonparametric}
Altman, N.~S. (1992).
\newblock An introduction to kernel and nearest-neighbor nonparametric
  regression.
\newblock {\em The American Statistician}, 46(3):175--185.

\bibitem[Antoniak, 1974]{antoniak1974mixtures}
Antoniak, C.~E. (1974).
\newblock Mixtures of dirichlet processes with applications to bayesian
  nonparametric problems.
\newblock {\em The Annals of Statistics}, 2(6):1152--1174.

\bibitem[Binder, 1978]{binder1978bayesian}
Binder, D.~A. (1978).
\newblock Bayesian cluster analysis.
\newblock {\em Biometrika}, 65(1):31--38.

\bibitem[Bound et~al., 2001]{bound2001measurement}
Bound, J., Brown, C., and Mathiowetz, N. (2001).
\newblock Measurement error in survey data.
\newblock In {\em Handbook of Econometrics}, volume~5, pages 3705--3843.
  Elsevier.

\bibitem[Carroll et~al., 1995]{carroll1995nonlinear}
Carroll, R., Ruppert, D., and Stefanski, L. (1995).
\newblock Nonlinear measurement error models.
\newblock {\em Monographs on Statistics and Applied Probability}, 63.

\bibitem[Chiu, 1991]{chiu1991bandwidth}
Chiu, S.-T. (1991).
\newblock Bandwidth selection for kernel density estimation.
\newblock {\em The Annals of Statistics}, 19(4):1883--1905.

\bibitem[Collins and Varmus, 2015]{collins2015new}
Collins, F.~S. and Varmus, H. (2015).
\newblock A new initiative on precision medicine.
\newblock {\em New England Journal of Medicine}, 372(9):793--795.

\bibitem[Diaconis and Freedman, 1986]{diaconis1986consistency}
Diaconis, P. and Freedman, D. (1986).
\newblock On the consistency of bayes estimates.
\newblock {\em The Annals of Statistics}, 14(1):1--26.

\bibitem[Duda and Hart, 1973]{duda1973pattern}
Duda, R.~O. and Hart, P.~E. (1973).
\newblock {\em Pattern Classification and Scene Analysis}.
\newblock John Willey \& Sons, New York.

\bibitem[Fama and French, 1993]{fama1993common}
Fama, E.~F. and French, K.~R. (1993).
\newblock Common risk factors in the returns on stocks and bonds.
\newblock {\em Journal of Financial Economics}, 33(1):3--56.

\bibitem[Fan and Truong, 1993]{fan1993nonparametric}
Fan, J. and Truong, Y.~K. (1993).
\newblock Nonparametric regression with errors in variables.
\newblock {\em The Annals of Statistics}, 21(4):1900--1925.

\bibitem[Ferguson, 1973]{ferguson1973bayesian}
Ferguson, T.~S. (1973).
\newblock A bayesian analysis of some nonparametric problems.
\newblock {\em The Annals of Statistics}, 1(2):209--230.

\bibitem[Figueiredo and Jain, 2000]{figueiredo2000unsupervised}
Figueiredo, M. A.~T. and Jain, A.~K. (2000).
\newblock Unsupervised learning of finite mixture models.
\newblock {\em IEEE Transaction on Pattern Analysis and Machine Intelligence},
  24:381--396.

\bibitem[Fuller, 2009]{fuller2009measurement}
Fuller, W.~A. (2009).
\newblock {\em Measurement error models}, volume 305.
\newblock John Wiley \& Sons.

\bibitem[Gan et~al., 2007]{gan2007data}
Gan, G., Ma, C., and Wu, J. (2007).
\newblock {\em Data Clustering: Theory, Algorithms, and Applications},
  volume~20.
\newblock Society for Industrial and Applied Mathematics.

\bibitem[Hall et~al., 2008]{Hall2008kth}
Hall, P., Park, B.~U., and Samworth, R.~J. (2008).
\newblock Choice of neighbor order in nearest-neighbor classification.
\newblock {\em The Annals of Statistics}, 36(5):2135--2152.

\bibitem[Jain, 2010]{jain2010data}
Jain, A.~K. (2010).
\newblock Data clustering: 50 years beyond k-means.
\newblock {\em Pattern Recognition Letters}, 31(8):651--666.

\bibitem[Jain et~al., 1999]{jain1999data}
Jain, A.~K., Murty, M.~N., and Flynn, P.~J. (1999).
\newblock Data clustering: a review.
\newblock {\em ACM Computing Surveys (CSUR)}, 31(3):264--323.

\bibitem[Juang, 1984]{juang1984hidden}
Juang, B.-H. (1984).
\newblock On the hidden markov model and dynamic time warping for speech
  recognition - a unified view.
\newblock {\em AT\&T Bell Laboratories Technical Journal}, 63(7):1213--1243.

\bibitem[Koenig et~al., 2008]{koenig2008speech}
Koenig, L.~L., Lucero, J.~C., and Perlman, E. (2008).
\newblock Speech production variability in fricatives of children and adults:
  Results of functional data analysis.
\newblock {\em The Journal of the Acoustical Society of America},
  124(5):3158--3170.

\bibitem[Liao, 2005]{liao2005clustering}
Liao, T.~W. (2005).
\newblock Clustering of time series data - a survey.
\newblock {\em Pattern Recognition}, 38(11):1857--1874.

\bibitem[Lindsay, 1995]{lindsay1995mixture}
Lindsay, B.~G. (1995).
\newblock Mixture models: Theory, geometry and applications.
\newblock {\em NSF-CBMS Regional Conference Series in Probability and
  Statistics}, 5:i--163.

\bibitem[Liu and Meng, 2016]{liu2016there}
Liu, K. and Meng, X.~L. (2016).
\newblock There is individualized treatment. why not individualized inference?
\newblock {\em Annual Review of Statistics and Its Application}, 3:79--111.

\bibitem[Lo, 1984]{lo1984class}
Lo, A.~Y. (1984).
\newblock On a class of bayesian nonparametric estimates: I. density estimates.
\newblock {\em The Annals of Statistics}, 12(1):351--357.

\bibitem[Nakagawa and Nakanishi, 1988]{nakagawa1988speaker}
Nakagawa, S. and Nakanishi, H. (1988).
\newblock Speaker-independent english consonant and japanese word recognition
  by a stochastic dynamic time warping method.
\newblock {\em IETE Journal of Research}, 34(1):87--95.

\bibitem[Ng and Han, 1994]{ng1994cient}
Ng, R.~T. and Han, J. (1994).
\newblock Efficient and effective clustering methods for spatial data mining.
\newblock In {\em Proceedings of the 20th International Conference on Very
  Large Data Bases}, VLDB '94, pages 144--155. Morgan Kaufmann Publishers Inc.

\bibitem[Qian and Murphy, 2011]{qian2011performance}
Qian, M. and Murphy, S.~A. (2011).
\newblock Performance guarantees for individualized treatment rules.
\newblock {\em The Annals of Statistics}, 39(2):1180.

\bibitem[Robbins, 1956]{robbins1985empirical}
Robbins, H. (1956).
\newblock An empirical bayes approach to statistics.
\newblock In {\em Proceedings of the Third Berkeley Symposium on Mathematical
  Statistics and Probability, Volume 1: Contributions to the Theory of
  Statistics}, pages 157--163. University of California Press.

\bibitem[Sakoe and Chiba, 1978]{sakoe1978dynamic}
Sakoe, H. and Chiba, S. (1978).
\newblock Dynamic programming algorithm optimization for spoken word
  recognition.
\newblock {\em IEEE Transactions on Acoustics, Speech, and Signal Processing},
  26(1):43--49.

\bibitem[Shen et~al., 2018]{Shen2018ifusion}
Shen, J., Liu, R., and Xie, M. (2018).
\newblock {$i$-Fusion}: Individualized fusion learning.
\newblock manuscript.

\bibitem[Stefanski and Carroll, 1990]{stefanski1990deconvolving}
Stefanski, L.~A. and Carroll, R.~J. (1990).
\newblock Deconvolving kernel density estimators.
\newblock {\em Statistics}, 21(2):169--184.

\bibitem[Teh et~al., 2005]{teh2005sharing}
Teh, Y.~W., Jordan, M.~I., Beal, M.~J., and Blei, D.~M. (2005).
\newblock Sharing clusters among related groups: Hierarchical dirichlet
  processes.
\newblock In {\em Advances in Neural Information Processing Systems 17}, pages
  1385--1392. MIT Press.

\bibitem[Van~der Vaart, 2000]{van2000asymptotic}
Van~der Vaart, A.~W. (2000).
\newblock {\em Asymptotic statistics}, volume~3.
\newblock Cambridge University Press.

\bibitem[Wang et~al., 2007]{wang2007statistics}
Wang, R., Lagakos, S.~W., Ware, J.~H., Hunter, D.~J., and Drazen, J.~M. (2007).
\newblock Statistics in medicine—reporting of subgroup analyses in clinical
  trials.
\newblock {\em New England Journal of Medicine}, 357(21):2189--2194.

\bibitem[Wansbeek and Meijer, 2000]{wansbeek2000measurement}
Wansbeek, T.~J. and Meijer, E. (2000).
\newblock {\em Measurement error and latent variables in econometrics},
  volume~37.
\newblock North-Holland.

\bibitem[Wasserman, 2010]{wassermann2006all}
Wasserman, L. (2010).
\newblock {\em All of Nonparametric Statistics}.
\newblock Springer Publishing Company, Incorporated.

\bibitem[Xu and Wunsch, 2005]{xu2005survey}
Xu, R. and Wunsch, D. (2005).
\newblock Survey of clustering algorithms.
\newblock {\em IEEE Transactions on Neural Networks}, 16(3):645--678.

\bibitem[Yang et~al., 2012]{yang2012classification}
Yang, J., Miescke, K., and McCullagh, P. (2012).
\newblock Classification based on a permanental process with cyclic
  approximation.
\newblock {\em Biometrika}, 99(4):775--786.

\bibitem[Zhao et~al., 2012]{zhao2012estimating}
Zhao, Y., Zeng, D., Rush, A.~J., and Kosorok, M.~R. (2012).
\newblock Estimating individualized treatment rules using outcome weighted
  learning.
\newblock {\em Journal of the American Statistical Association},
  107(499):1106--1118.

\end{thebibliography}

\newpage
\appendix
\section*{Proof of Theorem \ref{thm:mest}}\label{pf:mest}
We prove the consistency first. 
Define
\begin{align*}
    \psi_k(\theta) &= \dfrac{\partial }{\partial \theta}M_k(\theta),\\
    \Psi_K(\theta) &= \dfrac{\sum_{k=0}^K\mathcal K\left(\dfrac{\|\bm z_k-\bm z_0\|}{b}\right)\psi_k(\theta)}{\sum_{k=0}^K\mathcal K\left(\dfrac{\|\bm z_k-\bm z_0\|}{b}\right)},\\
    \Psi(\theta) &= \mathbb E_{\bm x|\bm z_0} \psi_{\bm x}(\theta).
\end{align*}
For any given $\theta$, $\Psi_K(\theta)$ is a kernel smoothing estimator for $\mathbb E_{\bm x|\bm z_0}[\psi_{\bm x}(\theta)]=\Psi(\theta)$. Hence $\Psi_K(\theta)\rightarrow \Psi(\theta)$ in probability for any given $\theta$, provided $\mathbb E_{\bm x|\bm z}[\psi_{\bm x}(\theta)]$ continuous at $\bm z_0$ \citep{wassermann2006all}. 
Due to the assumption that $M_k(\theta)$ is convex and second-order differentiable, $\psi_k(\theta)$ is a non-decreasing function for any $\bm x_k$. Therefore, both $\Psi_K$ and $\Psi$ are non-decreasing and continuous. By assumption, $\Theta_0$ is the unique root of $\Psi(\theta)$. Let $\theta^*_K$ be such a point that $\Psi_K(\theta^*_K)=0$. $\theta^*_K$ may not be unique and may not even exist for small $K$. For any $\epsilon>0$, it is immediate that $\Psi(\Theta_0-\epsilon)<0<\Psi(\Theta_0+\epsilon)$ and by the pointwise convergence in probability of $\Psi_K$, we have
\begin{align*}
    P\left[\left|\Psi_K(\Theta_0-\epsilon) - \Psi(\Theta_0-\epsilon)\right|\leqslant \dfrac{1}{2}\left|\Psi(\Theta_0-\epsilon)\right|\right]&\longrightarrow 1,\\
    P\left[\left|\Psi_K(\Theta_0+\epsilon) - \Psi(\Theta_0+\epsilon)\right|\leqslant \dfrac{1}{2}\left|\Psi(\Theta_0+\epsilon)\right|\right]&\longrightarrow 1.
\end{align*}
Therefore, 
$$P\left[\left|\Psi_K(\Theta_0-\epsilon) - \Psi(\Theta_0-\epsilon)\right|\leqslant \dfrac{1}{2}\left|\Psi(\Theta_0-\epsilon)\right|, \left|\Psi_K(\Theta_0+\epsilon) - \Psi(\Theta_0+\epsilon)\right|\leqslant \dfrac{1}{2}\left|\Psi(\Theta_0+\epsilon)\right|\right]\longrightarrow 1.$$
The event in the probability implies that $\Psi_K(\Theta_0-\epsilon)<0<\Psi_K(\Theta_0+\epsilon)$, which further implies the existence of $\theta_K^*$ in $(\Theta_0 - \epsilon, \Theta_0 + \epsilon)$ by continuity of $\Psi_K$. Hence
\begin{align*}
    &P\left[\left|\Psi_K(\Theta_0-\epsilon) - \Psi(\Theta_0-\epsilon)\right|\leqslant \dfrac{1}{2}\left|\Psi(\Theta_0-\epsilon)\right|, \left|\Psi_K(\Theta_0+\epsilon) - \Psi(\Theta_0+\epsilon)\right|\leqslant \dfrac{1}{2}\left|\Psi(\Theta_0+\epsilon)\right|\right]\\
    \leqslant & P\left[\Psi_K(\Theta_0-\epsilon)<0<\Psi_K(\Theta_0+\epsilon)\right]\\
    \leqslant & P\left[\Theta_0 - \epsilon < \theta_K^* <\Theta_0 + \epsilon\right].
\end{align*}
Since the first term converges to 1, the last term converges to 1 as well. Note that when $\tilde\theta_0^{(c)}$ exists, it equals $\theta_K^*$. The consistency of $\tilde\theta_0^{(c)}$ is proved.
\par
With $\tilde\theta_0^{(c)}\longrightarrow \theta_0$ in probability, it is reasonable to expand $\Psi_K(\tilde\theta_0^{(c)})$ at $\Theta_0$.
\begin{align*}
\sum_{k=0}^K\mathcal K\left(\dfrac{\|\bm z_k-\bm z_0\|}{b}\right)\psi_k(\Theta_0)+(\tilde\theta^{(c)}-\Theta_0)\sum_{k=0}^K\mathcal K\left(\dfrac{\|\bm z_k-\bm z_0\|}{b}\right)\psi'_k(\Theta_0)\\
+\sum_{k=0}^K\mathcal K\left(\dfrac{\|\bm z_k-\bm z_0\|}{b}\right)O((\tilde\theta^{(c)}_0-\Theta_0)^2)=0.
\end{align*}
Now we have
\begin{equation}\label{eq: aggtheta_kernel}
    \tilde\theta^{(c)}_0-\Theta_0=-\dfrac{\dfrac{1}{K+1}\sum_{k=0}^K\mathcal K\left(\dfrac{\|\bm z_k-\bm z_0\|}{b}\right)\psi_k(\Theta_0)}{\dfrac{1}{K+1}\sum_{k=0}^K\mathcal K\left(\dfrac{\|\bm z_k-\bm z_0\|}{b}\right)\psi'_k(\Theta_0)+O(\hat\theta^{(c)}_0-\Theta_0)}.
\end{equation}
Consider $K\rightarrow \infty$. On one hand, the numerator is a kernel smoothing estimator for $\mathbb E_{\bm x|\bm z_0}[\psi_{\bm x}(\Theta_0)]=0$ up to a normalizing constant. On the other hand, the denominator is a similar kernel smoothing estimator for $\mathbb E_{\bm x|\bm z_0} \psi'_{\bm x}(\Theta_0)$. By Slutsky's theorem, their ratio has a similar asymptotic distribution to the numerator kernel smoothing estimator up to a constant factor of $\mathbb E_{\bm x|\bm z_0} \psi'_{\bm x}(\Theta_0)$. Therefore, $\tilde\theta_0^{(c)}$ has an asymptotic bias $O_p(b^2)$ and an asymptotic variance $O_p(1/Kb^d)$  \citep{wassermann2006all}. Hence, the optimal choice of bandwidth in a bias-variance optimization scheme is $\hat b\asymp K^{-1/(d+4)}$ and the optimal MSE is of order $K^{-4/(d+4)}$.

\section*{Proof of Theorem \ref{thm:bias-variance}}\label{pf:bias-variance}
In this case, $\theta_0$ is assumed to be fixed, and $\hat\theta_0^{(c)}$ is a standard kernel smoothing estimator for $\mathbb E_\pi[\theta_0|\bm z_0] = \Theta_0$. By following the asymptotic property of a standard kernel smoothing estimator, we have
\begin{align*}
\mathbb E[\hat\theta_0^{(c)}|\bm z_0]=\Theta_0+O_p(b^2)\text{ and }
Var[\hat\theta_0^{(c)}|\bm z_0]=O_p\left(\dfrac{1}{Kd^b}\right).
\end{align*}
Therefore, we have
\begin{align*}
&\mathbb E_{\theta_0}[\hat\theta_0^{(c)}]=\mathbb E_{\theta_0}[\Theta_0]+O_p(b^2),\\
&Var_{\theta_0}[\hat\theta_0^{(c)}]=Var_{\theta_0}[\mathbb E[\hat\theta_0^{(c)}|\bm z_0]]+\mathbb E_{\theta_0}[Var[\hat\theta_0^{(c)}|\bm z_0]]=Var_{\theta_0}[\Theta_0]+O_p\left(\dfrac{1}{Kd^b}\right).
\end{align*}

\section*{Proof of Theorem \ref{thm: intrinsic_asymp}}\label{pf:intrinsic_asymp}
We first prove the following lemma, which would be used in the proof of Theorem \ref{thm: intrinsic_asymp}.
\begin{lem}
Suppose the random vector $\xi$ has a pdf $p_\xi$ and has zero mean, finite variance and finite higher moments such that
$$\mathbb E\xi =0,\quad \text{Var}(\xi)=\sigma^2\bm\Sigma,\quad \|\bm\Sigma\|=1.$$
Then for any second-order partially differentiable function $f$, we have
$$\int f(\bm x+t)p_\xi(t)dt = f(\bm x)+\dfrac{1}{2}\sigma^2\tr[\nabla^2 f(\bm x)\bm\Sigma]+o(\sigma^2),$$
when $\sigma^2\rightarrow 0$.
\end{lem}
\begin{proof}
Let $\xi_1=\xi/\sigma$, then $\mathbb E(\xi_1)=0$ and $\text{Var}(\xi_1)=\bm\Sigma$. Hence
\begin{align*}
\int f(\bm x+t)p_\xi(t)dt&=\int f(\bm x+\sigma s)p_{\xi_1}(s)ds\\
&=\int \left[f(\bm x)+\sigma s^T[\nabla f(\bm x)]+\dfrac{1}{2}\sigma^2s^T[\nabla^2 f(x)]s+o(\sigma^2)\right]p_{\xi_1}(s)ds\\
&=f(\bm x)+\dfrac{1}{2}\sigma^2\int s^T[\nabla^2f(\bm x)]sp_{\xi_1}(s)ds+o(\sigma^2)\\
&=f(\bm x)+\dfrac{1}{2}\sigma^2\tr[\nabla^2f(\bm x)\bm\Sigma]+o(\sigma^2).
\end{align*}
\end{proof}
Now we prove Theorem \ref{thm: intrinsic_asymp}. Let $\bar\pi()$ be the population distribution for $\bm \eta$. Since $\theta = g(\bm\eta)$, we have
\begin{align*}
\mathbb E_\pi[\theta_0|\bm z_0]&=\dfrac{\int g(\bm \eta)p(\bm z_0|\bm \eta)\bar\pi(\bm \eta)d\bm \eta}{\int p(\bm z_0|\bm \eta)\bar\pi(\bm\eta)d\bm \eta}\\
&=\dfrac{(g\bar\pi)(\bm z_0) +\dfrac{1}{2}\sigma_z^2 \tr[\nabla^2(g\bar\pi)(\bm z_0)\bm\Sigma_z]+o(\sigma_z^2)}{\bar\pi(\bm z_0)+\dfrac{1}{2}\sigma_z^2\tr[\nabla^2\bar\pi(\bm z_0)\bm\Sigma_z]+o(\sigma_z^2)}\\
&=\dfrac{(g\bar\pi)(\bm z_0) +\dfrac{1}{2}\sigma_z^2 \tr[\nabla^2(g\bar\pi)(\bm z_0)\bm\Sigma_z]+o(\sigma_z^2)}{\bar\pi(\bm z_0)}\left[1-\dfrac{1}{2}\sigma_z^2 \dfrac{\tr[\nabla^2\bar\pi(\bm z_0)\bm\Sigma_z]}{\bar\pi(\bm z_0)}+o(\sigma_z^2)\right]\\
&=g(\bm z_0)+\dfrac{\sigma_z^2}{2\bar\pi(\bm z_0)}\left(\tr[\nabla^2(g\bar\pi)(\bm z_0)\bm\Sigma_z]-g(\bm z_0)\tr[\nabla^2\bar\pi(\bm z_0)\bm \Sigma_z]\right)+o(\sigma_z^2)\\
&=g(\bm z_0)+\sigma_z^2\left(\dfrac{\tr[\nabla^2g(\bm z_0)\bm\Sigma_z]}{2}+\dfrac{\tr[\nabla \bar\pi(\bm z_0)^T\bm\Sigma_z\nabla g(\bm z_0)]}{\bar\pi(\bm z_0)}\right)+o(\sigma_z^2).
\end{align*}
Thus, the bias is 
\begin{align*}
\mathbb E_{\theta_0}[\mathbb E_\pi[g(\bm \eta)|\bm z_0]]-g(\bm \eta_0)&=\int \mathbb E_\pi[g(\bm \eta)|\bm z_0] p(\bm z_0|\bm \eta_0)d\bm z_0-g(\bm \eta_0)\\
&=\int \left(g(\bm z_0)+\sigma_z^2\left(\dfrac{\tr[\nabla^2g(\bm z_0)\bm\Sigma_z]}{2}+\dfrac{\tr[\nabla \bar\pi(\bm z_0)^T\bm\Sigma_z\nabla g(\bm z_0)]}{\bar\pi(\bm z_0)}\right)+o(\sigma_z^2)\right)p(\bm z_0|\bm \eta_0)d\bm z_0-g(\bm \eta_0)\\
&=g(\bm \eta_0)+\sigma_z^2\left(\dfrac{\tr[\nabla^2g(\bm \eta_0)\bm\Sigma_z]}{2}+\dfrac{\tr[\nabla \bar\pi(\bm \eta_0)^T\bm\Sigma_z\nabla g(\bm \eta_0)]}{\bar\pi(\bm \eta_0)}\right)\\
&\ \ +\dfrac{1}{2}\sigma_z^2\tr[\nabla^2g(\bm \eta_0)\bm\Sigma_z]+o(\sigma_z^2)-g(\bm \eta_0)\\
&=\sigma_z^2\left(\tr[\nabla^2g(\bm \eta_0)\bm\Sigma_z]+\dfrac{\tr[\nabla \bar\pi(\bm \eta_0)^T\bm\Sigma_z\nabla g(\bm \eta_0)]}{\bar\pi(\bm \eta_0)}\right)+o(\sigma_z^2)\\
&\asymp \sigma_z^2.
\end{align*}
On the other hand,
$$\left(\mathbb E_\pi[g(\bm \eta)|\bm z_0]\right)^2=g^2(\bm z_0)+\sigma_z^2\left[g\tr[\nabla^2g(\bm z_0)\bm\Sigma_z]+\dfrac{2g\tr[\nabla \bar\pi(\bm z_0)^T\bm\Sigma_z\nabla g(\bm z_0)]}{\bar\pi(\bm z_0)}\right]+o(\sigma_z^2),$$
hence
\begin{align*}
\mathbb E_{\theta_0}\left[\left(\mathbb E_\pi[g(\bm \eta)|\bm z_0]\right)^2\right]&=\int \left(\mathbb E_\pi[g(\bm \eta)|\bm z_0]\right)^2p(\bm z_0|\bm \eta_0)d\bm z_0\\
&=\int \left(g^2(\bm z_0)+\sigma_z^2\left[g\tr[\nabla^2g(\bm z_0)\bm\Sigma_z]+\dfrac{2g\tr[\nabla \bar\pi(\bm z_0)^T\bm\Sigma_z\nabla g(\bm z_0)]}{\bar\pi(\bm z_0)}\right]+o(\sigma_z^2)\right)p(\bm z_0|\bm \eta_0)d\bm z_0\\
&=g^2(\bm \eta_0)+\sigma_z^2\left[g\tr[\nabla^2g(\bm \eta_0)\bm\Sigma_z]+\dfrac{2g\tr[\nabla \bar\pi(\bm \eta_0)^T\bm\Sigma_z\nabla g(\bm \eta_0)]}{\bar\pi(\bm \eta_0)}\right]+\dfrac{1}{2}\sigma_z^2\tr[\nabla^2(g^2)(\bm \eta_0)\bm\Sigma_z]+o(\sigma_z^2)\\
&=g^2(\bm \eta_0)+\sigma_z^2\left[2g\tr[\nabla^2g(\bm \eta_0)\bm\Sigma_z]+\dfrac{2g\tr[\nabla \bar\pi(\bm \eta_0)^T\bm\Sigma_z\nabla g(\bm \eta_0)]}{\bar\pi(\bm \eta_0)}+\tr[\nabla g(\bm \eta_0)^T\bm\Sigma_z\nabla g(\bm \eta_0)]\right]+o(\sigma_z^2).
\end{align*}
Therefore, the variance is
\begin{align*}
\text{Var}_{\theta_0}[\mathbb E_\pi[g(\bm \eta)|\bm z_0]]&=\mathbb E_{\theta_0}\left[\left(\mathbb E_\pi[g(\bm \eta)|\bm z_0]\right)^2\right]-\left[\mathbb E_{\theta_0}[\mathbb E_\pi[g(\bm \eta)|\bm z_0]|\bm \eta_0]\right]^2\\
&=\sigma_z^2\nabla g(\bm \eta_0)^T\bm\Sigma_z\nabla g(\bm \eta_0)+o(\sigma_z^2)\\
&\asymp \sigma_z^2.
\end{align*}

\section*{Proof of Theorem \ref{thm: overall-risk}}\label{pf:overall-risk}
From Theorem \ref{thm:bias-variance}, we have
$$\mathbb E_{\theta_0}[(\hat\theta_0^{(c)}-\theta_0)^2] = B_0^2+2B_0O_p(b^2)+O_p(b^4)+V_0+O_p\left(\dfrac{1}{Kd^b}\right).$$
On the other hand, 
\begin{align*}
\mathbb E[B_0] &= \mathbb E[\mathbb E_{\theta_0}[\mathbb E_\pi[g(\bm \eta)|\bm z_0]]-g(\bm \eta_0)]=\mathbb E[g(\bm \eta)]-\mathbb E[g(\bm \eta_0)]=0,\\
\mathbb E[B_0^2+V_0] &=Var[B_0]+\mathbb E[V_0]=Var[\mathbb E_{\theta_0}[\mathbb E_\pi[g(\bm \eta)|\bm z_0]]-\theta_0]+\mathbb E[Var_{\theta_0}[E_\pi[g(\bm \eta)|\bm z_0]]]\\
&=Var[\mathbb E_{\theta_0}[\mathbb E_\pi[g(\bm \eta)|\bm z_0]-\theta_0]]+\mathbb E[Var_{\theta_0}[E_\pi[g(\bm \eta)|\bm z_0]-\theta_0]]\\
&=Var[\mathbb E_\pi[\theta_0|\bm z_0]-\theta_0].
\end{align*}
Therefore,
$$\mathbb E[(\hat\theta_0^{(c)}-\theta_0)^2]=\mathbb E[\mathbb E_{\theta_0}[(\hat\theta_0^{(c)}-\theta_0)^2]]=Var[\mathbb E_\pi[\theta_0|\bm z_0]-\theta_0]+O_p(b^4)+O_p\left(\dfrac{1}{Kd^b}\right).$$

\section*{Proof of Theorem \ref{thm: theta-theta-only}}\label{pf:theta-theta-only}
The combined estimator can be written as
$$\hat\theta_0^{(c)}=\dfrac{\sum_{k=0}^Kw(\hat\theta_k,\hat\theta_0)\hat\theta_k}{\sum_{k=0}^Kw(\hat\theta_k,\hat\theta_0)}=\dfrac{\dfrac{1}{K+1}\sum_{k=0}^Kw(\hat\theta_k,\hat\theta_0)\hat\theta_k}{\dfrac{1}{K+1}\sum_{k=0}^Kw(\hat\theta_k,\hat\theta_0)}.$$
Let
$$q(\hat\theta)=\int p(\hat\theta)\pi(\theta)d\theta.$$
By law of large number, when $K\rightarrow \infty$, the numerator is
\begin{align*}
\dfrac{1}{K+1}\sum_{k=0}^Kw(\hat\theta_k,\hat\theta_0)\hat\theta_k&\xrightarrow{\enskip P\enskip} \mathbb E[w(\hat\theta,\hat\theta_0)\hat\theta]\\
&=\int\left(\dfrac{1}{q(\hat\theta)q(\hat\theta_0)}\int p(\hat\theta|\theta')p(\hat\theta_0|\theta')\pi(\theta')d\theta'\right)\hat\theta q(\hat\theta)d\hat\theta\\
&=\dfrac{1}{q(\hat\theta_0)}\int\left(\int g(\hat\theta)p(\hat\theta|\theta')d\hat\theta\right)p(\hat\theta_0|\theta')\pi(\theta')d\theta'\\
&=\dfrac{1}{q(\hat\theta_0)}\int\theta'p(\hat\theta_0|\theta')\pi(\theta')d\theta'\\
&=\int \theta'\pi(\theta'|\hat\theta_0)d\theta'.
\end{align*}
Similarly, for the denominator, we have
\begin{align*}
\dfrac{1}{K+1}\sum_{k=0}^Kw(\hat\theta_k,\hat\theta_0)&\xrightarrow{\enskip P\enskip} \mathbb E[w(\hat\theta,\hat\theta_0)]\\
&=\int\left(\dfrac{1}{q(\hat\theta)q(\hat\theta_0)}\int p(\hat\theta|\theta')p(\hat\theta_0|\theta')\pi(\theta')d\theta'\right)q(\hat\theta)d\hat\theta\\
&=\dfrac{1}{q(\hat\theta_0)}\int\left(\int p(\hat\theta|\theta')d\hat\theta\right)p(\hat\theta_0|\theta')\pi(\theta')d\theta'\\
&=\dfrac{1}{q(\hat\theta_0)}\int p(\hat\theta_0|\theta')\pi(\theta')d\theta'\\
&=1.
\end{align*}
Hence, the combined estimator would converge in probability to the Bayes estimator with squared loss. 
On one hand, by central limit theorem, the numerator has asymptotic normality, provided finite second moment. On the other hand, the denominator converges to 1 in probability. By Slutsky's theorem, the ratio is also asymptotically normal with the same rate as central limit theorem. Therefore,
$$\sqrt{K}(\hat\theta_0^{(c)} - \mathbb E[\theta|\hat\theta_0])=O_p(1).$$

\section*{Proof of Theorem \ref{thm:m-theta-only}}\label{pf:m-theta-only}
When $K\rightarrow \infty$, the target function in optimization is now
\begin{align*}
\dfrac{1}{K+1}\sum_{k=0}^Kw(\hat\theta_k,\hat\theta_0)f(\theta,\hat\theta_k)&\xrightarrow{\enskip P\enskip}\int \dfrac{1}{q(\hat\theta)q(\hat\theta_0)}\left(\int p(\hat\theta|\theta')p(\hat\theta_0|\theta')\pi(\theta')d\theta'\right)f(\theta,\hat\theta)q(\hat\theta)d\hat\theta\\
&=\dfrac{1}{q(\hat\theta_0)}\int\left(\int f(\theta,\hat\theta)p(\hat\theta|\theta')d\hat\theta\right)p(\hat\theta_0|\theta')\pi(\theta')d\theta'\\
&=\dfrac{1}{q(\hat\theta_0)}\int L(\theta,\theta')p(\hat\theta_0|\theta')\pi(\theta')d\theta' + \dfrac{1}{q(\hat\theta_0)}\int C(\theta')p(\hat\theta_0|\theta')\pi(\theta')d\theta'.
\end{align*}
The second component here is a constant with respect to $\theta$. Given the assumptions on $M(\theta, \hat\theta)$ and following the proof in Appendix \ref{pf:mest}, we have 
$$\argmin_\theta \sum_{k=0}^Kw(\hat\theta_k,\hat\theta_0)M(\theta,\hat\theta_k)\xrightarrow{\enskip P\enskip}\argmin_\theta \int L(\theta,\theta')p(\hat\theta_0|\theta')\pi(\theta')d\theta' =\Theta_0.$$
Here, we simply denote the target estimator $\Theta_0(\bm x_0; L)$ as $\Theta_0$.
Let $M'_\theta(\theta, \hat\theta) = \dfrac{\partial M(\theta, \hat\theta)}{\partial \theta}$, $M''_\theta(\theta, \hat\theta) = \dfrac{\partial^2 M(\theta, \hat\theta)}{\partial \theta^2}$ and $\theta_K^* = \argmin_\theta \sum_{k=0}^Kw(\hat\theta_k,\hat\theta_0)M(\theta,\hat\theta_k)$. Then we have
$$\sum_{k=0}^Kw(\hat\theta_k,\hat\theta_0)M'_\theta(\theta_K^*,\hat\theta_k)=0.$$
Since $\theta_K^*$ converges to $\Theta_0$, it's reasonable to expand the equation at $\Theta_0$.
$$\sum_{k=0}^Kw(\hat\theta_k,\hat\theta_0)M'_\theta(\Theta_0,\hat\theta_k)+ (\theta_K^* - \Theta_0)\sum_{k=0}^Kw(\hat\theta_k,\hat\theta_0)M''_\theta(\Theta_0,\hat\theta_k) + O_p((\theta_K^*-\Theta_0)^2)=0.$$
Then
$$\theta_K^* - \Theta_0 = -\dfrac{\sum_{k=0}^Kw(\hat\theta_k,\hat\theta_0)f'_\theta(\Theta_0,\hat\theta_k)}{\sum_{k=0}^Kw(\hat\theta_k,\hat\theta_0)f''_\theta(\Theta_0,\hat\theta_k)+O_p(\theta_K^*-\Theta_0)}.$$
Given the numerator has a finite variance, by central limit theorem and Slutsky's theorem, it is immediate that
$$\sqrt{K}(\theta_K^* - \Theta_0) = O_p(1).$$

\section*{Proof of Theorem \ref{thm: intrinsic_asymp_theta_only}}
Similar to the proof of Theorem \ref{thm: intrinsic_asymp}, when $\sigma_\theta^2\rightarrow 0$, we have
\begin{align*}
    \Theta_0[\bm x_0;\ell_2]&=\mathbb E_\pi[\theta_0\mid\hat\theta_0]\\
    &=\dfrac{\int \theta_0 p(\hat\theta_0\mid\theta_0)\pi(\theta_0)d\theta_0}{\int \theta_0 p(\hat\theta_0\mid\theta_0)\pi(\theta_0)d\theta_0}\\
    &=\dfrac{\hat\theta_0\pi(\hat\theta_0) + \dfrac{1}{2}\sigma_\theta^2(\hat\theta_0\pi(\hat\theta_0))''+o_p(\sigma_\theta^2)}{\pi(\hat\theta_0) + \dfrac{1}{2}\sigma_\theta^2(\pi(\hat\theta_0))''+o_p(\sigma_\theta^2)}\\
    &=\hat\theta_0 + \dfrac{1}{2}\sigma_\theta^2\left(\dfrac{(\hat\theta_0\pi(\hat\theta_0))''}{\pi(\hat\theta_0)}-\dfrac{\hat\theta_0(\pi(\hat\theta_0))''}{\pi(\hat\theta_0)}\right) + o_p(\sigma_\theta^2)\\
    &=\hat\theta_0 + \sigma_\theta^2\dfrac{\pi'(\hat\theta_0)}{\pi(\hat\theta_0)}+ o_p(\sigma_\theta^2).
\end{align*}
Therefore, for any fixed $\theta_0$
\begin{align*}
    \mathbb E_{\theta_0}[\Theta_0[\bm x_0;\ell_2]] &= \int \left(\hat\theta_0 + \sigma_\theta^2\dfrac{\pi'(\hat\theta_0)}{\pi(\hat\theta_0)}+ o_p(\sigma_\theta^2)\right)p(\hat\theta_0\mid\theta_0)d\hat\theta_0\\
    &=\theta_0 + \sigma_\theta^2\dfrac{\pi'(\theta_0)}{\pi(\theta_0)} + \dfrac{1}{2}\sigma_\theta^2\left(\theta_0 + \sigma_\theta^2\dfrac{\pi'(\theta_0)}{\pi(\theta_0)}\right)'' + o_p(\sigma_\theta^2)\\
    &=\theta_0 + \sigma_\theta^2\dfrac{\pi'(\theta_0)}{\pi(\theta_0)} + o_p(\sigma_\theta^2),
\end{align*}
and similarly,
$$\mathbb E_{\theta_0}[\Theta_0[\bm x_0;\ell_2]^2]=\theta_0^2 + 2\sigma_\theta^2\dfrac{\theta_0\pi'(\theta_0)}{\pi(\theta_0)} + \sigma_\theta^2 + o_p(\sigma_\theta^2).$$
Hence, the bias is
$$B_0(\theta_0) = \mathbb E_{\theta_0}[\Theta_0[\bm x_0;\ell_2]]-\theta_0 = \sigma_\theta^2\dfrac{\pi'(\theta_0)}{\pi(\theta_0)} + o_p(\sigma_\theta^2)\asymp \sigma_\theta^2,$$
and the variance is
$$V_0(\theta_0) = \mathbb E_{\theta_0}[\Theta_0[\bm x_0;\ell_2]^2]- \mathbb E_{\theta_0}[\Theta_0[\bm x_0;\ell_2]]^2 = \sigma_\theta^2 + o_p(\sigma_\theta^2)\asymp \sigma_\theta^2.$$

\section*{Proof of Theorem \ref{thm:theta-complete}}\label{pf:theta-complete}
The iGroup estimator is
$$\hat\theta^{(c)}_0 = \dfrac{\dfrac{1}{K+1}\sum_{k=0}^Kw(\hat\theta_k, \bm z_k;\hat\theta_0, 
\bm z_0)\hat\theta_k}{\dfrac{1}{K+1}\sum_{k=0}^Kw(\hat\theta_k, \bm z_k;\hat\theta_0, 
\bm z_0)}.$$
When $K\rightarrow \infty$, the numerator converges to
\begin{align*}
\dfrac{1}{K+1}\sum_{k=0}^Kw(\hat\theta_k, \bm z_k;\hat\theta_0, 
\bm z_0)\hat\theta_k &\xrightarrow{\enskip P \enskip}\mathbb E[w(\hat\theta, \bm z;\hat\theta_0, 
\bm z_0)\hat\theta]\\
&=\iint \mathcal K\left(\dfrac{\|\bm z-\bm z_0\|}{b}\right)\dfrac{\int p(\hat\theta|\theta)p(\hat\theta_0|\theta)p(\theta|\bm z_0)d\theta}{p(\hat\theta|\bm z)p(\hat\theta_0|\bm z_0)}\hat\theta p(\hat\theta, \bm z)d\hat\theta d\bm z\\
&=\dfrac{1}{p(\hat\theta_0|\bm z_0)}\left(\iint p(\hat\theta|\theta)p(\hat\theta_0|\theta)p(\theta|\bm z_0)\hat\theta d\theta d\hat\theta\right)\left(\int \mathcal K\left(\dfrac{\|\bm z-\bm z_0\|}{b}\right) p(\bm z)d\bm z\right)\\
&\xrightarrow{\enskip P\enskip}\dfrac{p(\bm z_0)}{p(\hat\theta_0|\bm z_0)}\iint p(\hat\theta|\theta)p(\hat\theta_0|\theta)p(\theta|\bm z_0)g(\hat\theta)d\theta d\hat\theta\\
&= \dfrac{p(\bm z_0)}{p(\hat\theta_0|\bm z_0)} \int\left(\int p(\hat\theta|\theta)g(\hat\theta)d\hat\theta\right)p(\hat\theta_0|\theta)p(\theta|\bm z_0)d\theta\\
&=\dfrac{p(\bm z_0)}{p(\hat\theta_0|\bm z_0)} \int\theta p(\hat\theta_0|\theta)p(\theta|\bm z_0)d\theta\\
&=p(\bm z_0)\int \theta p(\theta|\hat\theta_0, \bm z_0)d\theta.
\end{align*}
Similarly for the denominator, we have
\begin{align*}
\dfrac{1}{K+1}\sum_{k=0}^Kw(\hat\theta_k, \bm z_k;\hat\theta_0, 
\bm z_0)\xrightarrow{\enskip P \enskip}\mathbb E[w(\hat\theta, \bm z;\hat\theta_0, 
\bm z_0)]
\xrightarrow{\enskip P\enskip} p(\bm z_0).
\end{align*}
Therefore, the ratio converges to the target estimator $\Theta_0(\bm x_0, \bm z_0;\ell_2) = \mathbb E_\pi[\theta_0|\hat\theta_0, \bm z_0]$. Moreover, by central limit theorem, given bandwidth $b$, the numerator has an error of order $1/\sqrt{K}$:
$$\dfrac{1}{K+1}\sum_{k=0}^Kw(\hat\theta_k, \bm z_k;\hat\theta_0, 
\bm z_0)\hat\theta_k - \mathbb E[w(\hat\theta, \bm z;\hat\theta_0, 
\bm z_0)\hat\theta]=O_p\left(K^{-1/2}\right).$$
It brings a zero bias bias and a $O_p(1/K)$ variance. 
Now consider the kernel smoothing part, which yields a bias of order $b^2$ and a variance of order $1/(Kb^d)$. Therefore, the overall bias is of order $b^2$ and the overall variance is of order $O_p(K^{-1})+O_p(1/(Kb^d))=O_p(1/(Kb^d))$. Both the bias and variance is of the same order as in a $d$-dimensional kernel smoothing estimator. Hence, the optimal choice of the bandwidth is $\hat b \asymp K^{1/(d+4)}$, under which the optimal mean squared error is $O_p(K^{-4/d+4})$.

% which gives an optimal choice of bandwidth as $\hat b \asymp K^{1/(d+4)}$. Therefore, the overall mean squared error is of order $O_p(K^{-1}) + O_p(K^{-4/d+4}) = O_p(K^{-4/d+4})$.

\section*{Proof of Theorem \ref{thm:m-complete}}\label{pf:m-complete}
When $K\rightarrow\infty$, the combined objective function is
\begin{align*}
\dfrac{1}{K+1}\sum_{k=0}^Kw(\hat\theta_k,\bm z_k;\hat\theta_0,\bm z_0)M(\theta,\hat\theta_k)&\xrightarrow{\enskip P\enskip}\iint \mathcal K\left(\dfrac{\|\bm z-\bm z_0\|}{b}\right)\dfrac{\int p(\hat\theta|\theta')p(\hat\theta_0|\theta')p(\theta'|\bm z_0)d\theta'}{p(\hat\theta|\bm z)p(\hat\theta_0|\bm z_0)}M(\theta, \hat\theta)p(\hat\theta, \bm z)d\hat\theta d\bm z\\
&=\dfrac{1}{p(\hat\theta_0|\bm z_0)}\left(\iint p(\hat\theta|\theta')p(\hat\theta_0|\theta')p(\theta'|\bm z_0)M(\theta, \hat\theta)d\theta' d\hat\theta\right)\left(\int \mathcal K\left(\dfrac{\|\bm z-\bm z_0\|}{b}\right) p(\bm z)d\bm z\right)\\
&\xrightarrow{\enskip P\enskip}\dfrac{p(\bm z_0)}{p(\hat\theta_0|\bm z_0)}\iint p(\hat\theta|\theta')p(\hat\theta_0|\theta')p(\theta'|\bm z_0)M(\theta, \hat\theta)d\theta' d\hat\theta\\
&=\dfrac{p(\bm z_0)}{p(\hat\theta_0|\bm z_0)} \int\left(\int p(\hat\theta|\theta')M(\theta, \hat\theta)d\hat\theta\right)p(\hat\theta_0|\theta')p(\theta'|\bm z_0)d\theta'\\
&=\dfrac{p(\bm z_0)}{p(\hat\theta_0|\bm z_0)} \int(L(\theta,\theta')+C(\theta'))p(\hat\theta_0|\theta')p(\theta'|\bm z_0)d\theta'\\
&=p(\bm z_0)\int L(\theta,\theta')p(\theta'|\hat\theta_0,\bm z_0)d\theta' +  p(\bm z_0)\int C(\theta')p(\theta'|\hat\theta_0,\bm z_0)d\theta'.
\end{align*}
The second term here is a constant with respect to $\theta$. Given the convex and second-order differentiable condition of $M(\theta, \hat\theta)$, following the proof in Appendix \ref{pf:mest}, the iGroup estimator converges to the target estimator $\Theta_0(\bm x_0, \bm z_0; L)$ in probability. Given the consistency, one can expand the term at $\Theta_0$ as in Appendix \ref{pf:m-theta-only} (proof of Theorem \ref{thm:m-theta-only}) except that the weight is replaced by the full weight $w(\hat\theta, \bm z;\hat\theta_0, \bm z_0)$. By following the same argument in Appendix \ref{pf:theta-complete}, the numerator has an asymptotic mean squared error of order $K^{-4/(d+4)}$ when the bandwidth is chosen to be optimal $\hat b\asymp K^{-1/(d+4)}$. Provided the denominator converges in probability to its expectation by law of large number, we have $\tilde \theta_0^{(c)}$ has a mean squared error of order $K^{-4/(d+4)}$.

\section*{Proof of Theorem \ref{thm: intrinsic_asymp_complete_case}}
For fixed $\sigma_z^2$, the result follows immediately from the proof of Theorem \ref{thm: intrinsic_asymp_theta_only} except that $\Theta_0[\hat\theta_0;\ell2] = \int \theta_0 p(\hat\theta_0\mid\theta_0)\pi(\theta_0)d\theta_0 / \int p(\hat\theta_0\mid\theta_0)\pi(\theta_0)d\theta_0$ is replaced by $\Theta_0[\hat\theta_0, \bm z_0;\ell2] = \int \theta_0 p(\hat\theta_0\mid\theta_0)p(\bm z_0\mid \theta_0)\pi(\theta_0)d\theta_0 / \int p(\hat\theta_0\mid\theta_0)p(\bm z_0\mid \theta_0)\pi(\theta_0)d\theta_0$. For fixed $\sigma_\theta^2$, the result follows from the same proof in Theorem \ref{thm: intrinsic_asymp}.

\section*{Proof of Proposition \ref{thm: risk-decomposition-simple-l2}}
Consider the problem based on both information sets $\mathcal D_x$ and $\mathcal D_z$ and use the notation of the target estimator $\Theta_0 = \mathbb E_\pi[\theta_0\mid \bm x_0, \bm z_0]$. Notice that 
$$\delta_0(\mathcal D_x, \mathcal D_z) -\theta_0 = (\delta_0(\mathcal D_x, \mathcal D_z)-\Theta_0) + (\Theta_0 - \theta_0).$$
Given any fixed $(\bm x_0, \bm z_0)$, the first term $\delta_0(\mathcal D_x, \mathcal D_z)-\Theta_0$ depends on other individuals' observations $(\bm x_1, \dots, \bm x_K,\allowbreak \bm z_1,\dots, \bm z_K)$, while the second term $\Theta_0 - \theta_0$ depends on the true parameter $\theta_0$, which is treated as random. Therefore, these two terms are independent conditioned on $(\bm x_0, \bm z_0)$, and we have
\begin{align*}
    \mathbb E[(\delta_0(\mathcal D_x, \mathcal D_z)-\theta_0)^2\mid \bm x_0, \bm z_0] =&  \mathbb E[(\delta_0(\mathcal D_x, \mathcal D_z)-\Theta_0)^2\mid \bm x_0, \bm z_0] + \mathbb E[(\Theta_0-\theta_0)^2\mid \bm x_0, \bm z_0] \\
    & + 2\mathbb E[\delta_0(\mathcal D_x, \mathcal D_z)-\Theta_0\mid \bm x_0, \bm z_0]
    \mathbb E[\Theta_0-\theta_0\mid \bm x_0, \bm z_0].
\end{align*}
The last term $\mathbb E[\Theta_0-\theta_0\mid \bm x_0, \bm z_0]$ is zero. By taking expectation over $\bm x_0$ and $\bm z_0$, the decomposition is proved. Similar procedure for information set $\mathcal D_x$ or $\mathcal D_z$.

\section*{Proof of Proposition \ref{thm: risk-decomposition-simple-L}}
Consider the problem based on both information sets $\mathcal D_x$ and $\mathcal D_z$. We expand the loss function at $\hat\theta = \Theta_0$ such that
\begin{equation}
L(\delta_0, \theta_0) = L(\Theta_0, \theta_0) + (\delta_0 - \Theta_0)L'(\Theta_0, \theta_0) + \dfrac{1}{2}(\delta_0 - \Theta_0)^2L''(\Theta_0, \theta_0) + o((\delta_0 - \Theta_0)^2). \label{eq: loss-function-decompose}
\end{equation}
Notice that
$$\mathbb E_\pi[(\delta_0 - \Theta_0)L'(\Theta_0, \theta_0)\mid \bm x_0, \bm z_0] = \mathbb E_\pi[(\delta_0 - \Theta_0)\mid \bm x_0, \bm z_0]\mathbb E_\pi[L'(\Theta_0, \theta_0)\mid \bm x_0, \bm z_0]=0.$$
The first equality is because for fixed $\bm x_0$ and $\bm z_0$, $\delta_0 - \Theta_0$ depends on other individuals' observations and $L'(\Theta_0, \theta_0)$ depends on the value of $\theta_0$. The second equality is because $\Theta_0$ is the minimizer of $\mathbb E_\pi[L(\Theta_0, \theta_0)\mid \bm x_0, \bm z_0]$.
Hence, by taking expectation of Equation (\ref{eq: loss-function-decompose}), we have the desired decomposition. Similar procedure for information set $\mathcal D_x$ or $\mathcal D_z$.

\section*{Proof of Proposition \ref{thm:decomposition}}{\label{pf:decomposition}}
We first calculate the mean squared error of $\hat\theta_0^{(c)}$ conditioned on $(\hat\theta_0, \bm z_0)$. Notice that 
$$\hat\theta_0^{(c)}-\theta_0 = (\hat\theta_0^{(c)}-\Theta_0) + (\Theta_0 - \theta_0).$$
Given any fixed $(\hat\theta_0, \bm z_0)$, the first term $\hat\theta_0^{(c)}-\Theta_0$ is a function of other individuals' observations $(\hat\theta_1,\dots, \hat\theta_K,\bm z_1,\dots, \bm z_K)$ , while the second term $\Theta_0 - \theta_0$ is a function of the true parameter $\theta_0$, which is treated as random. Therefore, these two terms are independent conditioned on $(\hat\theta_0, \bm z_0)$, and we have
\begin{align*}
    \mathbb E[(\hat\theta_0^{(c)}-\theta_0)^2\mid \hat\theta_0, \bm z_0] &=  \mathbb E[(\hat\theta_0^{(c)}-\Theta_0)^2\mid \hat\theta_0, \bm z_0] + \mathbb E[(\Theta_0-\theta_0)^2\mid \hat\theta_0, \bm z_0] + 
    2\mathbb E[\hat\theta_0^{(c)}-\Theta_0\mid \hat\theta_0, \bm z_0]
    \mathbb E[\Theta_0-\theta_0\mid \hat\theta_0, \bm z_0].
\end{align*}
Furthermore, since $\Theta_0$ is a function of $\hat\theta_0$ and $\bm z_0$, we have
\begin{align*}
\mathbb E[(\Theta_0-\theta_0)^2\mid \hat\theta_0, \bm z_0]&=\Theta_0^2 - 2\Theta_0\mathbb E[\theta_0\mid \hat\theta_0, \bm z_0] + \mathbb E[\theta_0^2\mid \hat\theta_0, \bm z_0]\\
&=(\Theta_0 - \mathbb E[\theta_0\mid \hat\theta_0, \bm z_0])^2+\mathbb E[\theta_0^2\mid \hat\theta_0, \bm z_0] - \left(\mathbb E[\theta_0\mid \hat\theta_0, \bm z_0]\right)^2\\
&=(\Theta_0 - \mathbb E[\theta_0\mid \hat\theta_0, \bm z_0])^2 + \mathbb E[(\mathbb E[\theta\mid\hat\theta_0, \bm z_0]-\theta_0)^2\mid \hat\theta_0, \bm z_0].
\end{align*}
Hence, the conditional mean squared error of $\hat\theta_0^{(c)}$ becomes
\begin{align*}
    \mathbb E[(\hat\theta_0^{(c)}-\theta_0)^2\mid \hat\theta_0, \bm z_0]
    =&\mathbb E[(\hat\theta_0^{(c)}-\Theta_0)^2\mid \hat\theta_0, \bm z_0] + (\Theta_0 - \mathbb E[\theta_0\mid \hat\theta_0, \bm z_0])^2 + \mathbb E[(\mathbb E[\theta_0\mid\hat\theta_0, \bm z_0]-\theta_0)^2\mid \hat\theta_0, \bm z_0] \\
    &+ 2\mathbb E[\hat\theta_0^{(c)}-\Theta_0\mid \hat\theta_0, \bm z_0]
    \mathbb E[\Theta_0-\theta_0\mid \hat\theta_0, \bm z_0].
\end{align*}
By taking the expectation for $\hat\theta_0$ and $\bm z_0$ on both sides, we have
\begin{align}
    \mathbb E[(\hat\theta_0^{(c)}-\theta_0)^2]
    =&\mathbb E[(\hat\theta_0^{(c)}-\Theta_0)^2] + \mathbb E(\Theta_0 - \mathbb E[\theta_0\mid \hat\theta_0, \bm z_0])^2 + \mathbb E[(\mathbb E[\theta_0\mid\hat\theta_0, \bm z_0]-\theta_0)^2]\nonumber \\
    &+ 2\mathbb E\left\{\mathbb E[\hat\theta_0^{(c)}-\Theta_0\mid \hat\theta_0, \bm z_0]
    \mathbb E_{\hat\theta_0, \bm z_0}[\Theta_0-\theta_0\mid \hat\theta_0, \bm z_0]\right\}\nonumber\\
    =&R_{np}(\hat\theta_0^{(c)}) + R_{inf}(\hat\theta_0^{(c)}) + R_0
    + 2\mathbb E\left\{\mathbb E[\hat\theta_0^{(c)}-\Theta_0\mid \hat\theta_0, \bm z_0]
    \mathbb E[\Theta_0-\theta_0\mid \hat\theta_0, \bm z_0]\right\}\label{eq:decompose-last-term}.
\end{align}
The only thing left is to show the last term is 0. When $\Theta_0 =\Theta_0(\bm x_0, \bm z_0;\ell_2)= \mathbb E_\pi[\theta_0\mid\hat\theta_0, \bm z_0]$ as in Case 3 in Section \ref{sec:complete-case}, it is straightforward that 
$$\mathbb E[\Theta_0-\theta_0\mid \hat\theta_0, \bm z_0]=\Theta_0-\mathbb E_\pi[\theta_0\mid \hat\theta_0, \bm z_0]=0.$$
When $\Theta_0 = \Theta_0(\bm x_0;\ell_2) = \mathbb E_\pi[\theta_0\mid\hat\theta_0]$ as in Case 2 in Section \ref{sec:theta-only}, neither $\hat\theta_0^{(c)}$ nor $\Theta_0$ depends on $\bm z$, and we can prove it by taking expectation over $\bm z_0$ first as follows
\begin{align*}
\mathbb E\left\{\mathbb E[\hat\theta_0^{(c)}-\Theta_0\mid \hat\theta_0, \bm z_0]
    \mathbb E[\Theta_0-\theta_0\mid \hat\theta_0, \bm z_0]\right\}
    &=\mathbb E\left\{\mathbb E[\hat\theta_0^{(c)}-\Theta_0\mid \hat\theta_0]
    \mathbb E[\Theta_0-\theta_0\mid \hat\theta_0, \bm z_0]\right\}\\
    &=\mathbb E\left\{\mathbb E[\hat\theta_0^{(c)}-\Theta_0\mid \hat\theta_0]
    \mathbb E\left(\mathbb E[\Theta_0-\theta_0\mid \hat\theta_0, \bm z_0]\mid \hat\theta_0\right)\right\}\\
    &=\mathbb E\left\{\mathbb E[\hat\theta_0^{(c)}-\Theta_0\mid \hat\theta_0]
    \mathbb E\left(\Theta_0-\mathbb E[\theta_0\mid \hat\theta_0, \bm z_0]\mid \hat\theta_0\right)\right\}\\
    &=\mathbb E\left\{\mathbb E[\hat\theta_0^{(c)}-\Theta_0\mid \hat\theta_0]
    \left(\Theta_0-\mathbb E_\pi[\theta_0\mid \hat\theta_0]\right)\right\}\\
    &=0.
\end{align*}
Similarly, when $\Theta_0 =\Theta_0(\bm z_0;\ell_2)= \mathbb E_\pi[\theta_0\mid\bm z_0]$ as in Case 1 in Section \ref{sec:exogenous-variable}, it can be shown by taking expectation over $\hat\theta_0$ first. Therefore, for all cases we considered, the last term in (\ref{eq:decompose-last-term}) equals 0, and we have 
$$R(\hat\theta_0^{(c)}) = R_{np}(\hat\theta_0^{(c)}) + R_{inf}(\hat\theta_0^{(c)}) + R_0.$$

\section*{Proof of Proposition \ref{thm:m-decomposition}}\label{pf:m-decomposition}
Since $\tilde\theta_0^{(c)} \rightarrow \Theta_0(\hat\theta_0, \bm z_0)$ for all $\hat\theta_0$ and $\bm z_0$, the loss function can be expanded at $\Theta_0$ as follows
$$
L(\tilde\theta_0^{(c)}, \theta_0)=L(\Theta_0, \theta_0) + L'(\Theta_0, \theta_0)(\tilde\theta_0^{(c)}-\Theta_0) +\dfrac{1}{2}L''(\Theta_0, \theta_0)(\tilde\theta_0^{(c)}-\Theta_0)^2 + o_p((\tilde\theta_0^{(c)}-\Theta_0)^2).
$$
By taking expectation on both sides, we have
\begin{align}
    \mathbb E[L(\tilde\theta_0^{(c)}, \theta_0)] &= \mathbb E[L(\Theta_0, \theta_0)] + \dfrac{1}{2}\mathbb E[L''(\Theta_0, \theta_0)(\tilde\theta_0^{(c)}-\Theta_0)^2]+o(\mathbb E[(\tilde\theta_0^{(c)}-\Theta_0)^2]) + \mathbb E[L'(\Theta_0, \theta_0)(\tilde\theta_0^{(c)}-\Theta_0)]\nonumber\\
    &=(\tilde R_0 + \tilde R_{inf}(\tilde\theta_0^{(c)})) + \tilde R_{np}(\tilde\theta_0^{(c)}) + o(\mathbb E[(\tilde\theta_0^{(c)}-\Theta_0)^2]) + \mathbb E[L'(\Theta_0, \theta_0)(\tilde\theta_0^{(c)}-\Theta_0)].\label{eq:m-decompose-last-term}
\end{align}
It only needs to show the last term is 0. When in Case 3, $\Theta_0 = \Theta_0(\bm x_0, \bm z_0; L)$, and $L'(\Theta_0, \theta_0)$ and $(\tilde\theta_0^{(c)}-\Theta_0)$ are independent conditioned on $(\hat\theta_0, \bm z_0)$. Therefore,
$$\mathbb E[L'(\Theta_0, \theta_0)(\tilde\theta_0^{(c)}-\Theta_0)\mid \hat\theta_0, \bm z_0] = \mathbb E[L'(\Theta_0, \theta_0)\mid \hat\theta_0, \bm z_0]\mathbb E[\tilde\theta_0^{(c)}-\Theta_0\mid \hat\theta_0, \bm z_0].$$
The first term $\mathbb E[L'(\Theta_0, \theta_0)\mid \hat\theta_0, \bm z_0]$ equals 0 because $\Theta_0=\argmin_\theta \mathbb E_\pi[L(\theta, \theta_0)\mid \hat\theta_0, \bm z_0]$.
Similarly, in Case 1 and Case 2, the conditional expectations of $L'(\Theta_0, \theta_0)$ conditioned on $\bm z_0$ and $\hat\theta_0$ respectively are 0. Hence, the last term in (\ref{eq:m-decompose-last-term}) is always 0.

\section*{Proof of Proposition \ref{thm:lemma}}\label{pf:lemma}
Noticing that $\hat\theta_{(-k)}^{(c)}-\theta_k$ and $\hat\theta_{k}^{(c)}-\theta_k$ are independent with each other, we have
\begin{align*}
\mathbb E(\hat\theta_{(-k)}^{(c)}-\hat\theta_k)^2&=\mathbb E(\hat\theta_{(-k)}^{(c)}-\theta_k)^2 + \mathbb E(\hat\theta_k-\theta_k)^2 + 0\\
&=\mathbb E\left(\hat\theta^{(c)}_k-\theta_k + \dfrac{w(k;k)}{\sum_{l\neq k}w(l;k)}(\hat\theta_k^{(c)}-\hat\theta_k)\right)^2+\mathbb E(\hat\theta_k-\theta_k)^2\\
&=\mathbb E\left(\hat\theta_k^{(c)}-\theta_k\right)^2+\mathbb E(\hat\theta_k-\theta_k)^2 + O\left(\dfrac{1}{K}\right).
\end{align*}
Therefore, the expectation of cross validation error is
\begin{align*}
\mathbb E(CV_{\Omega_0}(b)) & = \dfrac{1}{|\Omega_0|}\sum_{k\in\Omega_0}\mathbb E(\hat\theta_{(-k)}^{(c)}-\hat\theta_k)^2 \\
&= \dfrac{1}{|\Omega_0|}\sum_{k\in\Omega_0}\left[\mathbb E\left(\hat\theta_k^{(c)}-\theta_k\right)^2+\mathbb E(\hat\theta_k-\theta_k)^2 + O\left(\dfrac{1}{K}\right)\right]\\
&=R_K(b) + \mathbb E_{\Omega_0}(\hat\theta-\theta)^2 + O\left(\dfrac{1}{K}\right),
\end{align*}
where the second term is averaging over all individuals in $\Omega_0$ and hence a constant term with respect to bandwidth.

\end{document}